\documentclass[11pt,a4paper]{article}
\usepackage[english]{babel}  
\usepackage[margin=1.25in]{geometry}
\RequirePackage[OT1]{fontenc}
\RequirePackage{graphicx}
\RequirePackage{amsthm,amssymb}
\RequirePackage[cmex10]{amsmath}
\RequirePackage[colorlinks,citecolor=blue,urlcolor=blue]{hyperref}
\usepackage[round,authoryear]{natbib}
\usepackage{graphicx,epstopdf}
\usepackage{listing}
\usepackage{amsfonts}
\usepackage{amsmath}		
\usepackage{bbm} 	
\usepackage{amsmath,bm}
\usepackage{booktabs}
\usepackage{xcolor}
\usepackage{tabu}
\usepackage{caption} 
\usepackage{verbatim}
\usepackage{multirow}
\usepackage{authblk}
\usepackage[hang]{footmisc}
\usepackage{easy-todo}

\usepackage{color}
\usepackage{array}
\definecolor{dkgreen}{rgb}{0,0.6,0}
\definecolor{gray}{rgb}{0.5,0.5,0.5}
\definecolor{mauve}{rgb}{0.58,0,0.82}

\newtheorem{theorem}{Theorem}[section]

\newtheorem{assumption}{Assumption}
\newtheorem{definition}{Definition}
\newtheorem{proposition}{Proposition}[section]
\newtheorem{corollary}{Corollary}[section]
\newtheorem{lemma}{Lemma}[section]
\theoremstyle{remark}
\newtheorem{remark}{Remark}[section]

\newcommand{\law}{\mathrm{P}^{(T)}}
\newcommand{\prob}{\mathbb{P}}

\newcommand{\wto}{\Rightarrow}
\newcommand\indp{\protect\mathpalette{\protect\independenT}{\perp}}
\def\independenT#1#2{\mathrel{\rlap{$#1#2$}\mkern2mu{#1#2}}}
\newcommand{\E}{\mathbb{E}}
\newcommand{\rE}{\mathrm{E}^{(T)}}
\newcommand{\rEf}{\mathrm{E}_f}
\newcommand{\rElim}{\mathbb{E}}
\newcommand{\rL}{\mathrm{L}}
\newcommand{\rH}{\mathrm{H}}
\newcommand{\rP}{\mathrm{P}}
\newcommand{\F}{\mathfrak{F}}
\newcommand{\var}{\operatorname{Var}}
\newcommand{\cov}{\operatorname{Cov}}
\newcommand{\rd}{\mathrm{d}} 
\newcommand{\SN}{\mathbb{N}} 
\newcommand{\SR}{\mathbb{R}} 

\newcommand{\indicator}{\mathbbm{1}}
\newcommand{\trans}{\prime}
\newcommand{\mi}{\mathcal{M}}
\newcommand{\s}{s}
\newcommand{\stat}{S}
\newcommand{\test}{\varphi}
\newcommand{\cv}{\kappa}

\newcommand{\siglevel}{\alpha}
\newcommand{\transfor}{\mathfrak{g}}
\newcommand{\transgroup}{\mathfrak{G}}
\newcommand{\diag}{{\rm diag}}
\newcommand{\standardb}{\lpinte^{\star}}
\newcommand{\standardc}{c^{\star}}
\newcommand{\betaform}{\delta}
\newcommand{\fpert}{h}     
\newcommand{\pcons}{\mu}      
\newcommand{\pinte}{\beta}    
\newcommand{\ppers}{\gamma}   
\newcommand{\lpinte}{b}    
\newcommand{\lppers}{c}    
\newcommand{\e}{\varepsilon}
\newcommand{\et}{\varepsilon_t}
\newcommand{\ey}{\varepsilon^y}
\newcommand{\ex}{\varepsilon^x}
\newcommand{\eyt}{\varepsilon^y_t}
\newcommand{\ext}{\varepsilon^x_t}

\newcommand{\score}{\ell}
\newcommand{\llr}{\mathrm{LLR}}
\newcommand{\lllr}{\mathcal{L}}
\newcommand{\cs}{\Delta}
\newcommand{\qt}{\mathcal{Q}}
\newcommand{\yt}{y_t}

\newcommand{\coveg}{\bm{\sigma}_{\varepsilon g}}

\newcommand{\J}{J}
\newcommand{\Jf}{J_f}

\newcommand{\Jfyb}{J_{f_{y}\fpert}}
\newcommand{\Jfxb}{J_{f_{x}\fpert}}
\newcommand{\Jfybk}{J_{f_{y}\fpert_k}}
\newcommand{\Jfxbk}{J_{f_{x}\fpert_k}}
\newcommand{\Jg}{J_g}
\newcommand{\Jfyfy}{J_{f_{yy}}}
\newcommand{\Jfyfx}{J_{f_{yx}}}
\newcommand{\Jfxfx}{J_{f_{xx}}}

\newcommand{\Jgy}{J_{g_y}}
\newcommand{\Jgx}{J_{g_x}}
\newcommand{\Jc}{J_p}
\newcommand{\Jcycy}{J_{p_{yy}}}
\newcommand{\Jcycx}{J_{p_{yx}}}
\newcommand{\Jcxcx}{J_{p_{xx}}}

\newcommand{\Wf}{W_{\ell_{f}}}
\newcommand{\Wfy}{W_{\ell_{f_y}}}
\newcommand{\Wfx}{W_{\ell_{f_x}}}
\newcommand{\We}{W_{\varepsilon}}

\newcommand{\Wb}{W_{\fpert}}

\newcommand{\Zf}{Z_{\ell_{f}}}
\newcommand{\Zfy}{Z_{\ell_{f_y}}}
\newcommand{\Zfx}{Z_{\ell_{f_x}}}
\newcommand{\Zg}{Z_{\ell_{g}}}

\newcommand{\Ze}{Z_{\varepsilon}}

\newcommand{\Zb}{Z_{\fpert}}

\newcommand{\Bfy}{B_{\ell_{f_y}}}
\newcommand{\Bfx}{B_{\ell_{f_x}}}
\newcommand{\Bf}{B_{\ell_{f}}}
\newcommand{\Be}{B_{\varepsilon}}

\newcommand{\Bb}{B^{\Wb}}

\newcommand{\Bg}{B_{\ell_{g}}}
\newcommand{\Bgy}{B_{\ell_{g_y}}}
\newcommand{\Bgx}{B_{\ell_{g_x}}}
\newcommand{\By}{B_{g_y}}
\newcommand{\BTg}{B_{\ell_{g}}^{(T)}}
\newcommand{\BTgy}{B_{\ell_{g_y}}^{(T)}}
\newcommand{\BTgx}{B_{\ell_{g_x}}^{(T)}}
\newcommand{\W}{W}
\newcommand{\Wg}{W_{\ell_{g}}}

\newcommand{\Wy}{W_{g_y}}

\newcommand{\WTf}{W^{(T)}_{\ell_{f}}}
\newcommand{\WTfy}{W^{(T)}_{\ell_{f_y}}}
\newcommand{\WTfx}{W^{(T)}_{\ell_{f_x}}}
\newcommand{\WTb}{W^{(T)}_{\fpert}}
\newcommand{\WTe}{W^{(T)}_{\varepsilon}}
\newcommand{\WTk}{W^{(T)}_{\fpert_{k}}}

\newcommand{\Wperp}{W_\perp}
\newcommand{\Bperp}{B_\perp}

\newcommand{\SequenceExperiment}{\mathcal{E}^{(T)}\left(f\right)}
\newcommand{\LimitExperiment}{\mathcal{E}\left(f\right)}
\newcommand{\LimitExperimentM}{\mathcal{E}_\mi\left(f\right)}
\newcommand{\opone}{o_\rP(1)}
\newcommand{\refcorrm}{\mathbf{R}_g}
\newcommand{\refcorrp}{\rho_g}

\graphicspath {{figures/}}

\title{Semiparametric Testing with Highly Persistent Predictors\thanks{We thank Gaia Becheri for significant input on an earlier version of this paper. We also thank Peter Boswijk, Feike Drost, Ramon van den Akker, two referees, the associate editor, and participants at the European Conferences of the Econometrics Community (EC2) conference, Amsterdam, Dec 2017; Aarhus University, Sep 2018 for helpful comments.}}
\author[1]{Bas J.M. Werker}
\author[2]{Bo Zhou}
\affil[1]{Econometrics and Finance Group, Tilburg University}
\affil[2]{Department of Economics and Finance, Durham University}
\date{}

\begin{document}

\maketitle

\abstract{\noindent We address the issue of semiparametric efficiency in the bivariate regression problem with a highly persistent predictor, where the joint distribution of the innovations is regarded an infinite-dimensional nuisance parameter. Using a structural representation of the limit experiment and exploiting invariance relationships therein, we construct invariant point-optimal tests for the regression coefficient of interest. This approach naturally leads to a family of feasible tests based on the component-wise ranks of the innovations that can gain considerable power relative to existing tests under non-Gaussian innovation distributions, while behaving equivalently under Gaussianity. When an i.i.d.\ assumption on the innovations is appropriate for the data at hand, our tests exploit the efficiency gains possible. Moreover, we show by simulation that our test remains well behaved under some forms of conditional heteroskedasticity.



\clearpage
\section{Introduction}\label{sec:Introduction}
\noindent Over the past two decades, inference for the bivariate regression model with a highly persistent predictor has been well studied under the assumption of bivariate Gaussian innovations. Several procedures have been proposed in the econometric literature, see  \citet{CavanaghElliotStock1995}, \citet{CampbellYogo2005}, \citet{JanssonMoreira2006}, \citet{EMW2015}, and \citet{MoreiraMourao2017}. These inference procedures are all constructed based on the assumption of Gaussian innovations and, while their validity has been established under weaker assumptions, the asymptotic power of all these procedures cannot go beyond the Gaussian power envelope.

In the present paper we show that, when the application supports an additional assumption of serially independent innovations, sizable power gains are possible beyond the Gaussian power envelope. We establish this result by studying in detail the invariance structures that are present in the limiting experiment associated with the predictive regression model. This leads to a semiparametric power envelop which, under non-Gaussian innovation distributions, lies above the Gaussian power envelope. In that case, even without knowing the innovation distribution, our method dominates existing QMLE-based methods.

Our results precisely quantify the statistical efficiency gains from non-Gaussian innovation distributions when innovations are serially independent in predictive regression models. Under such, arguably restrictive assumption, we construct semiparametrically optimal (in a sense to be made precise later) tests. Whether in concrete applications the assumption of serially independence is warranted, is an empirical question. When it is, it can, as our results show, be exploited leading to sizable power gains (of, as Section~\ref{sec:MonteCarlo} shows, up to 30\% under Student-$t_3$ innovation distributions).
Symmetrically, to make an informed choice, we study the behavior of our test when the innovations are not i.i.d.\ but exhibit conditional heteroskedasticity as often found in (financial) applications. Section~\ref{sec:MCCH} shows that, for the deviations studied, our test still has desirable size and power properties.

We note that our conceptual ideas reach further. We could, for instance, allow for serial dependence along the lines of \citet{ZvdAW2016} where an AR-type model on the error is imposed. Conditional heterogeneity could formally be addressed along the lines of \citet{LingMcaleer2003} where a GARCH-type structure on the error is imposed; or following \citet{Boswijk2005} where the (potentially nonstationary) volatility is estimated nonparametrically. These relaxations would technically be non-trivial and are left for future research. Note that, in view of the robustness-efficiency trade-off \citep[see, e.g.,][]{Muller2011}, an i.i.d.\ assumption on the innovations ultimately driving the error term is not avoidable. Our test gives the empirical researchers an additional option: an improved power when innovations are i.i.d.\ and non-Gaussian. 

The study of (optimal) semiparametric inference in the predictive regression model is complicated by the nonstandard asymptotic behavior induced by the local-to-unity asymptotics on the persistence parameter. More precisely, the associated likelihood ratios are of the \textit{Locally Asymptotically Brownian Functional} (LABF) form in \citep[see][]{Jeganathan1995} and henceforth outside the conventional \textit{Locally Asymptotically Normality} (LAN) world. As a consequence, the usual semiparametric approach based on projecting the score of the parameter of interest on the tangent space of nuisance scores is not straightforward. In particular, the model does not feature an adaptiveness property, which complicates its analysis. \citet{Jansson2008} deals with the unit root testing problem, which also admits the LABF form, by guessing and then proving a least favorable direction of parametric submodels. An alternative approach has been proposed for the unit root testing problem in \citet{ZvdAW2016} and generalized to other common types of limiting experiments in \citet{Zhou2020}. In the present paper we apply these techniques to the predictive regression model.

The key idea is to exploit invariance structures in a so-called ``structural'' representation of the limit experiment. This approach sets us apart from most of the statistical and econometric literature where invariance arguments are used in the sequence of experiments. Instead, we obtain procedures which are invariant in the \textit{limit} experiment, thereby making the analysis tractable and applicable to many models. Furthermore, the unique bivariate nature of the predictive regression model leads to a nonstandard multivariate structure in the associated limit experiment (see Theorem~\ref{thm:StructuralLimitExperiment}). Therefore, we present the approach in detail in the present paper.

Our contribution is twofold. First, we derive the semiparametric power envelope for (asymptotically) invariant tests in case the predictor's persistence level is assumed to be known, based on the structural LABF limit experiments. More precisely, Girsanov's theorem, combined with the limiting likelihood ratios for LABF experiments, leads to a description of the limit experiment by stochastic differential equations (SDEs). The observations in the limit experiment correspond to the limits of partial-sum processes of the innovations and score functions in the predictive regression model. In this structural representation of the limit experiment, we find that the nuisance parameters induced by the density function of the innovations only appear in the drifts of the driving Brownian motions. This leads to an invariance restriction by taking the Brownian bridges (which are invariant with respect to these drifts) of these processes, and allows us to remove the nonparametric nuisance parameter (the density $f$ of the innovations). We show that this also generates the \textit{maximal invariant}. In this way, we avoid the problem of explicitly finding the least-favorable submodel. The likelihood of the maximal invariant immediately, by the Neyman-Pearson lemma, leads to the semiparametric power envelope. 

Second, we propose a family of semiparametric feasible tests that has desirable properties. These tests are constructed using (asymptotically) sufficient statistics that are based on the increments of innovations, their component-wise ranks, and a pair of chosen marginal reference densities for both innovations including a reference correlation parameter. The ranks appear naturally as rank-based partial-sum score processes which weakly converge to the Brownian bridge that is invariant w.r.t.\ the density perturbation parameters. To further eliminate the remaining nuisance parameter, namely the predictor's persistence level, we employ the \textit{Approximate Least Favorable Distribution} (ALFD) approach proposed by~\citet{EMW2015}. We also follow their suggestion to switch to standard asymptotic approximations when the persistence parameter is far from unity. This helps to control the size of our tests uniformly under both non-stationarity and stationarity, see Appendix~\ref{app:Switching}. The tests thus obtained are semiparametric in the sense that they have correct asymptotic sizes (under all innovation densities allowed) regardless of the choices of the marginal reference densities or the reference correlation.

Next to their uniform (relative to our model) validity, our test are more powerful than existing tests when the true innovation density is non-Gaussian. In particular, we compare our test to \citet{EMW2015} (henceforth denoted as EMW), which is based on Gaussian likelihood ratios \citep[see also][]{JanssonMoreira2006}. Our asymptotic analysis using invariance arguments shows that, under non-Gaussian innovations, the EMW test actually is measurable with respect to an invariant in the limit that is not \emph{maximally} invariant. As a result, under non-Gaussianity, we can construct tests that outperform the Gaussian power envelope and, thus, outperform the EMW test; see Remark~\ref{rem:GaussianInvariant}. The power improvement depends on the choices of the marginal reference densities: when they are ``closer'' to the true marginal densities, we gain more power (and, again, while always having the desired size). Additionally, if one fixes the marginal reference densities to be Gaussian, our test is generally still more powerful than the EMW test under non-Gaussian innovation density; while under Gaussian innovation density, our test performs equivalently to the EMW test. This property is often referred to as the Chernoff-Savage result (see \citet{ChernoffSavage1958}). In the present LABF setting we have not been able to formally prove this Chernoff-Savage result, but our simulations indicate that this property nevertheless may hold. 

Our rank-based test can be regarded a generalized version of quasi-likelihood ratio tests which take the reference density to be Gaussian. The extra freedom to choose the reference density also comes with the cost of actually choosing it. However, we note that, in line with traditional quasi-likelihood methods, one can always choose the Gaussian reference density. Based on the classical Chernoff-Savage result, we conjecture that our rank-based procedure will then always outperform the quasi-likelihood procedure. This is confirmed by simulations and intuition, but, as discussed below, given the non-standard limiting experiment structure, we have not been able to prove this formally. Alternative, one could study a plug-in estimator where the reference density is nonparametrically estimated. We do not study this formally in the present paper; however, see Section~\ref{sec:MCEstimatedG} for some simulation results. Similarly, one may envision an approach where one pre-tests the residuals for, e.g., high kurtosis and chooses a references density based on that pre-test result.

The paper is organized as follows. In Section~\ref{sec:Model}, we introduce the model and testing problem under consideration. In Section~\ref{sec:PowerEnvelope}, we develop the asymptotic power envelope for test that are (asymptotically) invariant with respect to the innovation density $f$, assuming the predictor's persistence parameter $\gamma$ is known. This development is based on the theory of limit experiments (see, e.g., \citet{LeCam1986} and \citet{vdVaart2000}) and a structural version for models of LABF likelihood ratios (see \citet{ZvdAW2016}). In particular, this section explains where our power gains come from, see Remark~\ref{rem:GaussianInvariant}. In Section~\ref{sec:FeasibleTests}, we employ the ALFD approach proposed by \citet{EMW2015}, among several available choices in the literature, to eliminate the nuisance parameter $\ppers$. In Section~\ref{sec:MonteCarlo}, we report large- and small-sample performances of our tests under both i.i.d.\ and conditional heteroskedastic errors. Section~\ref{sec:Conclusions} concludes. All proofs are gathered in the appendix.

\section{Model}\label{sec:Model}

\noindent Let $y_t$ denote a random variable, observable at time $t$, that we wish to predict at time $t-1$ using an observable explanatory variable $x_{t-1}$. We consider the predictive regression model
\begin{align}
\label{eqn:model_1}
y_t &= \pcons + \pinte  x_{t-1} + \eyt,\\
\label{eqn:model_2}
x_t-\alpha &= \ppers (x_{t-1}-\alpha) + \ext, 
\end{align}
with $x_0=0$.\footnote{Note that this assumption on the initial value $x_0$ could possibly be relaxed to the weaker assumption $T^{-1/2}x_0=\opone$ under $\pinte=0$ and $\ppers=1$. One can possibly proceed along the lines of \citet{MullerElliott2003}; see also a remark on this point in Section~4 of \citet{JanssonMoreira2006}. We keep the assumption $x_0=0$ for simplicity.} The parameter space is given by $\pcons\in\SR$, $\alpha\in\SR$, $\pinte\in\SR$, and $\ppers\in(-1,1]$. We have observations available for $t=1,\dots,T$.

Equation~(\ref{eqn:model_2}) features, along the lines of \citet{CavanaghElliotStock1995} and \citet{JanssonMoreira2006}, an intercept $\alpha$. However, as $\pcons$ is a nuisance parameter in our model, the intercept $\alpha$ can be subsumed in $\pcons$ without affecting inference on $\pinte$. Indeed, our test statistics will only depend on the \emph{increments} of $x_t$, denoted by $\Delta x_t$, and their associated ranks and, thus, they are invariant with respect to $\alpha$. We therefore omit $\alpha$ in the rest of this paper.

To eliminate the nuisance intercept parameter $\pcons$ in~(\ref{eqn:model_1}), one can directly impose an invariance restriction in the sequence of predictive regression experiments. For instance, the \citet{JanssonMoreira2006} test is based on the maximal invariant statistic $(y_2- y_1, y_3 - y_1, \dots, y_T - y_1)^\prime$. In the present paper, our statistic is only based on $y_t$'s through their ranks and, thus, also enjoys finite-sample invariance w.r.t.\ $\mu$. To simplify notation, we set $\pcons=0$ throughout the paper and nowhere assume $\rEf(\eyt)=0$. We will need to impose $\rEf(\ext)=0$: allowing for deterministic trends in $x_t$ would lead to an entirely different asymptotic analysis.

Summarizing, as outlined in the introduction, we assume that the innovations $\et=(\eyt,\ext)^\trans$ are independent and identically distributed (i.i.d.) with (bivariate) density $f$ satisfying the following condition.
\begin{assumption}\label{ass:density_f}
\begin{enumerate}
\item[(a)] $\rEf(\ext)=0$ and $\var_f(\et)=\begin{pmatrix}\sigma_{y}^2 & \rho\sigma_{y}\sigma_{x} \\ \rho\sigma_{y}\sigma_{x} & \sigma_{x}^2\end{pmatrix}$ is a finite positive-definite matrix.
\item[(b)] The density $f$ is absolutely continuous with a.e.\ derivative $\dot{f}=\begin{pmatrix}\dot{f}_y \\ \dot{f}_x\end{pmatrix}$.
\item[(c)] The (standardized) Fisher information for location,
\begin{align*}
\Jf=\begin{pmatrix}\Jfyfy & \Jfyfx \\ \Jfyfx & \Jfxfx\end{pmatrix}= \rEf\left(\score_f \score_f^\trans\right),
\end{align*} 
where $\score_f$ is the (standardized) location score function
\begin{align*}
\score_f
=\begin{pmatrix}\sigma_y\score_{f_y} \\ \sigma_x\score_{f_x}\end{pmatrix}
=\begin{pmatrix}-\sigma_y\dot{f}_y/f \\ -\sigma_x\dot{f}_x/f\end{pmatrix},
\end{align*}
is finite.\footnote{Being a Fisher information for location, $\Jf$ is automatically nonsingular and positive definite, see \citet[Theorem 2.3]{MayerWolf1990}.}
\item[(d)]  $f>0$.
$\hfill \Box$
\end{enumerate}
\end{assumption}

Let $\F$ denote the set of densities satisfying Assumption~\ref{ass:density_f}.

The Fisher information $\Jf$ and scores $\score_f$ for location are standardized in the sense that they are actually those related to $\eyt/\sigma_y$ and $\ext/\sigma_x$. As a result, $\score_f$ and $\Jf$ do not depend on $\sigma_y$ or $\sigma_x$. Note, however, that they both still depend on the correlation between the innovations $\eyt$ and $\ext$, i.e., they still depend on $\rho$.

We are interested in (optimal) tests for the (composite) null hypothesis
\begin{align}
\rH_0:\, \pinte=0,\,\ppers\in(-1,1],\, f\in\F,
\end{align}
versus the one-sided alternative
\begin{align}
\rH_1:\, \pinte>0,\,\ppers\in(-1,1],\, f\in\F.
\end{align}

As the literature focuses on test derived using an assumed Gaussian innovation density, we will throughout this paper consider Gaussian densities as a special case. This will allow us to make explicit where the power improvements come from in the case of non-Gaussian, serially independent, innovations $(\ey,\ex)$.

\begin{remark}[Gaussian $f$]\label{rem:Gaussian1}
In case $f$ is zero-mean bivariate Gaussian with correlation matrix $\mathbf{R} = \begin{pmatrix}
1 & \rho \\ \rho & 1\end{pmatrix}$, Assumption~\ref{ass:density_f} is satisfied with $\score_f(\ey,\ex) = \mathbf{R}^{-1}\begin{pmatrix} \ey/\sigma_y \\ \ex/\sigma_x \end{pmatrix}$ and $\Jf=\mathbf{R}^{-1}$.
\end{remark}

\subsection{Local perturbations}\label{subsec:LocalPerturbations}

\noindent Following the by now standard approach in the literature, we study the limit experiment in the sense of H\a'{a}jek-Le Cam by considering local alternatives for all model parameters, that is, for both the parameter of interest $\pinte$ and the nuisance parameters ($\ppers$ and $f$). For $\pinte$ and $\ppers$ the appropriate rates of convergence are well known, see, e.g., \citet{ElliottStock1994}, \citet{CampbellYogo2005}, or \cite{JanssonMoreira2006}. More precisely, we consider a $T^{-1}$-localization rate for $\pinte$ and $\ppers$, i.e.,
\begin{align} \label{eqn:localization_1}
\pinte=\pinte^{(T)}(\lpinte)=\frac{\lpinte}{T}\frac{\sigma_y}{\sigma_x}, ~~~ \ppers=\ppers^{(T)}(\lppers)=1+\frac{\lppers}{T},
\end{align}
with $\lpinte\in\SR$ and $\lppers\in(-\infty,0]$.\footnote{We use here the common approach in the literature to restrict the nuisance parameter $\lppers$ to $(-\infty,0]$. We conjecture that all results remain valid, with the obvious modifications, in case one would choose the larger parameter space $c\in\SR$; see, e.g., \citet{MoreiraMourao2017}.} Observe that the local perturbation for $b$ features a scaling by $\sigma_y/\sigma_x$. This ensures that the limit experiment will not depend on $\sigma_y$ and $\sigma_x$ (although it still depends on $\rho$).

The nuisance parameter $f$ is infinite dimensional, so it is somewhat more involved to describe its relevant local perturbations. Introduce the separable Hilbert space
\begin{equation}
\rL_2^{0,f}=\rL_2^{0,f}(\SR^2,\mathcal{B})
=\left\{  \fpert\in \rL_2^f(\SR^2,\mathcal{B})\, \left|\,  
\rEf \fpert(\e)=0,\, \rEf \ex\fpert(\e)=0\right.\right\}, 
\end{equation}
where $ \rL_2^f(\SR^2,\mathcal{B})$ denotes, the space of Borel-measurable functions $\fpert:\,\SR^2\to\SR$ satisfying $\rEf \fpert^2(\e)=\int_{\SR^2}\fpert^2(\e) f(\e)\rd\e<\infty$. The model assumption $\rEf(\ext)=0$ induces the restriction that local perturbations for $f$ are orthogonal to the first component of $\e$: $\rEf \ex\fpert(\e)=0$.

The separability of the Hilbert space $\rL_2^{0,f}$ ensures the existence of a countable orthonormal basis $\fpert_k$, $k\in\SN$, such that each $\fpert_k$ is bounded and two times continuously differentiable with bounded derivatives; see, e.g., \citet[Theorem~3.14]{Rudin1987}. Therefore, any function $\fpert\in \rL_2^{0,f}$ can be written as $\fpert= \sum_{k=1}^\infty  \eta_k \fpert_k $, for some $\eta = (\eta_k)_{k\in\SN} \in \ell_2=\{ (z_k)_{k\in\SN}\, | \, \sum_{k=1}^\infty z_k^ 2<\infty\}$. Besides the space $\ell_2$, we also need the space $c_{00}$ which is defined as the subset of sequences with finite support, i.e.,
\begin{align} \label{eqn:c00}
c_{00}=\left\{ (z_k)_{k\in\SN}\in\SR^\SN  \, \left|\, \sum_{k=1}^\infty 1\{z_k\neq 0\}<\infty\right.  \right\}. 
\end{align}

Observe that $c_{00}$ is a dense subspace of $\ell_2$. It is introduced only in the asymptotic analysis to avoid convergence of infinite-dimensional processes and possibly induced mathematical complications, see Section~\ref{subsec:PartialSumProcesses}. However, the restriction $\eta\in c_{00}$ will not affect our conclusions. Indeed, considering $\eta\in c_{00}$ restricts our analysis to a subset of all semiparametric models which potentially makes the obtained upper bound higher. However, as we are able to show that this higher upper bound is (point-wisely) attainable by feasible tests for arbitrary innovation density in sequence, see Remark~\ref{rem:attainability}, it constitutes the semiparametric power envelope and the test is semiparametrically optimal.

We  model local perturbations to the innovation density $f$ as
\begin{align} \label{eqn:dens_perturb}
f_{\eta}^{(T)}(e)
 =
f(e)\left(1+\frac{1}{\sqrt{T}}\sum_{k=1}^\infty \eta_k \fpert_k(e)\right)
\mbox{ for all }e\in\SR^2,
\end{align} 
where $\eta\in c_{00}$. We thus use a standard localization rate $T^{-1/2}$ for the bivariate density $f$. Indeed, Proposition~\ref{proposition_LAQ} below shows that all the above rates are appropriate in the sense that they lead to contiguous alternatives for the induced probability measures as $T$ tends to infinity.

In order to show that the above localization of the innovation density is valid, we need to establish that $f^{(T)}_{\eta}\in\F$. This is the content of the next proposition.
\begin{proposition} \label{prop:disturbanceinF}
Let $f\in\F$ and $\eta\in c_{00}$, then there exists a finite integer $\widetilde{T}$ such that for all $T\geq\widetilde{T}$ we have $f^{(T)}_{\eta}\in\F$.
\end{proposition}
The proof uses exactly the same arguments as in the proof of Proposition~3.1 in \citet{ZvdAW2016}, but with support $\SR^2$ instead of $\SR$. It is therefore omitted.

In terms of the local parameters $\lpinte$, $\lppers$, and $\eta$, the hypothesis of interest becomes
\begin{align}\label{eqn:hypothesis_local}
\rH_0:\, \lpinte=0,\,\lppers\in\SR,\,\eta\in c_{00},
\end{align}
versus the one-sided alternative
\begin{align}
\rH_1:\, \lpinte>0,\,\lppers\in\SR,\,\eta\in c_{00}.
\end{align}

\subsection{Partial-sum processes}\label{subsec:PartialSumProcesses}

\noindent In order to derive the limiting experiment for the predictive regression model, we need to introduce some partial-sum processes and study their asymptotic behavior. We denote by $\law_{\lpinte,\lppers,\eta;f}$ the law of $(y_1,x_1)^\trans,\dots,(y_T,x_T)^\trans$ under the model (\ref{eqn:model_1})--(\ref{eqn:model_2}), where the parameters $\pinte$ and $\ppers$ are given by~(\ref{eqn:localization_1}) and the innovation density is given by~(\ref{eqn:dens_perturb}). Formally, we define the sequence of experiments of interest as 
\begin{equation}
\SequenceExperiment
 :=
\left(\Omega^{(T)},\mathcal{F}^{(T)},\left\{\law_{\lpinte,\lppers,\eta;f}:\lpinte,\lppers\in\SR,\eta\in c_{00}\right\}\right), ~~~ T\in\SN,
\end{equation}
where $\Omega^{(T)}:=\SR^{2\times T}$ and $\mathcal{F}^{(T)}:=\mathcal{B}(\SR^{2\times T})$. We denote the expectation taken under the measure $\law_{0,0,0;f}$ by $\rE$.

Let us already mention that we will also introduce a collection of probability measures $\prob_{\lpinte,\lppers,\eta}$, defined on a probability space $(\Omega,\mathcal{F})$, representing the limit experiment $\LimitExperiment$ in Section~\ref{subsec:limitexperiment} below; see~(\ref{eqn:LimitExperiment}). We will denote he expectation taken under the measure $\prob_{0,0,0}$ by $\rElim$. That is, $\law$ and $\rE$ refer to finite-sample distributions in the sequence of experiments, while $\prob$ and $\rElim$ refer to distributions in the limit experiment.

As a final ingredient for our analysis, we introduce some partial-sum processes that we use throughout to link the sequence of experiments $\SequenceExperiment$ to the limit experiment $\LimitExperiment$. In particular, define, with $\Delta x_t := x_t - x_{t-1}$, the partial-sum processes\footnote{One may consider partial sum processes that start at $t=2$ in order to make them exactly invariant to translations in $x_t$. This would, clearly, have no effect on our asymptotic results.}
\begin{align}
\label{eqn:WTedef}
\WTe(s)
 &:=
\frac{1}{\sqrt{T}}\sum_{t=1}^{\lfloor sT\rfloor}\frac{\Delta x_t}{\sigma_x},\\
\label{eqn:WTfydef}
\WTfy(s)
 &:=
\frac{1}{\sqrt{T}}\sum_{t=1}^{\lfloor sT\rfloor}\sigma_y\score_{f_y}(\yt,\Delta x_t),\\
\label{eqn:WTfxdef}
\WTfx(s)
 &:=
\frac{1}{\sqrt{T}}\sum_{t=1}^{\lfloor sT\rfloor}\sigma_x\score_{f_x}(\yt,\Delta x_t),\\
\label{eqn:WTkdef}
\WTk(s)
 &:=
\frac{1}{\sqrt{T}}\sum_{t=1}^{\lfloor sT\rfloor}\fpert_{k}(\yt,\Delta x_t), ~~~ k \in \SN.
\end{align}
Here we standardize the first three partial-sum processes by the standard deviations $\sigma_y$ and $\sigma_x$ in order to make their limits scale invariant.
Under $\law_{0,0,0;f}$, by the Functional Central Limit Theorem (see also Lemma~\ref{lem:partialSum}), we have
\begin{align} \label{eqn:partialsumconvergence}
\begin{pmatrix}
\WTe(\s) \\ \WTfy(\s) \\ \WTfx(\s) \\ \WTb(\s) 
\end{pmatrix} 
\wto
\begin{pmatrix}
\We(\s) \\ \Wfy(\s) \\ \Wfx(\s) \\ \Wb(\s)
\end{pmatrix} 
, ~~~ \s\in[0,1],
\end{align}
where the Brownian motions $\We$, $\Wfy$, $\Wfx$ and $\Wb$ are defined on the common probability space $\left(\Omega,\mathcal{F},\prob_{0,0,0}\right)$. We have to be precise about the notion of weak convergence adopted in~(\ref{eqn:partialsumconvergence}) as $\Wb$ is infinite dimensional. In line with stochastic process theory, we mean that all finite-dimensional subprocesses of $\WTb$ weakly converges in the space $D^{M+3}[0,1]$ with the uniform topology, where $M$ is the dimension of the finite-dimensional subprocess considered. This is precisely because we take the local parameter $\eta$ to be in $c_{00}$. For the sake of convenient notation, we write the seemingly infinite-dimensional convergence~(\ref{eqn:partialsumconvergence}). As argued above, we are ultimately able to attain the semiparametric power envelope induced under the restriction $\eta\in c_{00}$ so that we can claim semiparametric optimality.

Next,
define the column vectors $\Jfyb=(\Jfybk)_{k\in\SN}$ and $\Jfxb=(\Jfxbk)_{k\in\SN}$, where $\Jfybk:=\rEf\left[\sigma_y\score_{f_y}(\et)\fpert_k(\et)\right]$ and $\Jfxbk:=\rEf\left[\sigma_x\score_{f_x}(\et)\fpert_k(\et)\right]$. As we have the equalities $\rEf\left[\ext\score_{f_y}(\et)\right]=-\sigma_y\int_{\SR^2}\ex\frac{\dot{f}_y(\e)}{f(\e)}f(\e)\rd\e=-\sigma_y\int_{\SR^2}\ex\dot{f}_y(\e)\rd\e=0$ and $\rEf\left[\ext\score_{f_x}(\et)\right]=-\sigma_x\int_{\SR^2}\ex\frac{\dot{f}_x(\e)}{f(\e)}f(\e)\rd\e=-\sigma_x\int_{\SR^2}\ex\dot{f}_x(\e)\rd\e=\sigma_x\int_{\SR^2}f(\e)\rd\e=\sigma_x$, the behavior of the Brownian motions $\We$, $\Wfy$, $\Wfx$ and $\Wb$ is described by the covariance matrix
\begin{align} \label{eqn:covariancematrix}
\var\begin{pmatrix} \We(1) \\ \Wfy(1) \\ \Wfx(1) \\ \W_{\fpert}(1) \end{pmatrix} = 
\begin{pmatrix}
1        & 0        & 1          & 0 \\
0        & \Jfyfy   & \Jfyfx     & \Jfyb^\trans \\
1        & \Jfyfx   & \Jfxfx     & \Jfxb^\trans \\
0        & \Jfyb    & \Jfxb      & I_{\infty}
\end{pmatrix},
\end{align}
where $I_{\infty}$ denotes the $\infty$-dimensional identity matrix. The scaling by $\sigma_x$ and $\sigma_y$ introduced in~(\ref{eqn:WTedef})--(\ref{eqn:WTkdef}) is indeed such that the covariance matrix~(\ref{eqn:covariancematrix}) does not depend on $\sigma_x$ or $\sigma_y$. Again, it still depends on $\rho$ through the various $J$ matrices.


Recall that the functions $h_k$ form an orthonormal basis for all zero-mean finite-variance functions that are orthogonal to $\ext$. In view of the covariance matrix~(\ref{eqn:covariancematrix}), we may thus write, for $\s\in[0,1]$,
\begin{align}
\label{eqn:decomposition_Wfy}
\Wfy(\s)
 &=
\Jfyb^\trans\Wb(\s),\\
\label{eqn:decomposition_Wfx}
\Wfx(\s)
 &=
\We(\s) + \Jfxb^\trans\Wb(\s).
\end{align}
Consequently, we also have
\begin{align}
\label{eqn:decomposition_Jfy}
\var\big[\Wfy(1)\big]
 &=
\Jfyfy = \Jfyb^\trans\Jfyb,\\
\label{eqn:decomposition_Jfx}
\var\big[\Wfx(1)\big]
 &=
\Jfxfx = 1+\Jfxb^\trans\Jfxb,\\
\label{eqn:decomposition_Jfyx}
\cov\big[\Wfy(1),\Wfx(1)\big]
 &=
\Jfyfx = \Jfyb^\trans\Jfxb.
\end{align} 

We again consider the special case of a Gaussian density $f$.

\begin{remark}[Gaussian $f$]\label{rem:Gaussian2}
In the situation of Gaussian $f$ as discussed in Remark~\ref{rem:Gaussian1}, we may write the decomposition~(\ref{eqn:decomposition_Jfx}) as $\Wfx=\We-\frac{\rho}{\sqrt{1-\rho^2}}\Wperp$ where $\Wperp$ is the standard Brownian motion generated by the increments $\left(\ey/\sigma_y-\rho\ex/\sigma_x\right)/\sqrt{1-\rho^2}$. Indeed, $\We$ and $\Wperp$ are independent (calculate the correlation of the increments that generate both processes). Thus, we also find $\Jfxb^\trans\Wb(\s)=-\rho\Wperp$ and the decomposition~(\ref{eqn:decomposition_Jfx}) becomes $\Jfxfx=1+\frac{\rho^2}{1-\rho^2}=\frac{1}{1-\rho^2}=\Jfyfy$. Moreover, we have $\Wfy=\frac{1}{\sqrt{1-\rho^2}}\Wperp$ and $\Jfyfx=-\frac{\rho}{1-\rho^2}$. 
\end{remark}

\section{Eliminating the nuisance parameter \lowercase{$f$} by invariance}\label{sec:PowerEnvelope}
\noindent We first focus on eliminating the nuisance parameter $f$ from the testing problem outlined in Section~\ref{sec:Model}. We will see that this can be handled using invariance arguments in the \emph{limit} experiment, which we derive in Section~\ref{subsec:limitexperiment}. In Section~\ref{sec:FeasibleTests}, we consider the nuisance parameter $\ppers$.

We take the following steps in this section:
\begin{enumerate}
	\item Provide a structural representation of the limit experiment (Section~\ref{subsec:limitexperiment}).
	\item Characterize maximally invariant test statistics in this limit experiment (Section~\ref{subsec:MaximalInvariant}).
	\item Provide a structural representation of the invariant limit experiment (Section~\ref{subsec:invariantlimitexperiment}).
	\item Provide a feasible version of the asymptotically invariant test statistics to be applied in the sequence of predictive regression experiments (Section~\ref{subsec:RankBasedStatistics}).
\end{enumerate}
These steps also show that, to eliminate the nuisance parameter $f$, instead of studying invariance restrictions in the \emph{sequence} of finite-sample experiments, we only impose them in the \emph{limit} experiment. Unlike for the location parameter $\mu$ (of $\eyt$), this limiting invariance property of the parameter $f$ does not follow directly from exact finite-sample invariance properties. Notably, the existing tests in the literature share this feature, as they also (implicitly) impose the invariance restriction in the limit, though not in the sequence; see Remark~\ref{rem:GaussianInvariant}. As far as we know, all existing tests belong to the class of asymptotically invariant (w.r.t.\ $f$) tests, while our test is semiparametrically optimal in the model we study. Section~\ref{sec:MonteCarlo} shows that this approach leads to considerable power gains in case the innovations are non-Gaussian, while no power is lost under Gaussianity.

\subsection{A Structural Representation of the Limit Experiment}\label{subsec:limitexperiment}
\noindent We consider the limit experiment corresponding to the predictive regression model~(\ref{eqn:model_1})--(\ref{eqn:model_2}) using the local perturbations~(\ref{eqn:localization_1}) and~(\ref{eqn:dens_perturb}), i.e., the limit of the experiments $\SequenceExperiment$ indexed by $T$, by studying the asymptotic behavior of the induced likelihood ratios. We expand the likelihood ratio around $(\pinte,\ppers,\eta)=(0,1,0)$ and derive its limit in the following proposition, which can be interpreted as a generalization of Lemma~4 in \citet{JanssonMoreira2006} by including non-Gaussian distributions and perturbations thereof.\footnote{As preparation for the results in Section~\ref{sec:FeasibleTests}, we allow in this proposition for local perturbations with respect to $\ppers$ even though, in the present section, $\ppers$ is assumed to be known.}
\begin{proposition}\label{proposition_LAQ}
Fix $f\in\F$. Consider the local parameters $\lpinte\in\SR$, $\lppers\in\SR$, and $\eta\in c_{00}$. Then,
\begin{itemize}
\item[(i)] Under $\law_{0,0,0;f}$, the log-likelihood ratio of the predictive regression experiment satisfies, as $T\to\infty$,
\begin{align}\label{eqn:LAQ}
\log\frac{\rd\law_{\lpinte,\lppers,\eta;f}}{\rd\law_{0,0,0;f}} = \cs^{(T)}(\lpinte,\lppers,\eta) - \frac12\qt^{(T)}(\lpinte,\lppers,\eta) + \opone,
\end{align}
where
\begin{align*}
\cs^{(T)}(\lpinte,\lppers,\eta) =&~ \frac{\lpinte}{T}\sum_{t=1}^T \frac{x_{t-1}}{\sigma_x} \sigma_y\score_{f_y}(\yt, \Delta x_t) + \frac{\lppers}{T}\sum_{t=1}^T x_{t-1} \score_{f_x}(\yt, \Delta x_t) \\
&~+ \frac{1}{\sqrt{T}}\sum_{t=1}^T \sum_{k}\eta_k\fpert_{k}(\yt, \Delta x_t), \\
\qt^{(T)}(\lpinte,\lppers,\eta) =&~  \left(\lpinte^2\Jfyfy+\lppers^2\Jfxfx+2\lpinte\lppers \Jfyfx\right)\frac{1}{T^2}\sum_{t=1}^T \frac{x_{t-1}^2}{\sigma_x^2} \\
&~ + \left(2b\Jfyb^\trans\eta+2c\Jfxb^\trans\eta\right)\frac{1}{T^{3/2}}\sum_{t=1}^{T}\frac{x_{t-1}}{\sigma_x} + \eta^\trans\eta.
\end{align*}
\item[(ii)] Still under $\law_{0,0,0;f}$, as $T\to\infty$, we have
\begin{align} \label{eqn:LLRlimit}
\log\frac{\rd\law_{\lpinte,\lppers,\eta;f}}{\rd\law_{0,0,0;f}} \wto \lllr(\lpinte,\lppers,\eta) = \cs(\lpinte,\lppers,\eta) - \frac12\qt (\lpinte,\lppers,\eta),
\end{align}
where
\begin{align*}
\cs(\lpinte,\lppers,\eta)
 =&~ \lpinte\int_0^1 \We(\s)\rd\Wfy(\s) + \lppers\int_0^1 \We(\s)\rd\Wfx(\s) + \eta^\trans\Wb(1)\\
 =&~ \int_0^1 \We(\s)\left(\lpinte\Jfyb+\lppers\Jfxb\right)^\trans\rd\Wb(\s) + \lppers\int_0^1 \We(\s)\rd\We(\s) + \eta^\trans\Wb(1), \\
\qt(\lpinte,\lppers,\eta)
 =&~
\left(\lpinte^2\Jfyfy+\lppers^2\Jfxfx+2\lpinte\lppers\Jfyfx\right)\int_0^1 \We(\s)^2\rd\s \\
 &~
+ \eta^\trans\eta + \left(2\lpinte\Jfyb^\trans\eta+2\lppers\Jfxb^\trans\eta\right)\int_0^1\We(\s)\rd\s\\
 =&~
\int_0^1\left|\left(\lpinte\Jfyb+\lppers\Jfxb\right)\We(\s)+\eta\right|^2\rd\s
	+ \lppers^2\int_0^1\We(\s)^2\rd\s.
\end{align*}
\item[(iii)] For every $\lpinte,\lppers\in\SR$ and $\eta\in c_{00}$, under $\prob_{0,0,0}$, $\E[\exp\left(\lllr(\lpinte,\lppers,\eta)\right)]=1$.
\end{itemize}
\end{proposition}
A proof of Proposition~\ref{proposition_LAQ} is provided in Appendix~\ref{app:Proofs}, but let us give a brief sketch here. Part~(i) is immediate from an informal Taylor expansion of the log-likelihood ratios and, formally, follows from \citet{HvdAW2015}, which provides generally applicable sufficient conditions for the quadratic expansion of likelihood ratios with densities that are differentiable in quadratic mean (DQM). This DQM condition is implied, for location models, by the absolutely continuity of the innovation density function and finiteness of the associated Fisher information, i.e., precisely the content of Assumption~\ref{ass:density_f}. A detailed discussion can be found in \citet[Section 17.3]{LeCam1986} or \citet[Section 7.3]{LeCamYang2000}. Part~(ii) follows from the continuous mapping theorem applied to the weak convergence in~(\ref{eqn:partialsumconvergence}).
Both forms of the central sequence $\cs$ and quadratic term $\qt$ follow from~(\ref{eqn:decomposition_Wfy}) and~(\ref{eqn:decomposition_Wfx}).  Part~(iii) follows from standard stochastic calculations concerning Dol\a'{e}ans-Dade exponentials. To see this, note that $\Wb$ and $\We$ are independent in view of~(\ref{eqn:covariancematrix}) and, thus, have vanishing quadratic covariation.

Part~(iii) of Proposition~\ref{proposition_LAQ} ensures that we can introduce a collection of probability measures $\prob_{\lpinte,\lppers,\eta}$ on the measurable space $\left(\Omega,\mathcal{F}\right)$ (on which the Brownian motions $\We$, $\Wfy$, $\Wfx$ and $\Wb$ are defined) by the Radon-Nikodym derivative
\begin{align} \label{eqn:RadonNikodym}
\frac{\rd\prob_{\lpinte,\lppers,\eta}}{\rd\prob_{0,0,0}} = \exp\lllr(\lpinte,\lppers,\eta),
\end{align}
where $\lllr(\lpinte,\lppers,\eta)$ is defined in~(\ref{eqn:LLRlimit}). Then, in the sense of H\a'{a}jek-Le Cam (see, for instance, \citet{vdVaart2000}, Chapter~9), the sequence of predictive regression experiments, indexed by sample size $T$, weakly converges to the limit experiment described by the measures $\prob_{\lpinte,\lppers,\eta}$. We formally define this limit experiment by
\begin{align}\label{eqn:LimitExperiment}
\LimitExperiment
 :=
\Big(\Omega, \mathcal{F}, \Big\{\prob_{\lpinte,\lppers,\eta}:\lpinte,\lppers\in\SR,\eta\in c_{00}\Big\}\Big),
\end{align}
where $\Omega:=C[0,1]\times C[0,1]\times C[0,1]\times C^{\SN}[0,1]$ and $\mathcal{F}:=\mathcal{B_C}\otimes\mathcal{B_C}\otimes\mathcal{B_C}\otimes(\otimes_{k=1}^\infty\mathcal{B_C})$.

The following statement is an immediate consequence of Proposition~\ref{proposition_LAQ}. 
\begin{corollary}
Let $f\in\F$, then the sequence of experiments $\SequenceExperiment$ converges to the limit experiment $\LimitExperiment$ as $T\to\infty$.
\end{corollary}

Although the log-likelihood ratios $\lllr(\lpinte,\lppers,\eta)$ formally describe the limiting experiment, it is more insightful to provide, what we call, a structural representation. This structural representation provides a fixed-horizon continuous-time model for which the likelihoods are exactly equal to $\exp\left(\lllr(\lpinte,\lppers,\eta)\right)$. From a statistical point of view, the induced experiments are thus equal. The result follows from an immediate application of Girsanov's theorem to the Radon-Nikodym derivates~(\ref{eqn:LLRlimit}). Its proof is therefore omitted.
\begin{theorem}\label{thm:StructuralLimitExperiment}
Fix $f\in\F$. Let, under $\prob_{0,0,0}$, $\Ze$, and $\Zb$ be zero-drift Brownian motions with covariance according to the first and last row and column of~(\ref{eqn:covariancematrix}). The limit experiment $\LimitExperiment$ can be described as: observe $\left\{\left(\We(s),\Wb(s)\right):s\in[0,1]\right\}$ generated by
\begin{align}
\label{eqn:StrucLimitWe}
\rd\We(s)
 &=
\lppers \We(s)\rd s + \rd \Ze(s),\\
\label{eqn:StrucLimitWb}
\rd\Wb(s)
 &=
(\lpinte\Jfyb+\lppers\Jfxb)\We(s)\rd s + \eta\rd s + \rd \Zb(s).
\end{align}
\end{theorem}
A few remarks can be made in relation to Theorem~\ref{thm:StructuralLimitExperiment}. First, note that for $b=c=0$ and $\eta=0$, we obtain $\We=\Ze$ and $\Wb=\Zb$.
Secondly, the theorem essentially states that while $\left(\We,\Wb^\trans\right)^\trans$ is a zero-drift Brownian motion under $\prob_{0,0,0}$, it becomes an Ornstein-Uhlenbeck process under $\prob_{\lpinte,\lppers,\eta}$, where the log-likelihood ratio $\log\left(\rd\prob_{\lpinte,\lppers,\eta}/\rd\prob_{0,0,0}\right)$ equals $\lllr(\lpinte,\lppers,\eta)$. Observe in particular that local perturbations of the innovation density $f$, as described by $\eta$, only affect the drift in~(\ref{eqn:StrucLimitWb}). We will consider inference procedures that are invariant with respect to $\eta$ in the limit experiment. In terms of the (sequence of) predictive regression model(s) this consequently translates into invariance with respect to (local perturbations in) the innovation density $f$.

In view of~(\ref{eqn:decomposition_Wfy})--(\ref{eqn:decomposition_Wfx}), we may also write
\begin{align}
\label{eqn:dWlfy}
\rd\Wfy(s)
 &=
(\lpinte\Jfyfy+\lppers\Jfyfx) \We(s)\rd s + \Jfyb^\trans\eta\rd s + \rd \Zfy(s), \\
\label{eqn:dWlfx}
\rd\Wfx(s)
 &=
(\lpinte\Jfyfx+\lppers\Jfxfx) \We(s)\rd s + \Jfxb^\trans\eta\rd s + \rd \Zfx(s),
\end{align}
where $\Zfx$ and $\Zfy$ are zero-drift Brownian motions under $\prob_{0,0,0}$. However, these equations do not contain any additional information, precisely given~(\ref{eqn:decomposition_Wfy}) and~(\ref{eqn:decomposition_Wfx}). Nevertheless, they will turn out useful when describing the likelihood ratio of the maximal invariant $\mi$ to be introduced below in~(\ref{eqn:maximalinvariant}).

\subsection{Maximal Invariant} \label{subsec:MaximalInvariant}

\noindent In the limit experiment $\LimitExperiment$, the parameter $\lpinte\in\SR$ is the parameter of interest, while $\lppers\in\SR$ and $\eta\in c_{00}$ are nuisance parameters. Observe that the nuisance parameter $\eta$ appears only in the drift of the SDEs in Theorem~\ref{thm:StructuralLimitExperiment}. This suggests an invariance restriction in line with the approach in~\citet{ZvdAW2016} for unit root testing.

To be specific, we first introduce, for $\eta\in c_{00}$, the transformations $\transfor_{\eta}:C^\SN[0,1]\to C^\SN[0,1]$ by
\begin{align}
[\transfor_{\eta}(W)](s) = W(s) - \eta\s,  
\end{align}
for $W\in C^{\SN}[0,1]$ and all $\s\in[0,1]$. The transformation $\transfor_{\eta}$ adds a drift $\s\mapsto-\eta\s$ to $W$. Thus, Theorem~\ref{thm:StructuralLimitExperiment} implies that the law of $\left(\We,(\transfor_{\eta}(\Wb))^\trans\right)^\trans$ under $\prob_{\lpinte,\lppers,0}$ is the same as the law of $\left(\We,\Wb^\trans\right)^\trans$ under $\prob_{\lpinte,\lppers,\eta}$.\footnote{By~(\ref{eqn:decomposition_Wfx}) and~(\ref{eqn:decomposition_Wfy}), the same holds for $\Wfx$ and $\Wfy$.} Denote by $\transgroup_{\eta}$ the group of transformations $\transfor_{\eta}$ for $\eta\in c_{00}$. We can now characterize the maximal invariant with respect to $\transgroup_{\eta}$ in the limit experiment $\LimitExperiment$.

For any process $W$, we define the associated \textit{bridge process} by 
\begin{align} \label{eqn:bridgeprocessoperator}
B^W(\s):=W(\s)-\s W(1),
\end{align}
for all $\s\in[0,1]$. Then, one readily verifies
\begin{align*}
B^{\transfor_{\eta}(W)}(\s) 
=&~ [\transfor_{\eta}(W)](\s) - \s[\transfor_{\eta}(W)](1)\\
=&~ W(\s)-\eta\s - \s(W(1)-\eta) \\
=&~ W(\s) - \s W(1) \\
=&~ B^{W}(\s).
\end{align*}
As a result, the bridges $B^{\Wb}$ are invariant under the transformations $\transfor_{\eta}$.

Define the mapping $M$ by $M(\We,\Wb):=(\We,\Bb)$. It then follows that statistics that are measurable with respect to the $\sigma$-field
\begin{align} \label{eqn:maximalinvariant}
\mi = \sigma\left(M(\We,\Wb)\right) = \sigma\left(\We,\Bb\right),
\end{align}
are invariant with respect to $\transfor_{\eta}$ for all $\eta\in c_{00}$. Moreover, in the following theorem, we show $\mi$ to be \emph{maximally} invariant. Its proof is, again, provided in Appendix~\ref{app:Proofs}.
\begin{theorem}\label{thm:MaximalInvariant}
In the limit experiment $\LimitExperiment$, for $\eta\in c_{00}$, the $\sigma$-field $\mi$ in~(\ref{eqn:maximalinvariant}) is maximally invariant with respect to $\transgroup_{\eta}$. 
\end{theorem}

\subsection{A Structural Representation of the Invariant Limit Experiment}\label{subsec:invariantlimitexperiment}
\noindent Theorem~\ref{thm:MaximalInvariant} implies that any inference invariant with respect to $\transgroup_{\eta}$ must be measurable with respect to $\mi$; see, e.g., \citet[Theorem 6.2.1]{LehmannRomano2005}. Therefore, by the Neyman-Pearson lemma, inference based on the likelihood ratio with respect to $\mi$ yields the power envelope for invariant tests in the limit experiment $\LimitExperiment$. The following result provides this likelihood ratio.
\begin{theorem} \label{thm:LAQ_M}
Fix $f\in\F$. Then the likelihood ratios in the limit experiment $\LimitExperiment$ restricted to the maximal invariant $\mi$ are given by
\begin{align} \label{eqn:LAQ_M}
\exp{\lllr_\mi(\lpinte,\lppers)}
:=
\frac{\rd\prob_{\lpinte,\lppers}^{\mi}}{\rd\prob_{0,0}^{\mi}}
 =
\rElim\left[\frac{\rd\mathbb{P}_{\lpinte,\lppers,\eta}}{\rd\mathbb{P}_{0,0,0}}|\mi\right]
 =
\exp\left(\cs_{\mi}(\lpinte,\lppers)-\frac{1}{2}\qt_{\mi}(\lpinte,\lppers)\right),
\end{align}
where
\begin{align}
\label{eqn:DeltaM}
\cs_{\mi}(\lpinte,\lppers)
 &=
\int_0^1\We(\s)\left(\lpinte\Jfyb+\lppers\Jfxb\right)^\trans\rd\Bb(\s) 
	+\lppers\int_0^1\We(\s)\rd\We(\s)\\
 &=
\nonumber
\lpinte\int_0^1\We(\s)\rd\Bfy(\s) + \lppers\left(\int_0^1\We(\s)\rd\Bfx(\s)
	+\We(1)\overline{\We}\right),\\
\label{eqn:QM}
\qt_{\mi}(\lpinte,\lppers) 
 &=
\left(\lpinte\Jfyb+\lppers\Jfxb\right)^2\int_0^1\left(\We(\s)-\overline{\We}\right)^2\rd\s
+ \lppers^2\int_0^1\We(\s)^2\rd\s\\
 &=
\nonumber
\left(\lpinte^2\Jfyfy+\lppers^2(\Jfxfx-1)+2\lpinte\lppers\Jfyfx\right)
	\left(\overline{\We^2}-(\overline{\We})^2\right)
	+ \lppers^2\left(\overline{\We}\right)^2,
\end{align}
with $\overline{\We^2} = \int_0^1\We(\s)^2\rd\s$ and $\overline{\We} = \int_0^1\We(\s)\rd\s$. 
\end{theorem}
The proof is provided in Appendix~\ref{app:Proofs}. The first ways to write $\cs_{\mi}(\lpinte,\lppers)$ and $\qt_{\mi}(\lpinte,\lppers)$ make explicit that the likelihood factorizes in a conditional likelihood given $\We$ and the marginal likelihood of $\We$. Both second ways to write $\cs_{\mi}(\lpinte,\lppers)$ and $\qt_{\mi}(\lpinte,\lppers)$ follow from~(\ref{eqn:decomposition_Wfy})--(\ref{eqn:decomposition_Wfx}) and~(\ref{eqn:decomposition_Jfy})--(\ref{eqn:decomposition_Jfyx}). Those are the versions that we use below to construct our feasible test statistics. Theorem~\ref{thm:LAQ_M} also immediately yields the semiparametric power envelope, still for fixed $c$, that we do not present in detail for brevity.

The restriction to invariant tests removes the nuisance parameter $\eta$ from the testing problem. Indeed, the likelihood ratio~(\ref{eqn:LAQ_M}) no longer depends on $\eta$. Therefore, we can formally define the limit experiment restricted to the maximal invariance $\mi$ as
\begin{equation}\label{eqn:LimitExperiment_M}
\LimitExperimentM
 :=
\Big(\Omega, \mi, \Big\{\prob_{\lpinte,\lppers}^{\mi}:\lpinte,\lppers\in\SR\Big\}\Big).
\end{equation}

Again, the likelihood ratios $\rd\mathbb{P}_{\lpinte,\lppers}^{\mi}/\rd\mathbb{P}_{0,0}^{\mi}$ can also be interpreted as Girsanov transformations. We state this as a corollary as the result follows immediately from calculating the bridges corresponding to $\Wfy$ and $\Wfx$ in Theorem~\ref{thm:LAQ_M}.
\begin{corollary}\label{corollary:StructuralLimitExperiment_M}
Fix $f\in\F$. Let, under $\prob_{0,0}^{\mi}$, $\Ze$ and $\Zb$ be zero-drift Brownian motions with covariance according to the first and last row and column of~(\ref{eqn:covariancematrix}). The limit experiment $\LimitExperimentM$ can be described as follows: we observe, with $\Bb(s)=\Wb(s)-s\Wb(1)$, $\left\{\left(\We(s),\Bb(s)\right):s\in[0,1]\right\}$ with $\left(\We,\Wb\right)$ generated by
\begin{align}
\label{eqn:StrucLimitInvariantWe}
\rd\We(s)
&=
\lppers \We(s)\rd s + \rd \Ze(s),\\
\label{eqn:StrucLimitInvariantWb}
\rd\Wb(s)
&=
(\lpinte\Jfyb+\lppers\Jfxb)\We(s)\rd s + \rd \Zb(s).
\end{align}
\end{corollary}
The difference between Corollary~\ref{corollary:StructuralLimitExperiment_M} and Theorem~\ref{thm:StructuralLimitExperiment} is twofold. First, besides the process $\We$, the observation in the invariant limit experiment in Corollary~\ref{corollary:StructuralLimitExperiment_M} is only the Brownian bridge $\Bb$ and not the complete Brownian motion $\Wb$. Second, as a consequence of this, the nuisance parameter $\eta$ disappeared from~(\ref{eqn:StrucLimitInvariantWb}).

Corollary~\ref{corollary:StructuralLimitExperiment_M} does not provide, as far as we know, a further invariance structure that can be used to eliminate the nuisance parameter $\lppers$. As a result, we rely, in Section~\ref{sec:FeasibleTests}, on the so-called Approximate Least Favorable Distribution method to deal with this last nuisance parameter.

We conclude this section by again considering the special case of a Gaussian innovation density $f$. This also shows where exactly our power gains, under serially independent innovations, come from relative to the Gaussian procedures in, for instance, \citet{JanssonMoreira2006}.

\begin{remark}[Attainability of the Semiparametric Power Envelope] \label{rem:attainability}
One may expect the semiparametric power envelope to be formally attainable by a likelihood-ratio test constructed using a nonparametric estimate of the score function $\score_f$. Intuitively, the argument is as follows. Rewrite $\int_0^1\We(\s)\rd\Bfy(\s) = \int_0^1\left(\We(\s) - \overline{\We}\right)\rd\Wfy(\s)$. Hence, even though there is a bias $a$ (at rate $\sqrt{T}$) in the estimated score function, this bias will be canceled out automatically since $\int_0^1\left(\We(\s) - \overline{\We}\right)\rd\left(a\s+\Wfy(\s)\right) = \int_0^1\left(\We(\s) - \overline{\We}\right)\rd\Wfy(\s)$. The same argument applies to the term $\int_0^1\We(\s)\rd\Bfx(\s)$. Compare the discussion in \citet[Section 6]{Jansson2008} for the unit root testing problem and \citet[Section 2]{Zhou2020} for general LAN, LAMN, and LABF experiments. 
\end{remark}


\begin{remark}[Gaussian $f$]\label{rem:GaussianInvariant}
In the situation of Gaussian $f$, Remark~\ref{rem:Gaussian1} and Remark~\ref{rem:Gaussian2} imply that $\Bfy$ and $\Bfx$ are linear combinations of $\Be$ and $\Bperp$ (the Brownian bridges generated by $\We$ and $\Wperp$, respectively). As a result, the optimal invariant procedures are measurable with respect to $\We$ and $\Bperp$. Using the same conditional expectation calculation, the associated log-likelihood ratio of the Gaussian $\sigma$-field, $\mi_{\rm Gaussian}=\sigma\left(\We,\Bperp\right)$, leads to the Gaussian log-likelihood ratio in \citet[Lemma~3]{JanssonMoreira2006}. As $\Bperp$ is spanned by $\Bb$, the $\sigma$-field $\mi_{\rm Gaussian}$ is also invariant w.r.t $\eta$ (or $f$), but it is not \emph{maximally} invariant. As a consequence, under non-Gaussianity, this leads to an efficiency loss in statistical inference. 

Note that all existing tests in the literature are (essentially) based on the Gaussian likelihood of the generally non-maximally invariant $\mi_{\rm Gaussian}$, e.g., \citet{JanssonMoreira2006} and \citet{EMW2015}. Therefore, these tests belong to the class of asymptotically invariant tests. This invariance imposed in the limiting experiment is associated to invariance w.r.t.\ the innovation density $f$ in the sequence as $\eta$ represents local perturbations precisely of $f$. Indeed, we have the convergence $\WTe(\s) \wto \We(\s)$ and the one associated to $\Wperp$ for all $f \in \F$, hence, $\eta$ will not enter the associated equation (\ref{eqn:StrucLimitWe}) in the limiting experiment. See \citet{Muller2011} for a more comprehensive analysis of this convergence. 
\end{remark}


\subsection{Rank-based asymptotically invariant statistics} \label{subsec:RankBasedStatistics}

\noindent The elimination of the nuisance parameter $\eta$ is performed in the limit experiment $\LimitExperiment$ and leads to $\LimitExperimentM$. We now show how this elimination can be mimicked in the actual predictive regression model of interest, i.e., in $\SequenceExperiment$. It is reasonable to expect that exploiting the asymptotic invariance structures also works ``well'' for the sequence of experiments. The claim will be substantiated by the simulation results in Section~\ref{sec:MonteCarlo}.

In line with the vast literature on rank-based inference, the appearance of the Brownian Bridges $\Bfx$ and $\Bfy$ in Corollary~\ref{corollary:StructuralLimitExperiment_M}, naturally suggest to use statistics that are based on ranks of the innovations $\eyt$ and $\ext$ in the predictive regression model. Indeed, we will follow that route. However, in the present situation we deal with bivariate innovations $(\eyt,\ext)$ which complicates the analysis considerably relative to models with univariate innovations that are mostly studied in the literature.

As the true innovation density $f$ is unknown, we actually base our test statistic on an assumed (so-called \textit{reference}) density $g$ that also satisfies Assumption~\ref{ass:density_f}. Let $g_y$ and $g_x$ denote the marginal densities for the first, respectively, second component of $g$. The bivariate nature of the innovations $(\eyt,\ext)$ implies that we cannot deal with a completely general reference bivariate density $g$. Thus, we choose marginal reference densities $g_y$ and $g_x$, and a reference correlation parameter $\refcorrp$. For the marginal reference densities, we impose the standard condition in the rank-based inference literature, see, e.g., Theorem~13.5 in \citet{vdVaart2000}.
\begin{assumption}\label{ass:ApproximateScores}
	The marginal \emph{reference densities} $g_i$, $i=\{y,x\}$, are strictly positive, absolutely continuous with derivative $\dot{g}_i$ and $J_{g_i}:=\int\left(\dot{g}_i/g_i\right)^2g_i<\infty$. Moreover, we have
	\begin{equation}\label{eqn:ApproximateScores}
	\lim_{T\to\infty}\frac{1}{T}\sum_{t=1}^T
	\left(-\frac{\dot{g}_i}{g_i}\left(G_i^{-1}\left(\frac{t}{T+1}\right)\right)\right)^2
	=
	J_{g_i},
	\end{equation}
	where $G_i^{-1}$ is the inverse cumulative distribution function associated to $g_i$.
\end{assumption}
Moreover, given an additionally chosen reference correlation $\refcorrp\in(-1,1)$, we define the associated bivariate reference score function 
\begin{align} \label{eqn:rankbasedscore_g}
\score_g(\ey,\ex) := \left(\score_{g_y}(\ey,\ex),\score_{g_x}(\ey,\ex)\right)^\trans
\end{align} 
where
\begin{align*}
\score_{g_y}(\ey,\ex)
&=
-\left(\frac{\dot{g}_y}{g_y}(\ey)-\refcorrp\frac{\dot{g}_x}{g_x}(\ex)\right)\Big/(1-\refcorrp^2),\\
\score_{g_x}(\ey,\ex)
&=
-\left(\frac{\dot{g}_x}{g_x}(\ex)-\refcorrp\frac{\dot{g}_y}{g_y}(\ey)\right)\Big/(1-\refcorrp^2).
\end{align*}
The linearity of the reference score functions $\score_{g_y}$ and $\score_{g_x}$ is key to the analysis that follows. It implies that, when using component-wise ranks of the innovations $(\ey,\ex)$, the resulting rank-based processes converge to a bivariate Brownian bridge. Despite its seemingly restrictive nature, the linearity allows use to fully exploit the invariance structures embedded in the predictive regression model of interest, leading to sizable power gains (see Section~\ref{sec:MonteCarlo}).

Now, let $R_{y,t}$ denote the rank of $y_t$ (among $y_1,\ldots,y_T$), while $R_{x,t}$ denotes the rank of $\Delta x_t=x_t-x_{t-1}$ (among $\Delta x_1,\ldots,\Delta x_T$). Note that the pairs $\left(R_{y,t},R_{x,t}\right)$ equal the (component-wise) ranks of $\left(\eyt,\ext\right)$ under $\pinte=0$ and $\ppers = 0$. We define the bivariate partial sum process of the rank-based scores by
\begin{align}\label{eqn:rankscorepartialsum}
\BTg(s) 
  =&~ \left(\BTgy(s),\BTgx(s)\right)^\trans  \nonumber \\
 :=&~ \frac{1}{\sqrt{T}}\sum_{t=1}^{\lfloor sT\rfloor} \score_g\left(G^{-1}_y\left(\frac{R_{y,t}}{T+1}\right),G^{-1}_{x}\left(\frac{R_{x,t}}{T+1}\right)\right),
\end{align}
for $s\in[0,1]$. The following result establishes the limiting behavior of $\BTg$ under $\law_{0,0,\eta;f}$. Its proof is again provided in Appendix~\ref{app:Proofs}.
\begin{proposition} \label{prop:limitbehavior_BTg}
Suppose $\et=(\eyt,\ext)^\trans$ are i.i.d.\ innovations with density $f\in\F$. Let $g_y$ and $g_x$ be reference densities that satisfy Assumption~\ref{ass:ApproximateScores} and fix the reference correlation $\refcorrp$. Then, under $\law_{0,0,\eta;f}$, we have
\begin{align} \label{eqn:partialsumconvergence_Bg}
\BTg\wto\Bg,
\end{align}
where $\Bg$ is a bivariate Brownian bridge, i.e., $\Bg(s)=\Wg(s)-s\Wg(1)$, with $\Wg$ a zero-drift Brownian motion. The covariance of $\Wg$ with $\We$ and $\Wf:=(\Wfy,\Wfx)^\trans$ is given by
\begin{equation}\label{eqn:covariancematrix_g}
\var\begin{pmatrix}
	\We(1)\\\Wf(1)\\\Wg(1)
\end{pmatrix}
 =
\begin{pmatrix}
1&\textbf{e}_1^\trans&\coveg^\trans\\\textbf{e}_1&\J_f&\J_{fg}\\\coveg&\J_{gf}&\J_g
\end{pmatrix},
\end{equation}
where
\begin{align*}
\textbf{e}_1
 &=
(0,1)^\trans, \\
\coveg
 &=
(\sigma_{\e g_y},\sigma_{\e g_x})^\trans
 =
\rEf\left[ \ext \score_g\left(G_y^{-1}(F_y(\eyt)),G_x^{-1}(F_x(\ext))\right) \right],\\
\J_{fg}
 &=
\J_{gf}^\trans
 =
\rEf\left[ \score_f(\eyt,\ext)\score_g\left(G_y^{-1}(F_y(\eyt)),G_x^{-1}(F_x(\ext))\right)^\trans \right],\\
\J_g
 &=
\rEf\left[ \score_g\left(G_y^{-1}(F_y(\eyt)),G_x^{-1}(F_x(\ext))\right)\score_g\left(G_y^{-1}(F_y(\eyt)),G_x^{-1}(F_x(\ext))\right)^\trans \right].
\end{align*}
\end{proposition}

The above result is classical for univariate rank statistics. In the present paper, we use component-wise bivariate ranks. One complication is that the matrix $\J_g$ depends on $f$ through its copula. This implies that, like $\J_g$, it will have to be estimated in applications; compare also to the discussion of Theorem~3.1 in \citet{Zhou2020}.

We use the rank-based processes $\BTg$ to replace $\Bf$ in the likelihood ratio in Theorem~\ref{thm:LAQ_M}; see Section~\ref{sec:PuttingItAllTogether} for details. In line with Remark~\ref{rem:attainability}}, one could contemplate to use reference densities $\hat{f}$ based on a non-parametric estimate of the true innovation density, but we leave a formal analysis for future work. As we will see in Section~\ref{sec:MonteCarlo}, even for incorrectly chosen reference densities (that is, for $g\neq f$), our procedure features power gains over existing Gaussian based procedures. These gains come from the assumption that the error term $\et$ is driven by some i.i.d.\ innovations, which may possibly be maintained in empirical work. It is important to note that choosing a reference density $g\neq f$ does not affect the validity of our test. The test will be of the appropriate level irrespective of the reference densities $g_y$ and $g_x$ chosen (provided they satisfy Assumption~\ref{ass:ApproximateScores}). But, likelihood ratio tests based on Theorem~\ref{thm:LAQ_M} still feature the nuisance parameter $\lppers$. We deal with this in the next section.

For completeness, we also provide the equivalent to Corollary~\ref{corollary:StructuralLimitExperiment_M} when using the reference density $g$.

\begin{corollary} \label{corollary:StructuralLimitExperiment_Mg}
Fix $f\in\F$. Let $\lpinte\in\SR$, $\lppers\in(-\infty,0]$, and $\eta\in c_{00}$. Then, under $\prob_{\lpinte,\lppers,\eta}$, the behavior of $\We$ and $\Bg$ follows
\begin{align}
\rd \We(s)
 &=
\lppers\We(s)\rd s + \rd \Ze(s),\\
\rd\Wg(s)
 &=
\J_{gf}\begin{pmatrix} \lpinte \\ \lppers \end{pmatrix}
	\We(s)\rd s + \rd\Zg(s),
\end{align}
where, under $\prob_{0,0,0}$, $\Zg$ is a bivariate Brownian motion with variance $J_g$ and covariance with $\We$ equal to $\coveg$. 
\end{corollary}


\section{Eliminating the nuisance parameter $\ppers$ by ALFD}\label{sec:FeasibleTests}

\noindent
In the previous section, we have developed the semiparametric power envelope for tests on $\lpinte$ that are invariant with respect to $\eta$, under the assumption that $\lppers$ is known. We now address the question of testing the regression coefficient $\pinte$ in case $\ppers$ is treated as a nuisance parameter as well.

As argued in the the discussion following Corollary~\ref{corollary:StructuralLimitExperiment_M}, we conjecture that the nuisance parameter $\lppers$ cannot be dealt with using invariance arguments. Various alternative methods to deal with nuisance parameters in testing problems have been used in the literature. In relation to the predictive regression model at hand, we mention the Bonferroni method (\citet{CavanaghElliotStock1995} and \citet{CampbellYogo2005}); tests based on a conditional unbiasedness condition (\citet{JanssonMoreira2006}); and tests based on a numerically calculated Approximate Least Favorable Distribution (ALFD) as more recently proposed in \citet{EMW2015}. All these techniques apply to the Gaussian likelihood ratio statistic in Remark~\ref{rem:GaussianInvariant}.

These approaches have different advantages and disadvantages. \citet{CampbellYogo2005} proposes a modified Bonferroni method to eliminate the nuisance parameter $\lppers$, leading to a simple yet more powerful test than the \citet{CavanaghElliotStock1995} test. However, as pointed out by \citet{Phillips2014}, inference based on Bonferroni bounds can be severely undersized when the predictor is ``far away'' from being a unit root process ($\ppers<<1$). In such a case, confidence intervals obtained by inverting the test may end up having essentially zero coverage probability. \citet{JanssonMoreira2006} develops an approach conditional on specific auxiliary statistics---the terms (only) associated with $\lppers$ in the Gaussian likelihood ratio---and derives an optimal test in the class of conditionally unbiased tests. Nevertheless, such a conditional unbiasedness constraint narrows the considered class and rules out some more powerful tests. Consequently, as shown by the simulation results of \citet{JanssonMoreira2006}, the associated test has relatively low power compared to the \citet{CampbellYogo2005} test under most alternatives.

Recently, \citet{EMW2015} proposes a numerical algorithm to determine an ALFD of the nuisance parameter $\lppers$ to optimize weighted average
power over some compact interval (of $\lppers$). Note that with respect to our parameter of interest $\lpinte$, we consider point-optimal test and do not use weighted powers over a discretized space to avoid the induced computational complexities. On one hand, the ALFD yields an upper bound of the weighted average power for all valid tests. On the other hand, integrating out the likelihood statistic w.r.t.\ the ALFD leads to a ``nearly optimal'' test whose power is close to the upper bound. Moreover, by switching to standard asymptotic approximations in case $\ppers$ appears to be far from unity, the associated test can achieve better size and power performances uniformly for all $\lppers\in (-\infty, 0]$ (e.g., across the parameter space $\ppers\in(-1,1]$).
Therefore, we employ this ALFD approach in the present paper, together with the switching mechanism (see Appendix~\ref{app:Switching}), to our rank-based likelihood statistics in~(\ref{eqn:likelihoodstatistic_g}).\footnote{We expect that other approaches based on likelihood ratios, e.g., the approaches of \citet{CampbellYogo2005} and \citet{JanssonMoreira2006}, will apply here as well. This is because (i) the semiparametric likelihood ratio $\lllr_\mi(b,c)$ in Theorem~\ref{thm:LAQ_M} is as the general version of the Gaussian likelihood ratio in \citet[Lemma 3]{JanssonMoreira2006}, thus when the true density is Gaussian, the former reduces to the latter; (ii) its rank-based proxy in (\ref{eqn:likelihoodstatistic_g}) has the same structure (exponential family); and (iii) the asymptotic behaviors of the associated rank-based processes are known and consistently estimable.} This leads to tests that are of correct size for all relevant $\lppers$ and have good power performance. We confirm these properties by simulations in Section~\ref{sec:MonteCarlo}. 

\subsection{The Approximately Least Favorable Distribution (ALFD) Approach} \label{subsec:ALFD}
\noindent In Section~\ref{sec:PowerEnvelope} we used invariance arguments to reduce the predictive regression testing problem towards log-likelihood ratio of the form~(\ref{eqn:LAQ_M}) where $\lpinte$ is the parameter of interest to be tested and $\lppers$ is a nuisance parameter. We briefly outline, in the present section, how the Approximate Least Favorable Distribution approach in \citet{EMW2015} works in our setting.
 
Rewrite the log-likelihood ratio of the maximal invariant $\mi$ in Theorem~\ref{thm:LAQ_M} as
\begin{align*}
\lllr_\mi(\lpinte,\lppers)
 =
\lpinte\stat_1+\lppers\stat_2-\frac{1}{2}\left((\lpinte,\lppers)\Jf(\lpinte,\lppers)^\trans-\lppers^2\right)\stat_3-\frac{1}{2}\lppers^2\stat_4,
\end{align*}
where
\begin{align} \label{eqn:sufficientstatistics_f}
\stat_1
 &=
\int_0^1\We(\s)\rd\Bfy(\s), ~~~
\stat_2
  =
\int_0^1\We(\s)\rd\Bfx(\s) + \We(1)\overline{\We},\\
\stat_3
 &=
\overline{\We^2} - \left(\overline{\We}\right)^2  ~~~ {\rm and} ~~~
\nonumber
\stat_4
  =
\overline{\We^2}.
\end{align}
One can thus consider the four-dimensional sufficient statistic $\stat:=\left(\stat_1,\stat_2,\stat_3,\stat_4\right)$. For notational simplicity, in the present section, we denote by $F_{\lpinte,\lppers}(\stat)$ the distribution of $\stat$ under $\prob_{b,c}$. The hypothesis of interest is 
\begin{align} \label{eqn:hypothesis_sufficientStat}
\rH_0: \lpinte=0,~~~\lppers\in(-\infty,0] ~~~{\rm versus}~~~ \rH_1:\lpinte>0,~~~\lppers\in(-\infty,0].
\end{align}
Note that, thus, both the null and the alternative hypothesis are composite. We first discuss elimination of the nuisance parameter $\lppers$ under the alternative and, subsequently, its elimination under the null.

To eliminate the nuisance parameter $\lppers$ under the alternative, a standard approach is to consider a so-called weighted average power (see, e.g., \citet{AndrewsPloberger1994})
\begin{align} \label{eqn:weightedaveragepower}
{\rm WAP}(\test)
 =
\int_c\left(\int_\stat\test(\stat)\rd F_{\lpinte,\lppers}(\stat)\right)\rd \Lambda_1(\lppers),
\end{align}   
where $\test$ is some test function for the problem above and $\Lambda_1$ is a probability weighting measure for $\lppers\in(-\infty,0]$. The weighting measure $\Lambda_1$ can be chosen by the researcher and reflects the weights that she assigns to various values of $\lppers$ under the alternative. Due to Fubini's Theorem, we have
\begin{align}
{\rm WAP}(\test)
 =
\int_\stat\test(\stat)\rd
	\int_\lppers F_{\lpinte,\lppers}(\stat)\rd\Lambda_1(\lppers),
\end{align}
which leads to the simple alternative hypothesis $\rH_{1;\Lambda_1}$, under which the distribution of $\stat$ is given by the mixture $F_{\lpinte;\Lambda_1}(\stat)=\int F_{\lpinte,\lppers}(\stat)\rd\Lambda_1(\lppers)$. In this way, the testing problem is reduced to testing $\rH_0$ against $\rH_{1;\Lambda_1}$.

Subsequently, in order to eliminate the nuisance parameter $\lppers$ under the null we proceed as follows. Again we impose a probability weighting measure $\Lambda_0$ for $\lppers$ and introduce the simple null hypothesis, denoted $\rH_{0;\Lambda_0}$, under which the distribution of $\stat$ is given by $F_{\lpinte;\Lambda_0}(\stat)=\int F_{\lpinte,\lppers}(\stat)\rd\Lambda_0(\lppers)$. Now we define the test $\test_{\bar{\lpinte};\Lambda}$ by
\begin{align}
\label{eqn:neymanpearsontest}
\test_{\bar{\lpinte},\Lambda_0}(\stat)
 =
\left\{\begin{array}{l}
1 {\rm~~~if~~~} \rd F_{\bar{\lpinte},\Lambda_1}(\stat) > \cv\rd F_{0,\Lambda_0}(\stat),\\
0 {\rm~~~if~~~} \rd F_{\bar{\lpinte},\Lambda_1}(\stat) \leq \cv\rd F_{0,\Lambda_0}(\stat),
\end{array}
\right.
\end{align}
where the critical value $\kappa$ is chosen to obtain the desired size. By the Neyman-Pearson Lemma, $\test_{\bar{\lpinte},\Lambda_0}$ is point optimal at $\lpinte=\bar{\lpinte}$, for the problem of testing the null $\rH_{0;\Lambda_0}$ against the alternative $\rH_{1;\Lambda_1}$.

The problem of choosing $\Lambda_0$ is, unfortunately, more complicated than that of choosing $\Lambda_1$. The reason is that we want to control the rejection probability of the test, not only under $\rH_{0;\Lambda_0}$, but for all values of $\lppers\in(-\infty,0]$. In general there is no reason to expect that a level-$\siglevel$ test under $\rH_{0;\Lambda_0}$ is of correct size for the entire null hypothesis $\rH_0$. However, for some specific choices of $\Lambda_0$ this statement is true, and such a distribution is called a \emph{least-favorable distribution}; see, e.g.,  \citet{LehmannRomano2005}, Theorem~3.8.1. Formally, a distribution $\Lambda_0^*$ is called \emph{least favorable} if the most powerful level-$\siglevel$ test~(\ref{eqn:neymanpearsontest}) for testing $\rH_{0;\Lambda_0^*}$ against $\rH_{1;\Lambda_1}$ is of the desired size for the (entire) null hypothesis $\rH_0$. Moreover, once more by Theorem~3.8.1 in \citet{LehmannRomano2005}, the test $\test_{\bar{\lpinte},\Lambda_0^*}$ is also point optimal (at $\lpinte=\bar{\lpinte}$) for this problem. A least-favorable distribution $\Lambda_0^*$ exists in most of the usual statistical problems. conditions that ensure this and associated references can be found in Section~3.8 of \citet{LehmannRomano2005}. 

As, in most cases, the least-favorable distribution $\Lambda_0^*$ is not easily obtained, \citet{EMW2015} propose a numerical method to find, what they call, an ``Approximate Least Favorable Distribution'' (ALFD). The ALFD is defined as follows.
\begin{definition}\label{def:eps-ALFD}
An $\epsilon$-ALFD is a probability distribution $\Lambda_0^{*\epsilon}$ over $(-\infty,0]$ satisfying
\begin{itemize}
\item[(i)] the Neyman-Pearson test~(\ref{eqn:neymanpearsontest}) with $\Lambda=\Lambda_0^{*\epsilon}$ and critical value $\kappa=\kappa^*$, i.e., $\test_{\bar{\lpinte},\Lambda_0^{*\epsilon}}$, is of size $\siglevel$ under $\rH_{0;\Lambda_0^{*\epsilon}}$ and has power $\bar{\pi}$ against $\rH_{1;\Lambda_1}$;
\item[(ii)] there exists $\cv^{*\epsilon}$ such that the test~(\ref{eqn:neymanpearsontest}) with $\Lambda=\Lambda_0^{*\epsilon}$ and $\cv=\cv^{*\epsilon}$, $\test_{\bar{\lpinte},\Lambda_0^{*\epsilon}}^\epsilon$, is of level $\siglevel$ under $\rH_0$, and has power of at least $\bar{\pi}-\epsilon$ against $\rH_{1;\Lambda_1}$.
\end{itemize}
\end{definition}
The test $\test_{\bar{\lpinte},\Lambda_0^{*\epsilon}}^{\epsilon}$ (in particular, the ALFD $\Lambda_0^{*\epsilon}$ and the critical value $\cv^{*\epsilon}$) is exactly what we are looking for, once we have set the weights $\Lambda_1$ of interest for the alternative hypothesis. Besides the size control under $\rH_0$, the definition above also ensures that the test $\test_{\bar{\lpinte},\Lambda_0^{*\epsilon}}^{\epsilon}$ enjoys a near-optimality property with a relatively small power loss (less than $\epsilon$).

Note that even for a given (small) value of $\epsilon$, the ALFD $\Lambda_0^{*\epsilon}$ is not necessarily ``close'' to the least favorable distribution $\Lambda_0^*$. Actually, (possibly infinitely) many pairs of $(\Lambda_0^{*\epsilon},\cv^{*\epsilon})$ may satisfy Definition~\ref{def:eps-ALFD}. The details about how to implement the numerical algorithm to determine a pair of $(\Lambda_0^{*\epsilon},\cv^{*\epsilon})$ (henceforth the test $\test_{\bar{\lpinte},\Lambda_0^{*\epsilon}}$) for a small $\epsilon$ can be found in Section~3 and Appendix~A of \citet{EMW2015}. As the nuisance parameter space $\lppers\in(-\infty,0]$ is unbounded, we also need to ``switch'' back to standard test statistics (i.e., in the stationary case) for large values of $|c|$. We provide in Appendix~\ref{app:Switching} the details about our test for the standard part of the limit experiment $\LimitExperiment$.

\subsection{Putting it all together}\label{sec:PuttingItAllTogether}

\noindent Putting everything together, our test for the predictive regression model is based on applying the ALFD approach to the rank-based counterpart (using Proposition~\ref{prop:limitbehavior_BTg}) of the asymptotically point-optimal invariant derived in Theorem~\ref{thm:LAQ_M}.

We thus replace, in the sufficient statistic $\stat=\left(\stat_1,\stat_2,\stat_3,\stat_4\right)$ in~(\ref{eqn:sufficientstatistics_f}), $\We$, $\Bfy$, and $\Bfx$ by $\WTe$, $\BTgy$, and $\BTgx$, leading to the feasible rank-based statistic
\begin{equation}\label{eqn:LLRMSufficientTRankBased}
\stat_{g}^{(T)}
 :=
\Big(\stat_{g,1}^{(T)}, \stat_{g,2}^{(T)}, \stat_{g,3}^{(T)}, \stat_{g,4}^{(T)}\Big),
\end{equation}
where
\begin{align*}
\stat_{g,1}^{(T)}
 &=
\int_0^1\WTe(\s)\rd\BTgy(\s),\\
\stat_{g,2}^{(T)}
 &=
\int_0^1\WTe(\s)\rd\BTgx(\s) + \WTe(1)\int_0^1\WTe(\s)\rd\s,\\
\stat_{g,3}^{(T)}
 &=
\int_0^1\WTe(\s)^2\rd\s - \left(\int_0^1\WTe(\s)\rd\s\right)^2,\\
\stat_{g,4}^{(T)}
 &=
\int_0^1\WTe(\s)^2\rd\s.
\end{align*}
To make the log-likelihood ratio $\lllr_{\mi}$ in~(\ref{eqn:LAQ_M}) fully feasible, we also have to deal with $J_f$. From \citet{KaganLandsman1999} we know that $J_f$ is diagonalized by the Cholesky root of the correlation matrix $\refcorrm$. Therefore, we replace $J_f$ by
\begin{align} \label{eqn:Fisherinformationpicked}
\Jc
 =
\begin{pmatrix}\Jcycy&\Jcycx\\\Jcycx&\Jcxcx\end{pmatrix}
 :=
{\refcorrm^{-\frac12}}^\trans\diag{\{\J_{g_y},\J_{g_x}\}}{\refcorrm^{-\frac12}},
\end{align}
where $\J_{g_y}$ and $\J_{g_x}$ are the Fisher information of the chosen marginal reference densities defined in Assumption~\ref{ass:ApproximateScores}, and $\refcorrm$ is the correlation matrix based on the chosen reference correlation $\refcorrp$, i.e., $\refcorrm := \left(\begin{smallmatrix} 1 & \refcorrp \\ \refcorrp & 1 \end{smallmatrix}\right)$. We recommend to use a consistent estimate of $\rho$ as $\refcorrp$ regarding the power of the test, although any choice of $\refcorrp$ would lead to correct sizes. This leads to our feasible rank-based log-likelihood statistic
\begin{align} \label{eqn:likelihoodstatistic_g}
\lllr_g^{(T)}(\lpinte,\lppers) := \lpinte\stat_{g,1}^{(T)}+\lppers\stat_{g,2}^{(T)}-\frac{1}{2}\left((\lpinte,\lppers)\Jc(\lpinte,\lppers)^\trans-\lppers^2 \right)\stat_{g,3}^{(T)}-\frac{1}{2}\lppers^2\stat_{g,4}^{(T)},
\end{align}
of which the limit is given by the proposition below.
\begin{proposition} \label{prop:weakconvergence_lllr_g}
Suppose $\et=(\eyt,\ext)^\trans$ are i.i.d.\ innovations with density $f\in\F$. Let $g_y$ and $g_x$ be reference densities that satisfy Assumption~\ref{ass:ApproximateScores} and fix the reference correlation $\refcorrp$. Then, for $\lpinte\in\SR$ and $\lppers\in(-\infty,0]$, under $\law_{0,0,\eta;f}$, we have
\begin{align*}
\lllr_g^{(T)}(\lpinte,\lppers) \wto \lllr_g(\lpinte,\lppers),
\end{align*}
where 
\begin{align} \label{eqn:LAQ_Mg}
\lllr_g(\lpinte,\lppers) := \lpinte\stat_{g,1}+\lppers\stat_{g,2}-\frac{1}{2}\left((\lpinte,\lppers)\Jc(\lpinte,\lppers)^\trans-\lppers^2 \right)\stat_{g,3}-\frac{1}{2}\lppers^2\stat_{g,4}
\end{align}
with
\begin{align*}
\stat_{g,1}
 &=
\int_0^1\We(\s)\rd\Bgy(\s), ~~~
\stat_{g,2}
 =
\int_0^1\We(\s)\rd\Bgx(\s) + \We(1)\overline{\We}, \\
\stat_{g,3}
 &=
\overline{\We^2} - \left(\overline{\We}\right)^2  {\rm ~~~and~~~}
\stat_{g,4}
 =
\overline{\We^2}.
\end{align*}
\end{proposition}
We omit the proof of Proposition~\ref{prop:weakconvergence_lllr_g} since it directly follows from the weak convergences in (\ref{eqn:partialsumconvergence}) and (\ref{eqn:partialsumconvergence_Bg}), the continuous mapping theorem, and the rank-based stochastic integral convergence argument in the proof of Lemma 4.1 of \citet{ZvdAW2016}.

Although not explicit in the above, observe that the statistic $\lllr_g^{(T)}(\lpinte,\lppers)$ in~(\ref{eqn:likelihoodstatistic_g}) still depends on $\sigma_x$ through $\WTe$ defined in~(\ref{eqn:WTedef}). We will simply replace $\sigma_x$ by its sample counterpart below. As long as this estimator is consistent, the continuous mapping theorem shows that this replacement has no asymptotic consequences. The statistic does not depend on $\sigma_y$, but it does depends on the reference correlation $\refcorrp$.

Now, applying the ALFD algorithm to $\lllr_g^{(T)}(\lpinte,\lppers)$, we obtain a distribution $\Lambda^{*\epsilon}_{0,g}$ and critical value $\cv_{g,n}$ such that the test
\begin{align} \label{eqn:test_nonstandard}
\test_{g,n}(\stat_g^{(T)},\refcorrp)=
\left\{
\begin{array}{l}
1 {\rm~~~if~~~} \int \lllr^{(T)}_g(\bar{\lpinte},\lppers)\rd\Lambda_1(\lppers)>\cv_{g,n} \int\lllr^{(T)}_g(0,\lppers)\rd\Lambda_{0,g}^{*\epsilon}(\lppers) \\
0 {\rm~~~if~~~} \int \lllr^{(T)}_g(\bar{\lpinte},\lppers)\rd\Lambda_1(\lppers)<\cv_{g,n} \int\lllr^{(T)}_g(0,\lppers)\rd\Lambda_{0,g}^{*\epsilon}(\lppers) \\
\end{array}
\right.
\end{align}
is of size $\siglevel$. Here $\bar{\lpinte}$ serves as a fixed alternative point for the quasi-likelihood statistic; see \citet{ERS1996}.

In order to get the appropriate critical values of the test, note that we need consistent estimates, under the null, of $\Jg$ and $\J_{fg}$. We need these in order to ensures the feasibility of the numerically determined pair $(\Lambda_0^{*\epsilon},\cv^{*\epsilon})$. In applications $\Jg$ and $\J_{fg}$ can easily be estimated, however, in the Monte Carlo study below we estimate $\Jg$ and $\J_{fg}$ based on the  known.\footnote{A simple consistent estimator for $\J_g$ would be the sample covariance of the rank-based scores $\score_g$ defined in (\ref{eqn:rankbasedscore_g}) and a direct rank-based estimator for $J_{fg}$ can be found in \citet{CassartHallinPain2010}.} This is necessary as we cannot afford to determine a pair $(\Lambda_0^{*\epsilon},\cv^{*\epsilon})$ for each repetition in the simulation. That would be too intensive computationally.

\section{A Monte Carlo Study}\label{sec:MonteCarlo}

\noindent In this section, we explore by Monte Carlo the size and power properties of our test~(\ref{eqn:test_nonstandard}), combined with the switching approach detailed in Appendix~\ref{app:Switching}, (labeled WZ) relative to the Gaussian quasi-likelihood counterpart in \citet{EMW2015} (labeled EMW). From the theoretical results, both tests should enjoy good size properties but the WZ test should exhibit larger power in case the true innovation distribution is not Gaussian. Under Gaussian innovation distribution, both tests should have similar power.

Section~\ref{sec:MCIID} provides simulations under the predictive regression model studied formally in this paper. Section~\ref{sec:MCCH} provides results of our test under conditional heteroskedasticity. Finally, Section~\ref{sec:MCEstimatedG} provides results when the reference density used in the test is estimated.

\subsection{Simulations under maintained i.i.d.\ assumption}\label{sec:MCIID}

\noindent We simulate the model~(\ref{eqn:model_1})--(\ref{eqn:model_2}) with $\pcons=2$, $\sigma_y = 3$, $\sigma_x=3$, and $\rho=-0.5$. All results reported in this section are based on 10,000 replications. 

For the ALFD approach, we choose a discrete weighting distribution $\Lambda_1$ in~(\ref{eqn:weightedaveragepower}) where each of the 57 points
\begin{align*}
\lppers\in\{0,-0.25^2,-0.5^2,\dots,-14^2\}
\end{align*}
of the support have equal weight. The same 57 points are also as the support of $\Lambda_{0}^{*\epsilon}$. For the test statistic in~(\ref{eqn:test_nonstandard}), we choose a fixed alternative $\bar{\lpinte}=B(1.645)$ where the power is about $50\%$. For the reference correlation $\refcorrp$, we use the simple sample correlation of $\hat\epsilon^y_t$ and $\hat\epsilon^x_t$ under the null, where $\hat\epsilon^y_t = y_t - \sum_{t=1}^{T}y_t$ and $\hat\epsilon^x_t$ is the residual of the regression of $x_t$ on $x_{t-1}$. 

We present the power curves in two ways. The first presentation follows \citet{EMW2015}: We let the local nuisance parameter $\lppers$ (which governs the persistence of the predictor) take 21 values $\lppers\in\{0,-10,-20,\dots,-200\}$. And to have roughly similar power for each value of $\lppers$, we transform the parameter $\lpinte$ by
\begin{align} \label{eqn:betaform}
\lpinte = B(\betaform) = \betaform \sqrt{\frac{-2c+6}{1-\rho^2}}, ~~{\rm for}~~ \lppers<0.
\end{align}
Alternatives for $\pinte$ are now characterized by different values of $\betaform$. The null hypothesis $\rH_0$ corresponds to $\delta=0$, and we let the parameter of interest $\lpinte$ take three alternatives: $\delta\in\{1,2,3\}$. Secondly, we present power curves where we fix the nuisance parameter $\lppers = -25$ and plot the rejection rates for $\delta \in [0, 6]$. The significance level $\siglevel$ is chosen to be $5\%$ in all cases.

\begin{figure}[!htb] 
\centering
\includegraphics[width=15cm]{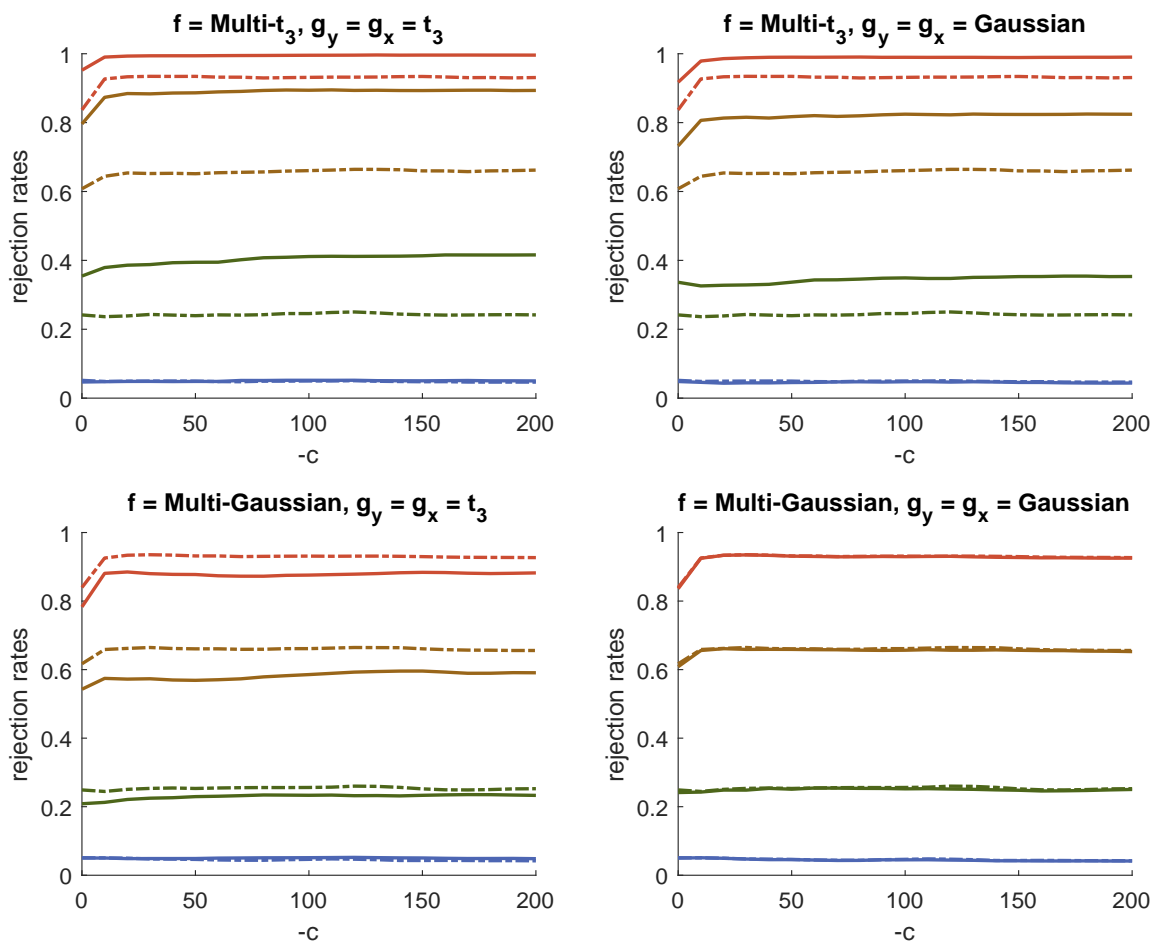}
\caption{Rejection rates of the WZ test (solid lines) and the EMW test (dashed lines) for different values of $\delta$ = 0, 1, 2, and 3, corresponding to lines in blue, green, brown, and red, respectively. For all the four cases, the correlation is $-0.5$. The sample size is 2,000. }
\label{figure:rho05_t2000}
\end{figure}

\begin{figure}[!htb] 
\centering
\includegraphics[width=15cm]{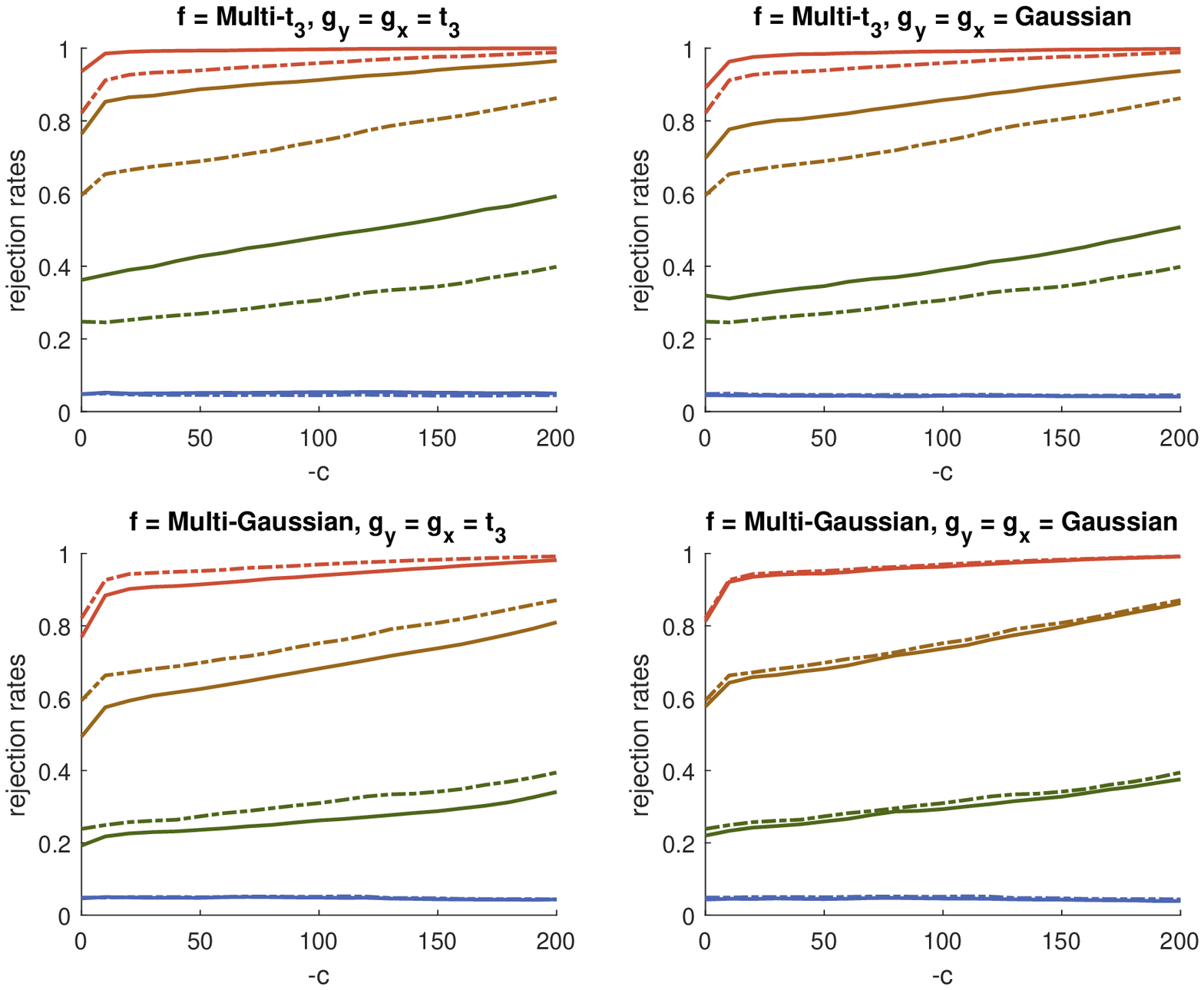}
\caption{Rejection rates of the WZ test (solid lines) and the EMW test (dashed lines) for different values of $\delta$ = 0, 1, 2, and 3, corresponding to lines in blue, green, brown, and red, respectively. For all the four cases, the correlation is $-0.5$. The sample size is 200.}
\label{figure:rho05_t200}
\end{figure}

\begin{figure}[!htb] 
\centering
\includegraphics[width=15cm]{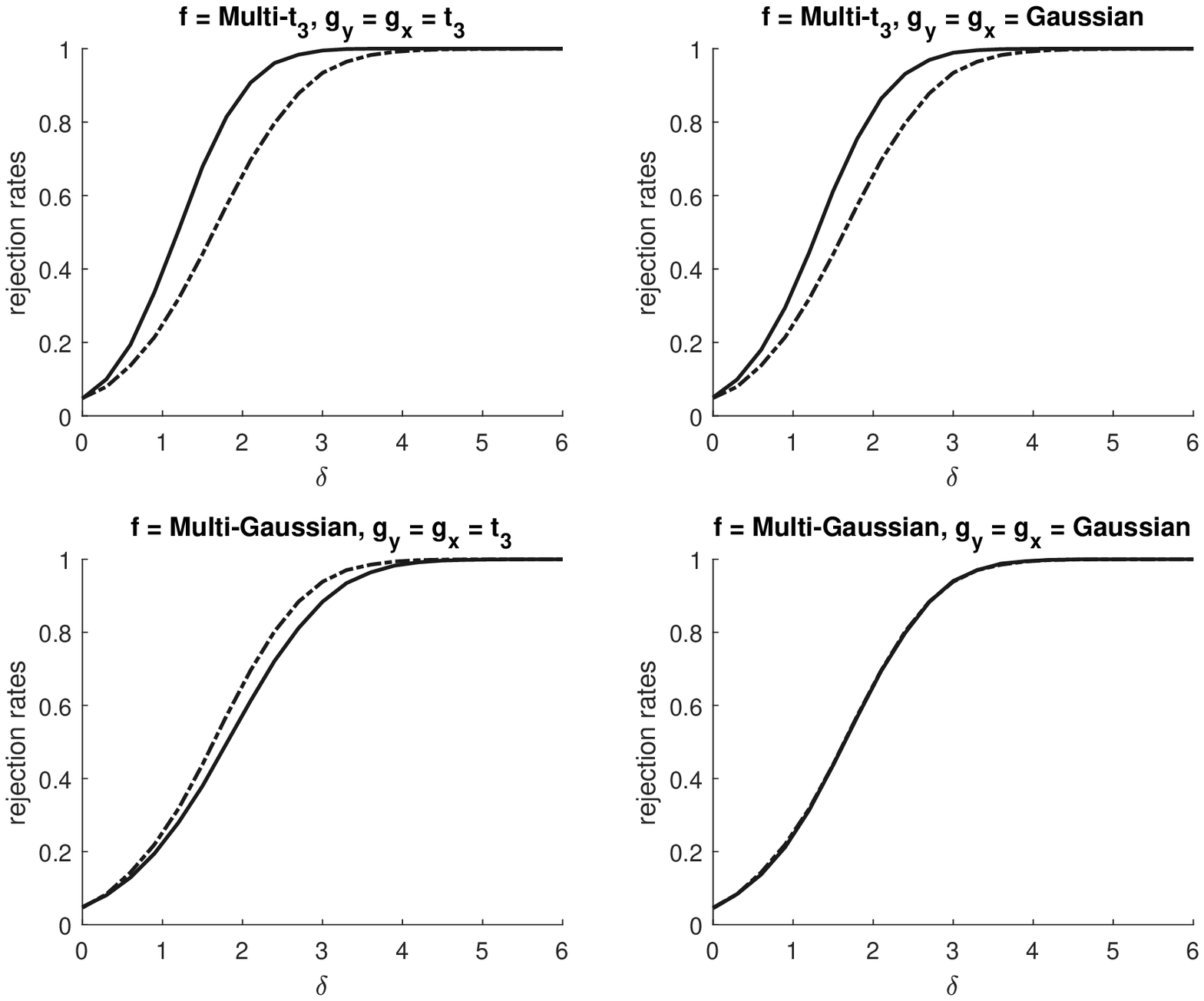}
\caption{Rejection rates of the WZ test (solid lines) and the EMW test (dashed lines) for fixed value of $c = -25$ and different values of $\delta \in [0, 6]$. For all the four cases, the correlation is $-0.5$. The sample size is 2,000.}
\label{hfigure:rho05_t2000}
\end{figure}

\begin{figure}[!htb] 
\centering
\includegraphics[width=15cm]{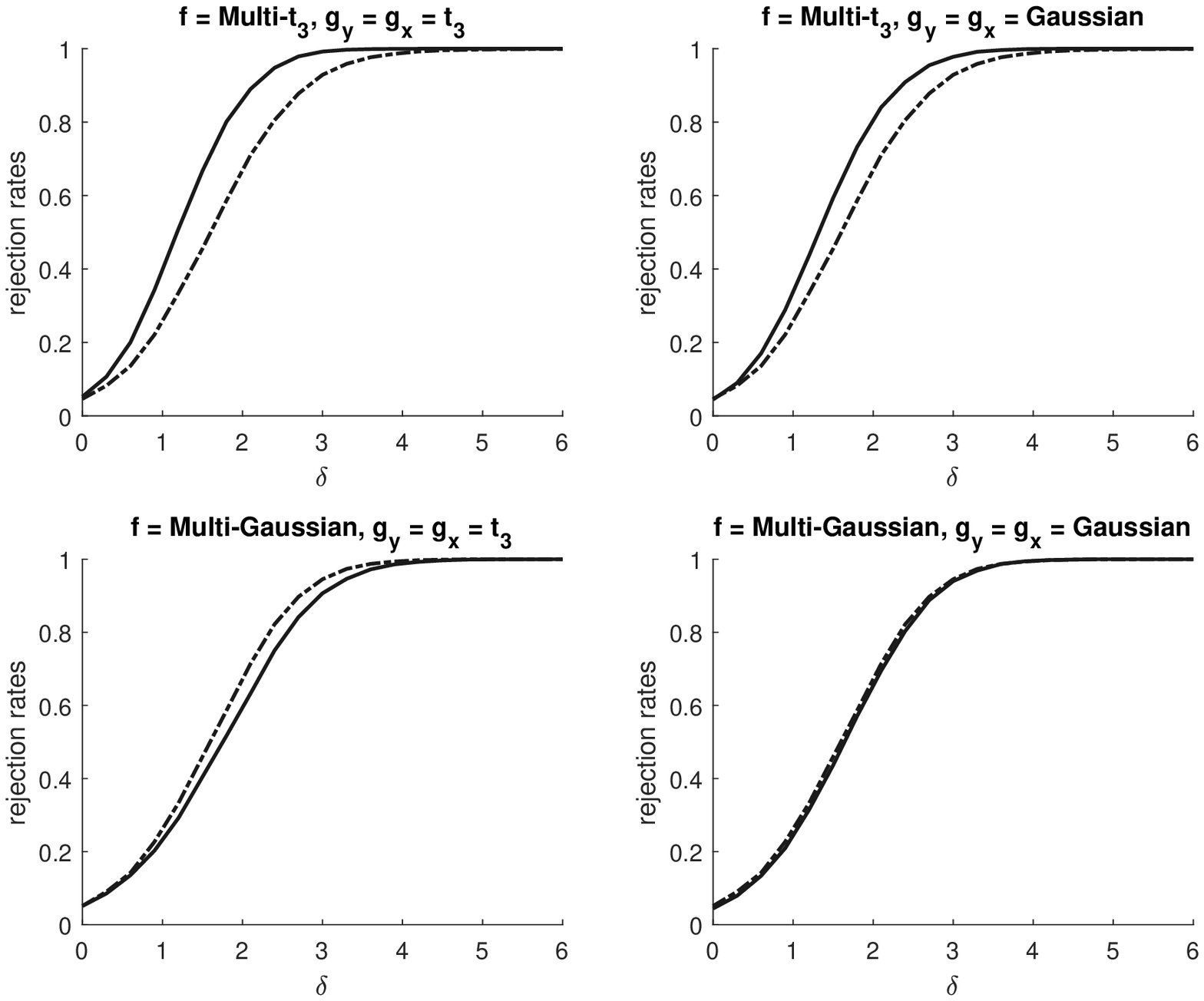}
\caption{Rejection rates of the WZ test (solid lines) and the EMW test (dashed lines) for fixed value of $c = -25$ and different values of $\delta \in [0, 6]$. For all the four cases, the correlation is $-0.5$. The sample size is 2,000.}
\label{hfigure:rho05_t200}
\end{figure}

In Figure~\ref{figure:rho05_t2000}, we reports the large-sample ($T=2,000$) size and power properties of our rank-based WZ test and the EMW test, for different combinations of the true density $f$ and the marginal reference densities $g_y$ and $g_x$. 
The upper-left subplot reports the case where $f$ is a multivariate $t_3$ density, while $g_y$ and $g_x$ are both univariate $t_3$ densities. Both the EMW test and the WZ test are of correct size for all chosen values of $\lppers$. Under the alternative hypothesis (i.e., for $\delta\in\{1,2,3\}$), the WZ test is more powerful than the EMW test. Taking the alternative $\delta=2$ as example, for most values of $\lppers$, the power of the EMW test is about $65\%$ while the WZ test attains about $90\%$ power. In the upper-right subplot, we keep $f$ unchanged and let $g_y$ and $g_x$ both be Gaussian. Both tests provide correct size and, again, the WZ test is more powerful than the EMW test. However, compared to the upper-left subplot, we observe that the WZ test suffers a small power loss when choosing reference densities that are further away from the true ones. When $f$ is Gaussian, the WZ test with Gaussian marginal reference densities shares almost the same size and power performances as the EMW test, as shown by the bottom-left subplot. The bottom-right subplot presents the case when $f$ is Gaussian, while the marginal reference densities $g_y$ and $g_x$ are univariate $t_3$. In this case, the WZ test is less powerful than the EMW test. In practice, we may want to avoid this power loss by pre-testing the residuals under the null hypothesis. We study this in Section~\ref{sec:MCEstimatedG}.

\begin{figure}[!htb] 
\centering
\includegraphics[width=15cm]{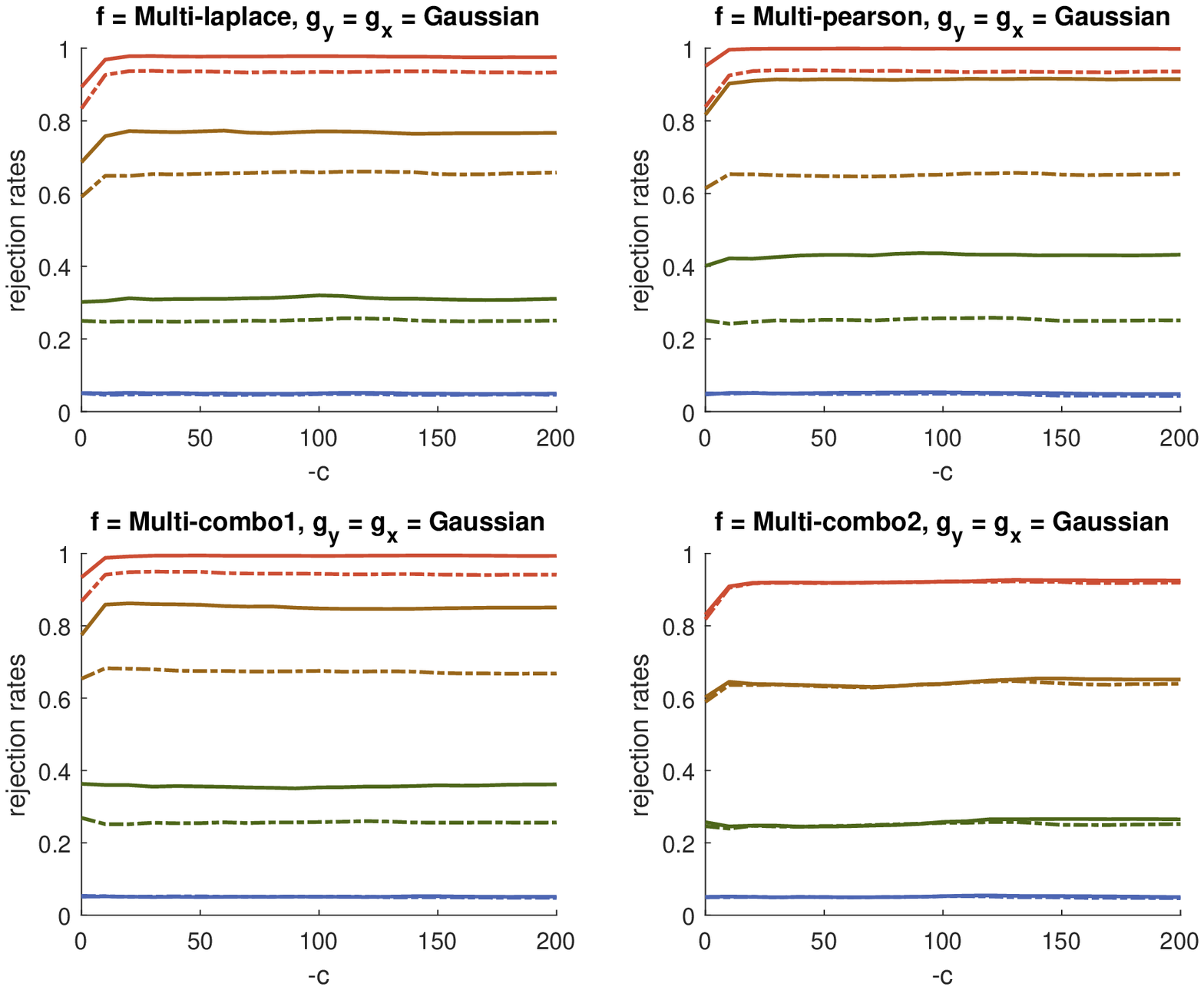}
\caption{Rejection rates of the WZ test (solid lines) and the EMW test (dashed lines) for different values of $\delta$ = 0, 1, 2, and 3, corresponding to lines in blue, green, brown, and red, respectively. For all the four cases, $\rho = -0.5$ and $T = 2,000$.}
\label{figure:others_rho05_t2000}
\end{figure}

\begin{figure}[!htb] 
\centering
\includegraphics[width=15cm]{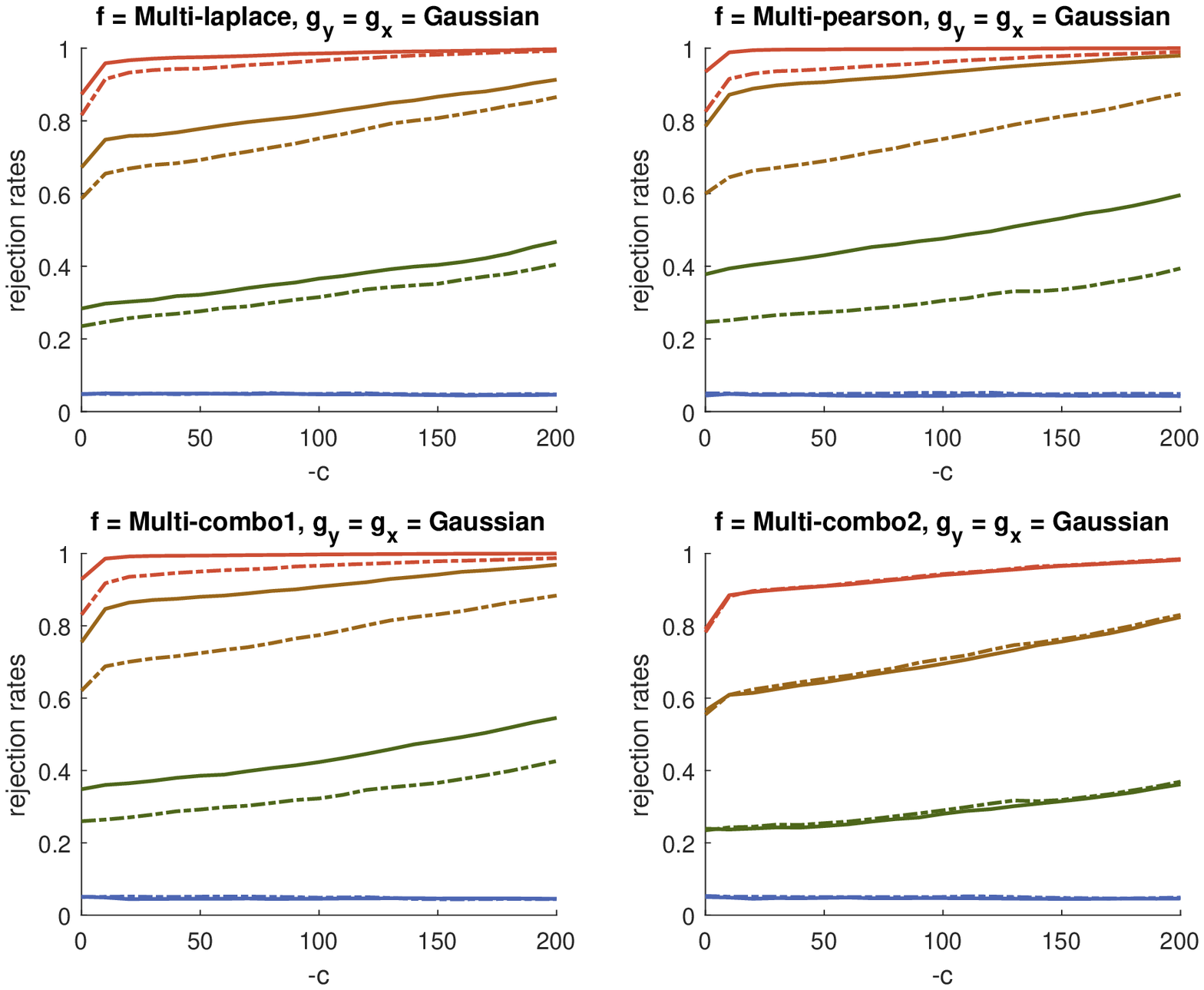}
\caption{Rejection rates of the WZ test (solid lines) and the EMW test (dashed lines) for different values of $\delta$ = 0, 1, 2, and 3, corresponding to lines in blue, green, brown, and red, respectively. For all the four cases, $\rho = -0.5$ and $T = 200$.}
\label{figure:others_rho05_t200}
\end{figure}

Actually, one can always use Gaussian reference densities as a conservative choice, which is based on a (numerical) \citet{ChernoffSavage1958} result --- keeping the marginal reference densities $g_y$ and $g_x$ Gaussian, the WZ test is always more powerful than the EMW test when $f$ is non-Gaussian, and it works as well as the EMW test when $f$ is Gaussian. A formal proof of this result in LABF-type experiments is still an open question, but we show that this property holds in some more simulations. In Figure~\ref{figure:others_rho05_t2000}, we fix $g_y$ and $g_x$ to be Gaussian, and choose four different multivariate innovation distributions: (i) Gaussian copula with Laplace marginal distributions (top-left, labeled Multi-Laplace); (ii) Multivariate Pearson distribution with skewness 3 and kurtosis 36 (top-right, labeled Multi-Pearson); (iii) Gaussian copula with $t_3$ distribution for the first dimension and Gaussian distribution for the second dimension (bottom-left, labeled Multi-combo1); and (iv) $t_3$ copula with Gaussian for the first dimension and $t_3$ for the second dimension (bottom-right, labeled Multi-combo2). These simulations support the Chernoff-Savage result and also show that the further away the true distribution is from Gaussian, the more power can be gained by the WZ test. Moreover, case (iv) in the bottom-right subplot shows that actually the power we gain by the WZ test is from the innovation of the first dimension, $\eyt$. When the distribution of $\eyt$ is Gaussian, we do as well as the EMW test. We conjecture that inference for $\pinte$ in the predictive regression model (\ref{eqn:model_1})-(\ref{eqn:model_2}) is adaptive with respect to the marginal density of $\ext$, when $\ppers$ is eliminated by the ALFD approach in \cite{EMW2015}.

In Figure~\ref{hfigure:rho05_t2000} and Figure~\ref{hfigure:rho05_t200}, we present the powers of the WZ test and the EMW test (for fixed $c = -15$ and for $\delta \in [0, 6]$) under the same settings as in Figure~\ref{figure:rho05_t2000} and Figure~\ref{figure:rho05_t200}, respectively. The results show the power gain of the WZ test over the EMW test uniformly for all alternative values.

We also provide some small-sample ($T=200$) results for both tests in Figure~\ref{figure:rho05_t200} and Figure~\ref{figure:others_rho05_t200} (the small-sample counterparts of Figure~\ref{figure:rho05_t2000} and Figure~\ref{figure:others_rho05_t2000}, respectively). The conclusions are similar: both tests are of good size (all around $4.5\%$) using the same combinations of $\Lambda_{0}^{*\epsilon}$ and $\cv_g$. The WZ test still gains considerable power in the case of non-Gaussian densities, though the gain is slightly smaller than in the large-sample case. This once more shows the additional information present, when supported by the application at hand, of an i.i.d.-ness assumption on the innovations.  Appendix~\ref{app:additional_simulation_results} provides additional simulation results  in Figure~\ref{hfigure:others_rho05_t2000} and Figure~\ref{hfigure:others_rho05_t200} for Figure~\ref{figure:others_rho05_t2000} and Figure~\ref{figure:others_rho05_t200} using $c = -15$ and $\delta \in [0, 6]$, respectively.

Finally, we repeat the simulations of Figure~\ref{figure:rho05_t2000} and Figure~\ref{figure:rho05_t200}, but for $\rho = -0.9$, in Figure~\ref{figure:rho09_t2000} and Figure~\ref{figure:rho09_t200} respectively in Appendix~\ref{app:additional_simulation_results}. These simulations confirm our previous conclusions about the WZ test: correct sizes, power gain under non-Gaussian $f$, the Chernoff-Savage result, and decent small-sample performances.


\subsection{Simulations under conditional heteroskedasticity}\label{sec:MCCH}

\noindent 
In many (financial) applications the maintained assumption of i.i.d.\ innovations will not be satisfied. We therefore study, by simulation, the behavior of the tests when the innovations exhibit conditional heteroskedasticity. The tests are identical to those in the previous sections, thus not adapted to deal with possible heteroskedasticity.

Keeping everything else unchanged, we replace the i.i.d.\ innovations $(\eyt, \ext)^\trans$, by a univariate GARCH(1,1) model (i) for $\eyt$ only; or (ii) for both $\eyt$ and $\ext$ in the data generating process. Formally, we choose
\begin{align*}
\eyt &= \sqrt{1-\rho^2}\e_{1,t} + \rho\e_{2,t},  \\
\ext &=\e_{2,t},
\end{align*}
where, for case~(ii), $\e_{1,t}$ and $\e_{2,t}$ are independently generated by the GARCH(1,1) model
\begin{align*}
\e_{j,t} &= \nu_{j,t} \sqrt{h_{j,t}},  \\
h_{j,t} &= 1 + 0.07\e_{j,t-1} + 0.92h_{j,t-1},
\end{align*}
for $j = 1,2$, where $\nu_{j,t}$'s are i.i.d.\ innovations. For case (i), we let $\e_{2,t}$ be i.i.d.\ and independent of $\e_{1,t}$. The joint density of $\nu_{1,t}$ and $\nu_{2,t}$ is denoted by $f$. The GARCH parameters are chosen based on common empirical findings.

\begin{figure}[!htb] 
\hspace{-30mm}
\includegraphics[width=20cm]{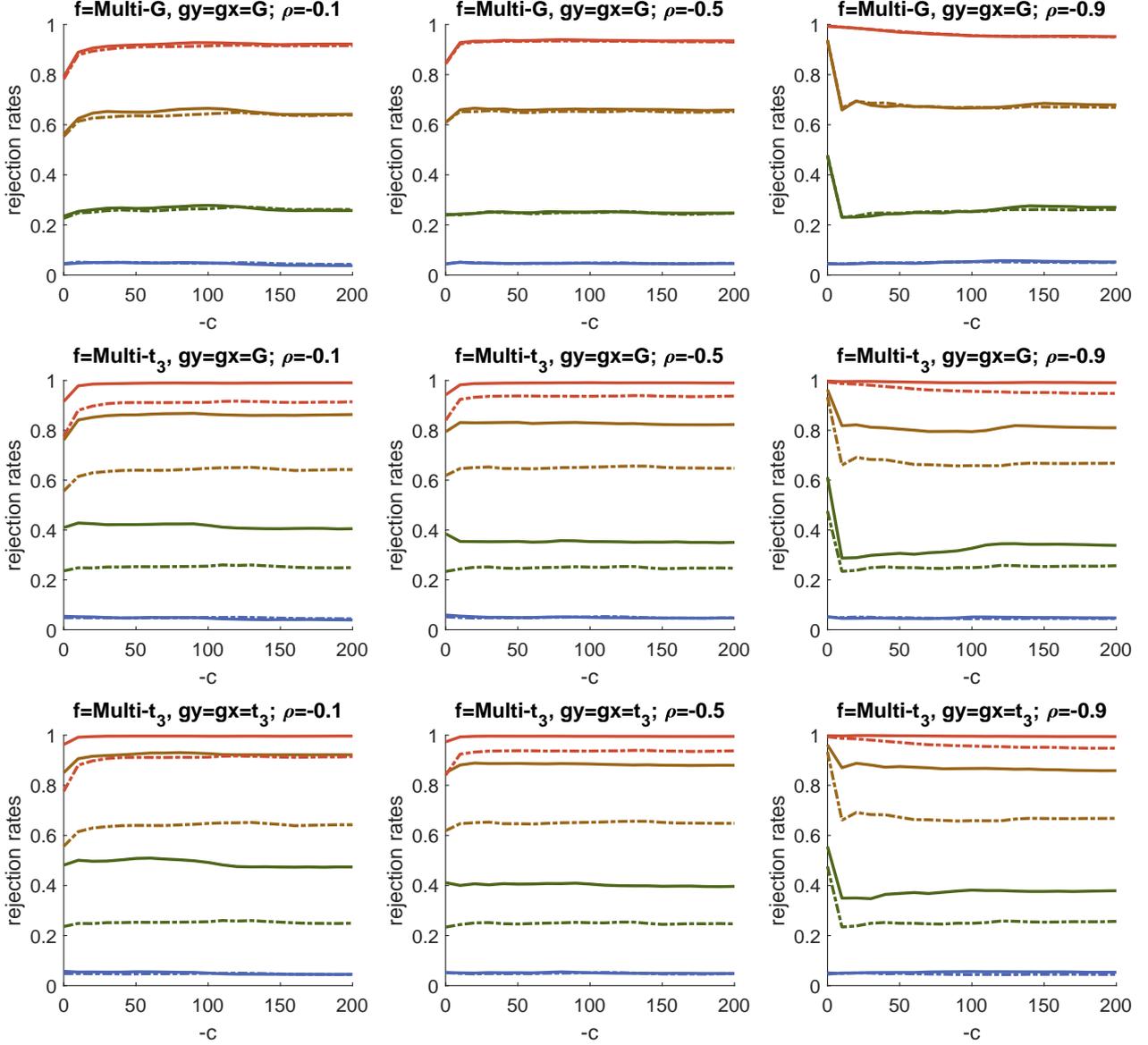}
\caption{Rejection rates of the WZ test (solid lines) and the EMW test (dashed lines) for different values of $\delta$ = 0, 1, 2, and 3, corresponding to lines in blue, green, brown, and red, respectively, under \emph{heteroskedasticity}. For all the four cases, $T = 2,000$.}
\label{figure:garch_t2000}
\end{figure}

\begin{figure}[!htb] 
\hspace{-30mm}
\includegraphics[width=20cm]{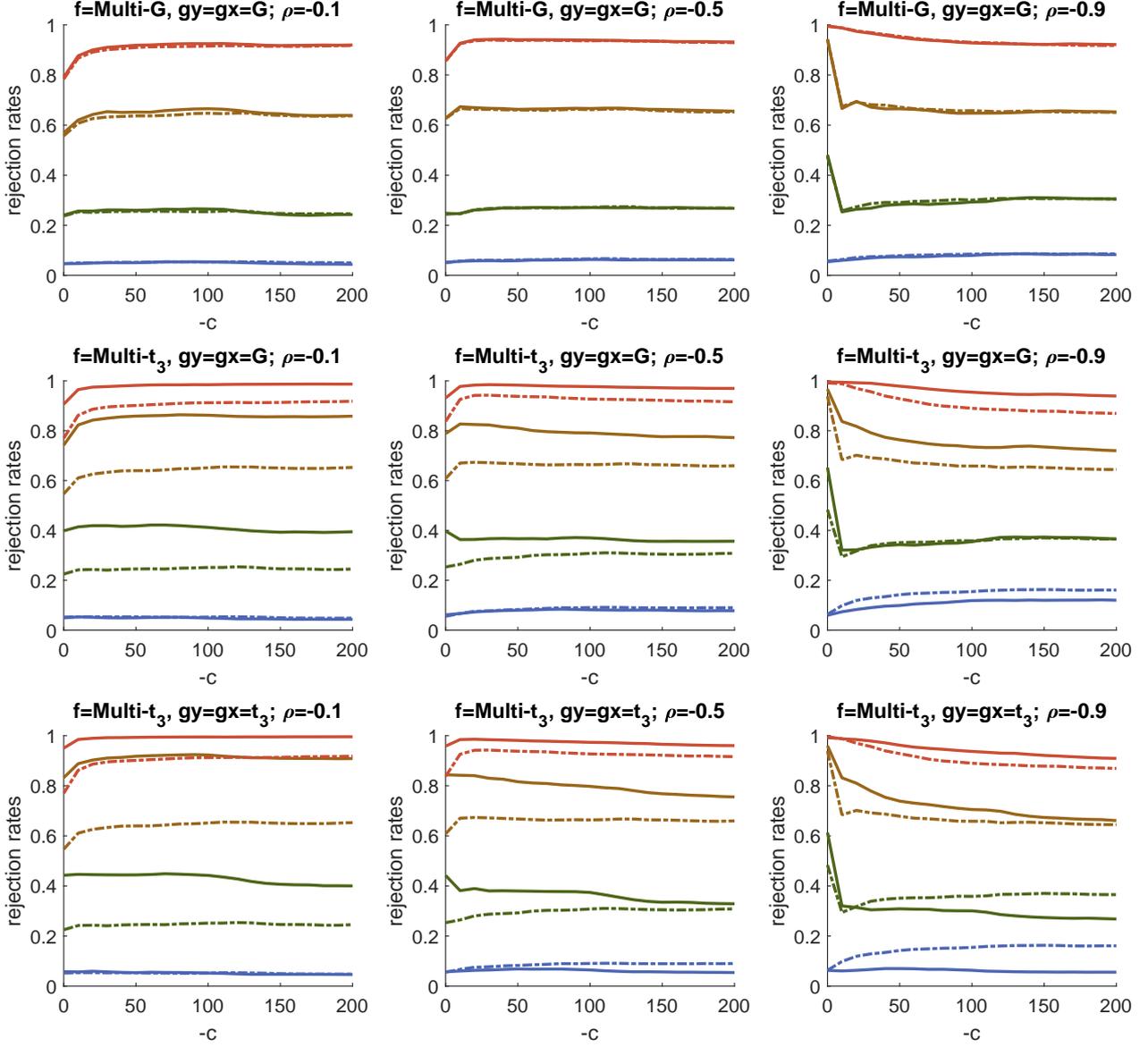}
\caption{Rejection rates of the WZ test (solid lines) and the EMW test (dashed lines) for different values of $\delta$ = 0, 1, 2, and 3, corresponding to lines in blue, green, brown, and red, respectively, under \emph{heteroskedasticity}. For all the four cases, $T = 2,000$.}
\label{figure:garchyx_t2000}
\end{figure}

In Figure~\ref{figure:garch_t2000}, we present case (i) where only the innovations of the response variable, $\ey$, exhibit conditional heteroskedasticity, while the predictor innovations $\ex$ are still i.i.d. We show results for three density combinations as mentioned in the title of each subplot and three different values for the correlation of innovations ($\rho = -0.1$, $-0.5$, and $-0.9$). In all nine cases, we find that both the EMW and WZ tests still have decent sizes, i.e., heteroskedasticity appearing only in $\ey$ will not affect their size performances much. In terms of power, the WZ test outperforms the EMW test under the $t_3$ distribution, and both tests have similar powers under Gaussianity. In addition, when $\ey$ is exhibits more heteroskedasticity (i.e., when $\rho$ is close to 0), the WZ test gains more power as heteroskedasticity pushes the unconditional innovation distribution further away from Gaussianity.

Figure~\ref{figure:garchyx_t2000} presents the results for case (ii) where both innovations, $\ey$ and $\ex$, are heteroskedastic and correlated as modeled above. When $\rho$ is close to zero, the size distortion becomes smaller, while for larger (absolute) values of $\rho$, both tests become more oversized, especially under heavy-tailed innovation distribution. The WZ test suffers less size distortion than the EMW test under $t_3$ distributions and, using $t_3$ reference marginal densities, the size distortion bcomes even smaller (see the bottom panel).

The small-sample counterparts of Figure~\ref{figure:garch_t2000} and Figure~\ref{figure:garchyx_t2000} with $T = 200$ are provided in Appendix~\ref{app:additional_simulation_results}. We draw conclusions similar to the i.i.d.\ case in Section~\ref{sec:MCIID}. Additionally, we find that, when both $\ey$ and $\ex$ are heteroskedastic and their correlation is close to $-1$, both the EMW and WZ tests are less over-sized in the small-sample case. 

These conclusions above also apply to other GARCH settings with different value chosen for parameters. These simulation results are available upon request.

\subsection{Simulations under estimated reference density}\label{sec:MCEstimatedG}

\noindent In this section, we provide simulation results for the WZ test based on nonparametrically estimated reference densities, i.e., $g_y = \hat{f}_y$ and $g_x = \hat{f}_x$, under the i.i.d.\ setting as in Section~\ref{sec:MCIID}.

\begin{figure}[!htb] 
\hspace{-20mm}
\includegraphics[width=18cm]{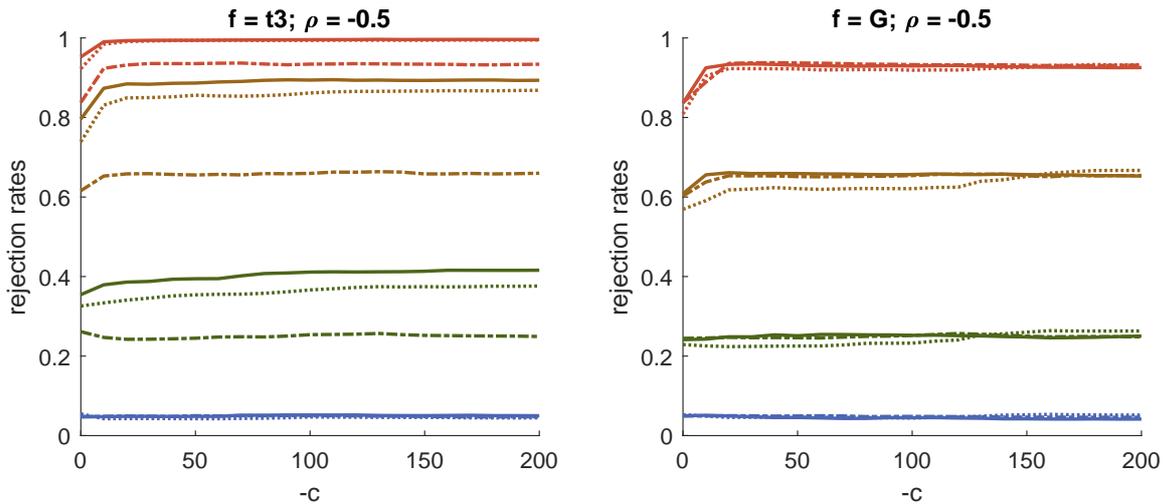}
\caption{Rejection rates of the WZ test (solid lines), the EMW test with Student-$t_3$ marginal reference densities (dashed lines), and the EMW test with nonparametrically estimated density $\hat{f}$ (dotted lines) for different values of $\delta$ = 0, 1, 2, and 3, corresponding to lines in blue, green, brown, and red, respectively. For all the four cases, $T = 2,000$.}
\label{figure:rho05_t2000_np}
\end{figure}

\begin{figure}[!htb] 
\hspace{-20mm}
\includegraphics[width=18cm]{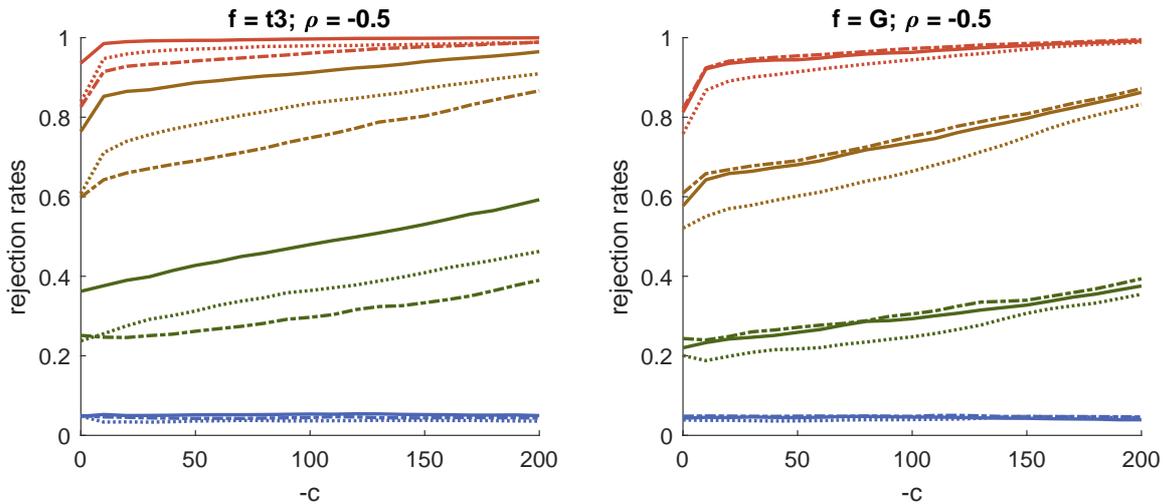}
\caption{Rejection rates of the WZ test (solid lines), the EMW test with Student-$t_3$ marginal reference densities (dashed lines), and the EMW test with nonparametrically estimated density $\hat{f}$ (dotted lines) for different values of $\delta$ = 0, 1, 2, and 3, corresponding to lines in blue, green, brown, and red, respectively. For all the four cases, $T = 200$.}
\label{figure:rho05_t200_np}
\end{figure}

In Figure~\ref{figure:rho05_t2000_np}, we compare the WZ test with $g_y = \hat{f}_y$ and $g_x = \hat{f}_x$ (dotted lines) with the EMW test (dashed lines) and with as the WZ test using correctly specified reference marginal densities (solid lines). When both the true and the reference densities are Gaussian (right plot), we see that all three tests perform similarly with decent size and power properties. When the true innovation distribution is Student-$t_3$, all three tests control the sizes well, while in terms of power, both WZ tests outperform the Gaussian-based EMW test. The WZ test with estimated reference densities suffers a small efficiency loss due to the nonparametric estimation.

Figure~\ref{figure:rho05_t200_np} provides the small-sample results under the same setting but with sample size $T = 200$. In general, the smaller sample leads to lower size and power for the WZ test with estimated reference densities relative to the large-sample case. But again, when $f$ is heavy-tailed, it can be more powerful than the EMW test.

\section{Conclusion}\label{sec:Conclusions}

\noindent In this paper, we show that there is significant statistical information, when supported by the application at hand, in a maintained assumption of serially independent innovations in a predictive regression model. We exploit this information by deriving the (maximal) invariance structures in the associated limit experiment.

Specifically, we first derive the maximal invariant in the (structural) limit experiment where the predictor's persistence parameter is assumed to be known. This leads to the semiparametric power envelope for test that are invariant with respect to the innovation density. The associated likelihood ratio thus gives the semiparametric counterparts of the Gaussian sufficient statistics of~\citet{JanssonMoreira2006}. Under non-Gaussianity, larger powers are possible than under Gaussianity; a well-known result in many classical statistical models. To eliminate the predictor's persistence nuisance parameter, we employ the ALFD approach recently proposed in \citet{EMW2015}.

Our analysis naturally leads to statistics based on the bivariate component-wise ranks of the innovations in the model. Our statistics involve a choice of reference densities that is, subject to some mild regularity conditions, largely arbitrary. Irrespective of the choice of reference densities, our test are of correct asymptotic size. Under non-Gaussianity, even with incorrectly specified reference densities, our test have better power properties than existing tests in the literature that are derived under the assumption of Gaussian innovation densities. These alternative tests do not need serially independent innovations and, as a result, we precisely quantify the power improvements possible when such an assumption is supported by the data. Monte Carlo simulations corroborate our asymptotic results and illustrate that the rank-based tests also work well in smaller samples.

\bibliographystyle{asa}
\bibliography{references}

\appendix
\section{Auxiliaries} \label{sec:appendix_A}


\noindent The lemma below shows that the partial sum processes introduced in Section~\ref{subsec:PartialSumProcesses} weakly converge to the associated Brownian motions. Due to the i.i.d.-ness of the innovations, the lemma follows, e.g., from the functional central limit theorem~VIII. 3.33 in \citet{JacodShiryaev2002}. 

\begin{lemma} \label{lem:partialSum}
Let $f\in\F$ and let, with $m\geq 4$, $k_1,\dots,k_{m-3}\in\SN$. Define, with the notation of Section~\ref{subsec:PartialSumProcesses},
\begin{align*}
\mathcal{W}^{(T)} = \big(W^{(T)}_\varepsilon, W_{\score_{f_y}}^{(T)},W_{\score_{f_x}}^{(T)},W_{\fpert_1}^{(T)},\dots,W^{(T)}_{\fpert_{m-3}}\big)^\trans
\end{align*}
and
\begin{align*}
\mathcal{W} = \big(W_\varepsilon, W_{\score_{f_y}}, W_{\score_{f_x}}, W_{\fpert_1},\dots,W_{\fpert_{m-3}}\big)^\trans.
\end{align*}
Then, in $D_{\SR^m}[0,1]$ under $\law_{0,0,0,0;f}$, we have
\begin{align*}
\mathcal{W}^{(T)}
 &\wto
\mathcal{W},\\
\big\langle\mathcal{W}^{(T)}, \mathcal{W}^{(T)}\big\rangle(1)
 &=
\big[\mathcal{W}^{(T)},\mathcal{W}^{(T)} \big](1)+\opone
 =
\var\big(\mathcal{W}(1)\big) + \opone.
\end{align*}
\end{lemma}

\section{Proofs}\label{app:Proofs}

\begin{proof}[Proof of Proposition~\ref{proposition_LAQ}]~\\
\textit{Proof of Part (i):} \\
Suppose $y_t$ and $x_{t-1}$, for $t=1,2,\dots,T$, are generated from~(\ref{eqn:model_1})--(\ref{eqn:model_2}). Then, using the local parameter perturbations~(\ref{eqn:localization_1}), the log-likelihood ratio equals
\begin{align} \label{appeqn:LLR}
\log\frac{\rd\law_{\lpinte,\lppers,\eta;f}}{\rd\law_{0,0,0;f}}
 =
\llr^{(T)}_{I}(\lpinte,\lppers) + \llr^{(T)}_{II}(\lpinte,\lppers,\eta),
\end{align}
where
\begin{align*}
\llr^{(T)}_{I}(\lpinte,\lppers)
 &:=
\sum_{t=1}^{T} \log \frac{f\left(\yt-\frac{\lpinte}{T}\frac{\sigma_y}{\sigma_x}x_{t-1},\Delta x_t-\frac{\lppers}{T}x_{t-1}\right)}{f\left(\yt,\Delta x_t\right)}, \\
\llr^{(T)}_{II}(\lpinte,\lppers,\eta)
 &:=
\sum_{t=1}^{T}\log\left(1+\frac{1}{\sqrt{T}}\sum_{k=1}^{\infty}\eta_k\fpert_{k}\left(\yt-\frac{\lpinte}{T}\frac{\sigma_y}{\sigma_x}x_{t-1},\Delta x_t-\frac{\lppers}{T}x_{t-1}\right)\right).
\end{align*}

We first use Proposition~1 in \citet{HvdAW2015} to prove
\begin{align} \label{appeqn:LLR_I}
\llr^{(T)}_{I}(\lpinte,\lppers)
 = &~
\frac{\lpinte}{T}\sum_{t=1}^T \frac{x_{t-1}}{\sigma_x} \sigma_y\score_{f_y}(\yt, \Delta x_t) + \frac{\lppers}{T}\sum_{t=1}^T x_{t-1} \score_{f_x}(\yt, \Delta x_t)\\
\nonumber
 &~ -\frac{1}{2}\left(\left(\lpinte^2\Jfyfy+\lppers^2\Jfxfx+2\lpinte\lppers\Jfyfx\right)\frac{1}{T^2}\sum_{t=1}^T \frac{x_{t-1}^2}{\sigma_x^2} \right) + \opone.
\end{align}
Assumption~\ref{ass:density_f}\textit{(a)} implies that the density $f$ is differentiable in quadratic mean, i.e., 
\begin{align} \label{eqn:DQM}
\frac{\sqrt{f}(e-w)}{\sqrt{f}(e)}=1+\frac{1}{2}\left[w^\trans\score_{f}(e)+r(e,w)\right],~~e,w\in\SR^2,
\end{align}
where
\begin{align} \label{appeqn:DQMrem1}
\rEf r^2(\et,w)=o(w^2).
\end{align}
In the notation of \citet{HvdAW2015}, we have 
\begin{align*}
LR_{Tt}
 &=
\frac{f\left(\yt-\frac{\lpinte}{T}\frac{\sigma_y}{\sigma_x}x_{t-1},\Delta x_t-\frac{\lppers}{T}x_{t-1}\right)}{f\left(\yt,\Delta x_t\right)},\\
S_{Tt}
 &=
\left(
\frac{1}{T}\frac{x_{t-1}}{\sigma_x}\sigma_y\score_{f_y}(\yt,\Delta x_t),
\frac{1}{T}x_{t-1} \score_{f_x}(\yt, \Delta x_t)\right)^\trans,\\
R_{Tt}
 &=
r(\et,w_{Tt}),
\end{align*}
where $r$ is implicitly defined in~(\ref{eqn:DQM}), $w_{Tt}=\left(-\frac{\lpinte}{T}\frac{\sigma_y}{\sigma_x}x_{t-1},-\frac{\lppers}{T}x_{t-1}\right)^\trans$, and $h_T=\left(\lpinte, \lppers\right)^\trans$. Thus, (\ref{eqn:DQM}) implies
\begin{align*}
LR_{Tt}=\left(1+\frac{1}{2}\left(h_T^\trans S_{Tt}+R_{Tt}\right)\right)^2.
\end{align*} 
To complete the proof of Part~(i), we show that condition $(a)$, $(b)$, $(c)$, and $(d)$ in Proposition~1 of \citet{HvdAW2015} are satisfied.

\textit{Condition (a)}. This is immediate since $h_T=(\lpinte,\lppers)^\trans$ is a constant vector. 

\textit{Condition (b)}. Display (2), $\rE\left[S_{Tt}\big|\mathcal{F}_{T,t-1}\right]=0$ with $\mathcal{F}_{T,s-1}=\sigma\left(\eyt,\ext:~t<s\right)$, follows immediately from the independence of $\et$ and $\mathcal{F}_{T,t-1}$, $\rEf\left[\score_{f_y}(\et)\right]=0$, and $\rEf\left[\score_{f_x}(\et)\right]=0$. The second equation in Display (3) is met as 
\begin{align*}
J_T:=&~\sum_{t=1}^{T}\rE\left[S_{Tt}S_{Tt}^\trans|\mathcal{F}_{T,t-1}\right]\\
=&~\sum_{t=1}^{T}\begin{pmatrix}
\frac{1}{T^2}\frac{x_{t-1}^2}{\sigma_x^2}\Jfyfy&\frac{1}{T^{2}}\frac{x_{t-1}^2}{\sigma_x^2}\Jfyfx\\
\frac{1}{T^{2}}\frac{x_{t-1}^2}{\sigma_x^2}\Jfyfx&\frac{1}{T^2}\frac{x_{t-1}^2}{\sigma_x^2}\Jfxfx
\end{pmatrix}\\
	\wto
J:=&~\begin{pmatrix}
\Jfyfy\int_0^1\We^2(s)\rd s&\Jfyfx\int_0^1\We^2(s)\rd s\\
\Jfyfx\int_0^1\We^2(s)\rd s&\Jfxfx\int_0^1\We^2(s)\rd s
\end{pmatrix}
,
\end{align*}
where the weak convergence follows from a combination of Lemma \ref{lem:partialSum}, Theorem 2.1 in \citet{Hansen1992}, and the continuous mapping theorem. Next we verify the conditional Lindeberg condition (the first equation in Display (3)), which is, for all $\delta>0$,
\begin{align*}
\sum_{t=1}^{T}\rE\left[\left(h_T^\trans S_{Tt}\right)^2\indicator_{\left\{|h_T^\trans S_{Tt}|>\delta\right\}}\big|\mathcal{F}_{T,t-1}\right]=\opone.
\end{align*}
Observe
\begin{align*}
&~\lefteqn{\sum_{t=1}^{T}\rE\left[\left(h_T^\trans S_{Tt}\right)^2\indicator_{\left\{|h_T^\trans S_{Tt}|>\delta\right\}}\big|\mathcal{F}_{T,t-1}\right]}\\
 =&~
\sum_{t=1}^{T}\rE\left[\left(
\frac{\lpinte}{T}\frac{x_{t-1}}{\sigma_x} \sigma_y\score_{f_y}(\yt, \Delta x_t) + \frac{\lppers}{T}x_{t-1} \score_{f_x}(\yt, \Delta x_t)\right)^2
\indicator_{\left\{(h_T^\trans S_{Tt})^2>\delta^2\right\}}\big|\mathcal{F}_{T,t-1}\right]\\
\leq&~
4\sum_{t=1}^{T}\rE\left[\left(\frac{\lpinte}{T}\frac{x_{t-1}}{\sigma_x} \sigma_y\score_{f_y}(\yt, \Delta x_t)\right)^2\indicator_{\left\{4(bx_{t-1}\sigma_y\score_{f_y}(\yt,\Delta x_t))^2>\delta^2T^2\sigma_x^2\right\}}\big|\mathcal{F}_{T,t-1}\right]\\
 & 
\mbox{}+4\sum_{t=1}^{T}\rE\left[\left(\frac{\lppers}{T}x_{t-1} \score_{f_x}(\yt, \Delta x_t)\right)^2\indicator_{\left\{4(cx_{t-1}\score_{f_y}(\yt,\Delta x_t))^2>\delta^2T^2\right\}}\big|\mathcal{F}_{T,t-1}\right] .
\end{align*}
To complete the proof, we just need to show separately, for any given $\delta>0$,
\begin{align*}
&\sum_{t=1}^{T}\rE\left[\left(\frac{\lpinte}{T}\frac{x_{t-1}}{\sigma_x} \sigma_y\score_{f_y}(\yt, \Delta x_t)\right)^2\indicator_{\left\{2|bx_{t-1}\sigma_y\score_{f_y}(\yt,\Delta x_t)|>\delta T\sigma_x\right\}}\big|\mathcal{F}_{T,t-1}\right] = \opone, \\
&\sum_{t=1}^{T}\rE\left[\bigg(\frac{\lppers}{T}x_{t-1} \score_{f_x}(\yt, \Delta x_t)\bigg)^2\indicator_{\left\{2|cx_{t-1}\score_{f_y}(\yt,\Delta x_t)|>\delta T\right\}}\big|\mathcal{F}_{T,t-1}\right] = \opone.
\end{align*}
Using the notation $\zeta(M)=\rEf\left[\left(\lpinte\sigma_y\score_{f_y}(\yt, \Delta x_t)\right)^2\indicator_{\left\{2|\lpinte\sigma_y\score_{f_y}(\yt,\Delta x_t)|>\delta T\right\}}\right]$, we see, for instance, that the left-hand-side of the second term of the previous display is bounded by 
\begin{align*}
\zeta\left(\frac{\delta \sqrt{T}}{\|\We^{(T)}\|_{\infty}}\right)\int_0^1\left(\We^{(T)}(u-)\right)^2\rd u = \opone,
\end{align*}
by a combination of Lemma~\ref{lem:partialSum}, the continuous mapping theorem, and $\zeta(M)\to0$ as $M\to\infty$ (dominated convergence). The same strategy works for the other term.

\textit{Condition (c)}. This condition consists two asymptotic negligibility properties (the Displays~(4) and~(5) in \citet{HvdAW2015}) of the remainder terms $R_{Tt}=r(\et,w_{Tt})$. Recall $w_{Tt}=\left(
-\frac{\lpinte}{T}\frac{\sigma_y}{\sigma_x}x_{t-1},
-\frac{\lppers}{T}x_{t-1}\right)^\trans$, by~(\ref{appeqn:DQMrem1}), we have
\begin{align*}
T \rEf \left[r^2(\et,w_{Tt})|\mathcal{F}_{T,t-1}\right] = \opone,
\end{align*}
which ensures the Display~(4): $\sum_{t=1}^T\rE \left[R_{Tt}^2|\mathcal{F}_{T,t-1}\right]=\opone$.
Display~(5), that is 
\begin{align*}
\sum_{t=1}^{T}\left(1-\rE\left[LR_{Tt}|\mathcal{F}_{T,t-1}\right]\right)=\opone,
\end{align*}
is trivially met by plugging in $LR_{Tt}=\llr_{I}^{(T)}(\lpinte,\lppers)$ to the left-hand-side which gives zero due to the assumed non-negativity of $f$.

\textit{Condition (d)}. This condition is satisfied since $x_0=0$, so that
\begin{align*}
\log LR_{Tt}
=
\log\frac{f\left(\yt-\frac{\lpinte}{T}\frac{\sigma_y}{\sigma_x}x_{t-1},\Delta x_t-\frac{\lppers}{T}x_{t-1}\right)}{f\left(\yt,\Delta x_t\right)} 
=\log \frac{f\left(\yt,\Delta x_t\right)}{f\left(\yt,\Delta x_t\right)} = \opone.
\end{align*}

Subsequently, for the second term of the log likelihood ratio, $\llr^{(T)}_{II}(\lpinte,\lppers,\eta)$, we prove that it equals
\begin{align*}
\frac{\eta^\trans}{\sqrt{T}}\sum_{t=1}^T \sum_{k}\fpert_k(\yt, \Delta x_t)-\frac{1}{2}\left(2a\Jfyb^\trans\eta + \left(2b\Jfyb^\trans\eta+2c\Jfxb^\trans\eta\right)\frac{1}{T^{3/2}}\sum_{t=1}^{T}\frac{x_{t-1}}{\sigma_x} + \eta^\trans\eta\right)+\opone.
\end{align*}
This completes the proof for Part (i). Since we assume that the functions $\fpert_k$, $k\in\SN$, are two times continuously differentiable with bounded derivatives, by a Taylor Series expansion, we have 
\begin{align} \label{appeqn:taylorseries_1}
&\fpert_k\left(\yt-\frac{\lpinte}{T}\frac{\sigma_y}{\sigma_x}x_{t-1},\Delta x_t-\frac{\lppers}{T}x_{t-1}\right) \\
=~&\fpert_k\left(\yt,\Delta x_t\right)-
\frac{\lpinte}{T}\frac{\sigma_y}{\sigma_x}x_{t-1}\dot{\fpert}_{k,y}\left(\yt,\Delta x_t\right) - \frac{\lppers}{T}x_{t-1}\dot{\fpert}_{k,y}\left(\yt,\Delta x_t\right)+ \opone, \notag
\end{align}
where $\dot{\fpert}_{k,y}$ and $\dot{\fpert}_{k,y}$ are the first-order derivatives of $\dot{\fpert}_{k}$ with respect to the first and second argument, respectively. In this equality, higher-order terms are omitted since the second-order derivatives, denoted by $\ddot{\fpert}_{k,yy}$, $\ddot{\fpert}_{k,yx}$, and $\ddot{\fpert}_{k,xx}$, are bounded, i.e., there exists a real number $M$, such that $\left|\ddot{\fpert}_{k,yy}\right|<M$, $\left|\ddot{\fpert}_{k,yx}\right|<M$, and $\left|\ddot{\fpert}_{k,xx}\right|<M$. Therefore,
\begin{align*}
&\frac{1}{\sqrt{T}}\sum_{t=1}^{T}\left(\frac{\lpinte}{T}\frac{\sigma_y}{\sigma_x}x_{t-1}\right)^2\ddot{\fpert}_{k,yy}\left(\yt,\Delta x_t\right)\\
<~&\frac{1}{\sqrt{T}}\sum_{t=1}^{T}\left(\frac{\lpinte}{T}\frac{\sigma_y}{\sigma_x}x_{t-1}\right)^2M \\
\wto~& \frac{1}{\sqrt{T}}\left(ab\sigma_y^2\int_0^1\We(s)\rd s + \lpinte^2\sigma_y^2\int_0^1\We^2(s)\rd s\right)M = O_{\rP}\left(\frac{1}{\sqrt{T}}\right) = \opone,
\end{align*}
and similar results hold for other higher order terms of $\ddot{\fpert}_{k,yx}$ and $\ddot{\fpert}_{k,xx}$. Also, using $\log(1+x)=x-\frac{1}{2}x^2+O(x^3)$, we have
\begin{align}
&~\llr^{(T)}_{II}(\lpinte,\lppers,\eta) \label{appeqn:LLR_II} \\ \notag
=&~\frac{1}{\sqrt{T}}\sum_{t=1}^{T}\sum_{k=1}^{\infty}\eta_k\fpert_{k}\left(\yt-\frac{\lpinte}{T}\frac{\sigma_y}{\sigma_x}x_{t-1},\Delta x_t-\frac{\lppers}{T}x_{t-1}\right)\\ \notag
&-\frac{1}{2}\frac{1}{T}\sum_{t=1}^{T}\left[\sum_{k=1}^{\infty}\eta_k\fpert_{k}\left(\yt-\frac{\lpinte}{T}\frac{\sigma_y}{\sigma_x}x_{t-1},\Delta x_t-\frac{\lppers}{T}x_{t-1}\right)\right]^2+\opone\\ \notag
=&~\frac{1}{\sqrt{T}}\sum_{t=1}^{T}\sum_{k=1}^{\infty}\eta_k\left[\fpert_k\left(\yt,\Delta x_t\right)-\frac{\lpinte}{T}\frac{\sigma_y}{\sigma_x}x_{t-1}\dot{\fpert}_{k,y}\left(\yt,\Delta x_t\right) - \frac{\lppers}{T}x_{t-1}\dot{\fpert}_{k,x}\left(\yt,\Delta x_t\right)\right] \\ \notag
&-\frac{1}{2}\frac{1}{T}\sum_{t=1}^{T}\left[\sum_{k=1}^{\infty}\eta_k\fpert_{k}\left(\yt-\frac{\lpinte}{T}\frac{\sigma_y}{\sigma_x}x_{t-1},\Delta x_t-\frac{\lppers}{T}x_{t-1}\right)\right]^2+\opone\\ \notag
=&~\frac{1}{\sqrt{T}}\sum_{t=1}^{T}\sum_{k=1}^{\infty}\eta_k\left[\fpert_k\left(\yt,\Delta x_t\right)-\frac{\lpinte}{T}\frac{x_{t-1}}{\sigma_x}\Jfybk - \frac{\lppers}{T}\frac{x_{t-1}}{\sigma_x}\Jfxbk\right]-\frac{1}{2}\sum_{k=1}^{\infty}\eta_k^2+\opone\\ \notag
=&~\frac{\eta^\trans}{\sqrt{T}}\sum_{t=1}^{T}\sum_{k}\fpert_k(\yt,\Delta x_t)-a\Jfyb^\trans\eta-\left(\lpinte\Jfyb^\trans\eta+\lppers\Jfxb^\trans\eta\right)\frac{1}{T^{3/2}}\sum_{t=1}^{T}\frac{x_{t-1}}{\sigma_x}-\frac{1}{2}\eta^\trans\eta+\opone \notag.
\end{align}
The third equality follows from Lemma~\ref{lem:partialSum}, $\rEf\big[\dot{\fpert}_{k,y}(\yt,\Delta x_t)\big]=\int_{\SR^2}\dot{\fpert}_{k,y}(e)f(e)\rd e=\fpert_{k}(e)f(e)\big|_{\SR^2}-\int_{\SR^2} \fpert_{k}(e)\frac{\dot{f}_y}{f}(e)\rd e=J_{f_y \fpert_k}$, $\rEf\big[\dot{\fpert}_{k,x}(\yt,\Delta x_t)\big]=J_{f_x \fpert_k}$, the assumption $\rEf\left[\fpert_{k}^2(e)\right]=1$, and $\rEf\left[\fpert_i(e)\fpert_j(e)\right]=0$ when $i\neq j$. 

Putting together (\ref{appeqn:LLR_I}) and (\ref{appeqn:LLR_II}) completes the proof of the LAQ result in Part (i).

\textit{Proof of Part (ii):} 
The proof for this part follows immediately from the Functional Central Limit Theorem (see, e.g., Lemma~\ref{lem:partialSum} and Theorem 2.4 in \citet{ChanWei1988}). The convergence of integrals as $\int_0^1\We(s)\rd\Wfy(s)$ needs an additional argument as it does not follow automatically from Lemma~\ref{lem:partialSum}. The argument is identical to that in the proof of Proposition~3.2 in \citet{ZvdAW2016}.

\textit{Proof of Part (iii):}
Taking the expectation of $\exp\lllr(\lpinte,\lppers,\eta)$ under $\prob_{0,0,0}$ will directly lead to the result. 
\end{proof}

\begin{proof}[Proof of Theorem~\ref{thm:MaximalInvariant}]
\noindent The proof follows from the definition of the \textit{maximal invariant} in Section~6.2 of \citet{LehmannRomano2005}, which, in terms of the present problem, is: $\mi$ is called maximally invariant with respect to $\transgroup_\eta$ if (i) it is invariant, and if (ii) the equality $M(\We,\Wb)=M(\widetilde{W}_{\e},\widetilde{W}_\fpert)$, with the mapping $M$ defined in Section~\ref{subsec:MaximalInvariant}, implies that $(\We,\Wb)$ can be transformed into $(\widetilde{W}_{\e},\widetilde{W}_\fpert)$ with some transformation $\transfor_\eta\in\transgroup_\eta$. Since (i) is trivially met, the proof is complete by establishing~(ii).

Suppose, indeed, $M(\We(\s),\Wb(\s))=M(\widetilde{W}_{\e}(\s),\widetilde{W}_\fpert(\s))$, $\s\in[0,1]$. Then
\begin{align*}
\We(\s)=\widetilde{W}_\e(\s) {\rm~~and~~} \Bb(\s)=\widetilde{B}_\fpert(\s),
\mbox{ for all }s\in[0,1].
\end{align*}
This in turn implies, for all $s\in[0,1]$,
\begin{align*}
\We(\s)-\widetilde{W}_\e(\s)=0 {\rm~~and~~} \Wb(\s)-\widetilde{W}_\fpert(\s)=\lppers_\transfor\s
\end{align*}
with $\lppers_\transfor=\Wb(1)-\widetilde{W}_\fpert(1)\in\SR$. This shows that $(\We,\Wb)$ can indeed be transformed to $(\widetilde{W}_{\e},\widetilde{W}_\fpert)$ by the transformation $\transfor_\eta\in\transgroup_\eta$ with $\eta=\lppers_\transfor$. Thus condition (ii) is verified and the proof is complete.
\end{proof}

\begin{proof}[Proof of Theorem~\ref{thm:LAQ_M}]
Observe that we can decompose the central sequence $\cs(\lpinte,\lppers,\eta)$ in~(\ref{eqn:LLRlimit}) as
\begin{equation}
\cs(\lpinte,\lppers,\eta) = \cs_{\mi}(\lpinte,\lppers) + \cs_{\indp}(\lpinte,\lppers,\eta),
\end{equation}
with
\begin{equation}
\cs_{\indp}(\lpinte,\lppers,\eta)
 =
\overline{\We}\left(\lpinte\Jfyb+\lppers\Jfxb\right)^\trans\Wb(1) + \eta^\trans\Wb(1).
\end{equation}
Under $\prob_{0,0,0}$, $\Wb(1)$ is independent of $\mi$ while $\overline{\We}$ is measurable with respect to $\mi$. As a result, under $\prob_{0,0,0}$ and conditionally on $\mi$, $\cs_{\indp}(\lpinte,\lppers,\eta)$ is normally distributed with mean zero and variance $\left|\overline{\We}\left(b\Jfyb+c\Jfxb\right)+\eta\right|^2$. As $\cs_{\mi}$ and $\qt$ are obviously $\mi$-measurable, we find
\begin{align*}
\rElim\left[\frac{\rd\prob_{\lpinte,\lppers,\eta}}{\rd\prob_{0,0,0}}|\mi\right]
 &=
\rElim\left[\exp\left(\cs_{\mi}(\lpinte,\lppers)+\cs_{\indp}(\lpinte,\lppers,\eta)-\frac12\qt(\lpinte,\lppers,\eta)\right)|\mi\right]\\
 &=
\exp\left(\cs_{\mi}(\lpinte,\lppers)-\frac12\qt(\lpinte,\lppers,\eta)\right)
	\rElim\left[\exp\cs_{\indp}(\lpinte,\lppers,\eta)|\mi\right]\\
 &=
\exp\left(\cs_{\mi}(\lpinte,\lppers)-\frac12\qt_{\mi}(\lpinte,\lppers)\right).
\end{align*}
This completes the proof.
\end{proof}

\begin{proof}[Proof of Proposition~\ref{prop:limitbehavior_BTg}]
The proposition is somewhat nonstandard as it deals with bivariate component-wise ranks, but otherwise its proof mimics that of Lemma~A.1 in \citet{HvdAW2011}. Tightness of the processes follows exactly as in that lemma, so we only consider convergence of the finite-dimensional distributions. We now from the so-called H\a'{a}jek Representation Theorem (we use it in the version of Theorem~13.5 in \citet{vdVaart2000}), that we may write
\begin{align*}
\lefteqn{\frac{1}{\sqrt{T}}\sum_{t=1}^{\lfloor sT\rfloor}
	\frac{-\dot{g}_y}{g_y}\left(G_y^{-1}\left(\frac{R_{y,t}}{T+1}\right)\right)}\\
 =&~
\frac{1}{\sqrt{T}}\sum_{t=1}^{\lfloor sT\rfloor}
	\frac{-\dot{g}_y}{g_y}\left(G_y^{-1}\left(F_y\left(\eyt\right)\right)\right)
	 -\frac{1}{\sqrt{T}}\sum_{t=1}^{T}
	 \frac{-\dot{g}_y}{g_y}\left(G_y^{-1}\left(F_y\left(\eyt\right)\right)\right)
	 +\opone.
\end{align*}
The equivalent statement holds for the ranks $R_{x,t}$, with $y$ replaced by $x$ everywhere in the above expression. The claim then follows from the functional central limit theorem applied to the partial sums of  $\frac{-\dot{g}_y}{g_y}\left(G_y^{-1}\left(F_y\left(\eyt\right)\right)\right)$ and $\frac{-\dot{g}_x}{g_x}\left(G_x^{-1}\left(F_x\left(\ext\right)\right)\right)$, jointly with $\WTe$ and $\WTf$.
\end{proof}

\begin{proof}[Proof of Corollary~\ref{corollary:StructuralLimitExperiment_Mg}]
The behavior of $\We$ under $\prob_{\lpinte,\lppers,\eta}$ is already given in the structural limit experiment associated to the maximal invariant $\mi$ in Corollary~\ref{corollary:StructuralLimitExperiment_M}. To get the behavior of $\Wg$ under $\prob_{\lpinte,\lppers,\eta}$, first decompose it as
\begin{align*}
\Wg(s) = v \We(s) + A \Wf(\s) + W_{\perp}(s)
\end{align*} 
for some $v\in\SR^{2\times 1}$ and $A\in\SR^{2\times 2}$, where $W_{\perp}$ is a Brownian motion independent of $\We$ and $\Wf$. The appropriate values of $v$ and $A$ satisfy the relation\footnote{Note that in $v$ and $A$ there are 6 unknowns and here there are only two equations, which we only need for this proof. The other four equations are given by the equalities $\cov\left[\Wg(1), \We(1)\right]=\coveg$ and $\cov\left[\Wg(1), \Wg(1)\right]=\J_g$.}
\begin{align*}
\J_{gf} 
&= \cov\left[\Wg(1),\Wf(1)\right] \\
&= \cov\left[v\We(1)+A\Wf(1)+W_\perp(1),\Wf(1)\right] \\
&= v\textbf{e}_1^\trans + A\Jf.
\end{align*}
Then the proof is complete upon noting that, under $\prob_{\lpinte,\lppers,\eta}$, we have 
\begin{align*}
\rd\Wg(s)
 &=
v\rd\We(s)+A\rd\Wf(s)+\rd W_{\perp}(s) \\
 &=
v\left(\lppers\We\rd s+\rd\Ze(s)\right) +
	A\left(\Jf(\lpinte,\lppers)^\trans\We(s)\rd s+\rd \Zf(s)\right) +
	\rd W_{\perp}(s)\\
 &=
(v\textbf{e}_1^\trans+A\Jf)(\lpinte,\lppers)^\trans\We(s)\rd s+
	\left(\rd \Ze(s)+\rd \Zf(s)+\rd Z_{\perp}(s)\right)\\
 &=
\J_{gf}(\lpinte,\lppers)^\trans\We(s)\rd s+\rd\Zg(s).
\end{align*}
\end{proof}

\section{Switching Tests to Standard Case}\label{app:Switching}
\noindent The numerical approach of \citet{EMW2015} needs to discretize the nuisance parameter space under the null hypothesis (and the associated mesh is regarded as the support of $\Lambda_0^{*\epsilon}$). However, in the present case, the null parameter space of $\lppers$ is $(-\infty,0]$, which is unbounded. This complicates the algorithm in terms of computation. To address this issue, \citet{EMW2015} proposes to switch to a standard test when $|\lppers|$ is large enough so that the predictor essentially behaves like a stationary time series. In that case, the problem reduces to a standard test with a stationary regressor. In particular, the authors propose to use a ``switching'' function $\chi=\indicator\{\hat{\lppers}<K\}$ based on some estimator $\hat{\lppers}$ of $\lppers$ and a chosen ``threshold'' $K$ to distinguish the nonstandard situation from the standard one. Then, one can employ the following (combined) test function
\begin{align} \label{eqn:combinedtest}
\test_{n,s,\chi}(\stat)=\chi\test_{s}(\stat) + \left(1-\chi\right)\test_{n}(\stat),
\end{align}
where $\test_s$ is some test for the standard case, and $\test_n$ is the test~(\ref{eqn:test_nonstandard}) for the nonstandard case.
For the standard test $\test_s$, following the argument in the same paper, we use the semiparametric version of the $t$-test
\begin{align} \label{eqn:standardtest_f}
\test_s(\stat) = \indicator\left\{\standardb\big/\sigma_{\standardb}>\cv_s\right\}
\end{align}
with
\begin{align*}
&\standardb=\frac{\stat_1}{\stat_3\Jfyfy} - \frac{\Jfyfx}{\Jfyfy}\standardc, ~~~ \standardc = \frac{\stat_2-(\Jfyfx/\Jfyfy)\stat_1}{\big((\Jfxfx-1)-\Jfyfx^2/\Jfyfy\big)\stat_3+\stat_4}, {\rm ~~~and~}\\
&\sigma_{\standardb} = \sqrt{\frac{1}{\Jfyfy \stat_3}+\left(\frac{\Jfyfx}{\Jfyfy}\right)^2\frac{1}{\big((\Jfxfx-1)-\Jfyfx^2/\Jfyfy\big)\stat_3+\stat_4}}.
\end{align*}
Here $\standardb$ and $\standardc$ are the maximum likelihood estimators of $\lpinte$ and $\lppers$ based on the likelihood ratio in~(\ref{eqn:LAQ_M}). 

The proof of the following lemma can be found in the Supplementary Material of \citet{EMW2015} (Appendix C.4).
\begin{lemma} \label{lem:convergetostandard}
	For $\s\in[0,1]$, let $Z_1(s)$ and $Z_2(s)$ be two independent standard Brownian motions, and $W_1(\s)$ be the associated Ornstein-Uhlenbeck process of $Z_1(\s)$, defined by $\rd W_1(\s)=\lppers W_1(\s)\rd\s+\rd Z_1(\s)$. Define the demeaned process $W_1^{\mu}(\s)=W_1(\s)-\int_0^1 W_1(\s)\rd\s$. Then, as $\lppers\to-\infty$, we have 
	\begin{align} \label{eqn:covergetostandard}
	\begin{pmatrix}
	\sqrt{-2c}\int_0^1 W_1(\s)\rd Z_1(\s) \\
	\sqrt{-2c}\int_0^1 W_1^{\mu}(\s)\rd Z_2(\s) \\
	-2c\int_0^1W_1(\s)^2\rd\s \\
	-2c\int_0^1W_1^{\mu}(\s)^2\rd\s  
	\end{pmatrix} \wto
	\begin{pmatrix}
	z_1\\z_2\\1\\1
	\end{pmatrix},
	\end{align}
	where $z_1$ and $z_2$ are two independent standard normal random variables.
\end{lemma}

\begin{lemma} \label{lem:standardsufficientstatistic}
	Suppose the sufficient statistics $\stat_1, \stat_2, \stat_3, \stat_4$ are defined in (\ref{eqn:sufficientstatistics_f}), where the behavior of $(\We,\Bfy,\Bfx)^\trans$ is described by the limit experiment $\mathcal{E}_{\mi}(f)$ in Corollary~\ref{corollary:StructuralLimitExperiment_M}. Then, under $\prob_{\lppers,\eta}$ and as $\lppers\to-\infty$, we have
	\begin{align*}
	&\sqrt{-2c}\begin{pmatrix}\stat_1+\Jfyfx/2 \\ \stat_2+\Jfxfx/2\end{pmatrix}
	\wto \mathcal{N}\left(
	0
	,\begin{pmatrix}\Jfyfy & \Jfyfx \\ \Jfyfx & \Jfxfx\end{pmatrix}\right),\\
	&-2c\stat_3 \wto 1, {\rm ~~~and~~~} -2c\stat_4 \wto 1.
	\end{align*}
	Subsequently, still under $\prob_{\lppers,\eta}$ and as $\lppers\to-\infty$, we have 
	\begin{align*}
	\standardb/\sigma_{\standardb} \wto \mathcal{N}(0,1).
	\end{align*}
\end{lemma}
\begin{proof}[Proof of Lemma~\ref{lem:standardsufficientstatistic}]
	Note that, in this proof, all convergence results (as $\lppers\to-\infty$) follow immediately from Lemma~\ref{lem:convergetostandard}. 
	
	First, we give the convergence results for $\stat_3$ and $\stat_4$: Recall $\rd\We(\s)=\lppers\We(\s)\rd\s+\rd\Ze(\s)$ for $\s\in[0,1]$ which makes $\We(\s)$ an Ornstein-Uhlenbeck process. Then we have, as $\lppers\to-\infty$,
	\begin{align} \label{eqn:app_S3S4convergence}
	&-2c\stat_3 = -2c\left(\overline{\We^2} - \left(\overline{\We}\right)^2\right)=-2c \int_0^1\We^{\mu}(\s)^2\rd\s \to 1, \\
	&-2c\stat_4 = -2c\overline{\We^2} = -2c\int_0^1\We(\s)^2\rd\s \to 1. \notag
	\end{align}
	
	Next, we give the convergence results of statistics $\stat_1$ and $\stat_2$: To this end, we state first some results derived from Lemma~\ref{lem:convergetostandard}: Define $\We^{\mu}(\s)=\We(\s)-\int_0^1\We(\s)\rd\s$ for $\s\in[0,1]$ and any infinite-dimensional vector $A_1,A_2\in\SR^{\infty\times 1}$, we have  
	\begin{align*}
	\begin{pmatrix}
	{-2c}A_1^\trans\int_0^1\We^{\mu}(\s)\rd\Zb(\s) \\
	{-2c}A_2^\trans\int_0^1\We^{\mu}(\s)\rd\Zb(\s)
	\end{pmatrix}
	\wto
	\mathcal{N}\left(\begin{pmatrix}0\\0\end{pmatrix},\begin{pmatrix}
	A_1^\trans A_1 & A_1^\trans A_2 \\ A_1^\trans A_2 & A_2^\trans A_2\end{pmatrix}\right).
	\end{align*}
	Hence, following the decomposition
	\begin{align*}
	\sqrt{-2c}\stat_1 
	& = \sqrt{-2c}\int_0^1\We(\s)\rd\Bfy(\s) \\
	& = \sqrt{-2c}\int_0^1\We^{\mu}(\s)\rd\Wfy(\s) \\
	& = \sqrt{-2c}\int_0^1\We^{\mu}(\s)\rd\Zfy(\s) + \sqrt{-2c}\times \lppers\Jfyfx\int_0^1\We^{\mu}(\s)\We(\s)\rd\s \\
	& = \sqrt{-2c}\int_0^1\We^{\mu}(\s)\rd\Zfy(\s) - \frac{\sqrt{-2c}}{2}\Jfyfx\left(-2c\int_0^1\We^{\mu}(\s)^2\rd\s\right),
	\end{align*}
	we find
	\begin{align*}
	\sqrt{-2c}\stat_1 + \frac{\sqrt{-2c}}{2}\Jfyfx \wto \mathcal{N}\left(0,\Jfyfy\right).
	\end{align*}
	Similarly, by the decomposition
	\begin{align*}
	\sqrt{-2c}\stat_2 
	& = \sqrt{-2c}\left(\int_0^1\We(\s)\rd\Bfx(\s) + \We(1)\int_0^1\We(\s)\rd\s\right) \\
	& = \sqrt{-2c}\left(\int_0^1\We(\s)\rd\We(\s)+\Jfxb\int_0^1\We^{\mu}(\s)\rd\Wb(\s)\right) \\
	& = \sqrt{-2c}\left(\int_0^1\We(\s)\rd\Ze(\s)+\Jfxb\int_0^1\We^{\mu}(\s)\rd\Zb(\s)\right) \\
	& ~~~~ - \frac{\sqrt{-2c}}{2}\left(-2c\int_0^1\We(\s)^2\rd\s - 2c\Jfxb\Jfxb^\trans\int_0^1\left(\We^{\mu}(\s)\right)^2\rd\s\right),
	\end{align*}
    and $\Jfxfx = 1+\Jfxb\Jfxb^\trans$, we have	
\begin{align*}
	& \sqrt{-2c}\stat_2 + \frac{\sqrt{-2c}}{2}\Jfxfx \wto \mathcal{N}\left(0,\Jfxfx\right).
\end{align*}

The covariance of $\sqrt{-2c}\stat_1$ and $\sqrt{-2c}\stat_1$ is $\Jfyb\Jfxb^\trans=\Jfyfx$. In total, we have
\begin{align} \label{eqn:app_S1S2convergence}
&\sqrt{-2c}\left(\begin{pmatrix}\stat_1\\\stat_2\end{pmatrix} + \frac{1}{2}\begin{pmatrix}\Jfyfx\\\Jfxfx\end{pmatrix}\right) \wto \mathcal{N}\left(\begin{pmatrix} 0 \\ 0 \end{pmatrix},\begin{pmatrix}\Jfyfy & \Jfyfx \\ \Jfyfx & \Jfxfx\end{pmatrix}\right).
\end{align}

	Finally, we show $\standardb/\sigma_{\standardb}\wto\mathcal{N}(0,1)$: Using~(\ref{eqn:app_S1S2convergence}), we find
	\begin{align*} 
	&\sqrt{-2c}\left(\begin{pmatrix}\stat_1\\\stat_2-\frac{\Jfyfx}{\Jfyfy}\stat_1\end{pmatrix} + \frac{1}{2}\begin{pmatrix}\Jfyfx\\\Jfxfx-\frac{\Jfyfx^2}{\Jfyfy}\end{pmatrix}\right) \wto \mathcal{N}\left(\begin{pmatrix} 0 \\ 0 \end{pmatrix},\begin{pmatrix}\Jfyfy & 0 \\ 0 & \Jfxfx-\frac{\Jfyfx^2}{\Jfyfy}\end{pmatrix}\right).
	\end{align*}
	Thus, after some algebra,
	\begin{align*}
	\frac{\standardb}{\sqrt{-2c}}
	& = \frac{\sqrt{-2c}\stat_1}{\Jfyfy(-2c\stat_3)} - \frac{\Jfyfx}{\Jfyfy}\frac{\standardc}{\sqrt{-2c}} \\
	& = \frac{\sqrt{-2c}\stat_1}{\Jfyfy(-2c\stat_3)} - \frac{\Jfyfx}{\Jfyfy}\frac{\sqrt{-2c}\stat_2-(\Jfyfx/\Jfyfy)\sqrt{-2c}\stat_1}{((\Jfxfx-1)-\Jfyfx^2/\Jfyfy)(-2c\stat_3)+(-2c\stat_4)} \\
	&\wto \mathcal{N}\left(0,\left(\frac{1}{\Jfyfy}+\left(\frac{\Jfyfx}{\Jfyfy}\right)^2\frac{1}{\Jfxfx-\Jfyfx^2/\Jfyfy}\right)\right).
	\end{align*}
	Moreover, following~(\ref{eqn:app_S3S4convergence}), we have
	\begin{align*}
	\frac{\sigma_{\standardb}}{\sqrt{-2c}} 
	&= \sqrt{\frac{1}{\Jfyfy(-2c\stat_3)}+\left(\frac{\Jfyfx}{\Jfyfy}\right)^2\frac{1}{((\Jfxfx-1)-\Jfyfx^2/\Jfyfy)(-2c\stat_3)+(-2c\stat_4)}} \\
	&\wto \sqrt{\frac{1}{\Jfyfy}+\left(\frac{\Jfyfx}{\Jfyfy}\right)^2\frac{1}{\Jfxfx-\Jfyfx^2/\Jfyfy}},
	\end{align*}
	which completes the proof.	
\end{proof}

~\\
To introduce the rank-based standard test $\test_s$, we define, in terms of $\stat_{g, 1}$, $\stat_{g, 2}$, $\stat_{g, 3}$ and $\stat_{g, 4}$, the rank-based statistics
\begin{align*}
\standardb_g = \frac{\stat_{g,1} + \refcorrp\stat_{g, 2}}{\stat_{g,3}\Jgy}, ~~~ \standardc_g = \frac{\stat_{g,2}-\refcorrp\stat_{g,1}}{\stat_{g,3}\Jgx},
{~~\rm and~~}
\sigma_{\standardb_g} = \sqrt{\frac{1}{\stat_{g,3}\Jgy}}.
\end{align*}
Note, $\stat_{g,3}=\stat_3$ and $\stat_{g,4}=\stat_4$. Now, $\standardb_g$ and $\standardc_g$ serve as rank-based estimators of  $\lpinte$ and $\lppers$. The following lemma can be regarded as the rank-based version of Lemma~\ref{lem:standardsufficientstatistic}.

\begin{lemma}\label{lem:standardsufficientstatistic_g}
Define the statistic $\stat_g := \left(\stat_{g, 1}, \stat_{g, 2}, \stat_{g, 3}, \stat_{g, 4}\right)$ where $\stat_{g, 1}$, $\stat_{g, 2}$, $\stat_{g, 3}$ and $\stat_{g, 4}$ are introduced in Proposition~\ref{prop:weakconvergence_lllr_g}. Then, under $\prob_{c,\eta}$ and as $\lppers\to-\infty$, we have
	\begin{align}
	\standardb_g/\sigma_{\standardb_g} \wto \mathcal{N}(0,1).
	\end{align}
\end{lemma}
\begin{proof}
	Recall, as $\lppers\to-\infty$, $-2c\stat_3\to 1$ and $-2c\stat_4\to 1$, hence
	\begin{align*}
	\frac{\sigma_{\standardb_g}}{\sqrt{-2c}} 
	\to \frac{1}{\sqrt{\Jgy}}.
	\end{align*}
	Rewrite
	\begin{align*}
	\frac{\standardb_g}{\sqrt{-2c}}
	= \frac{\sqrt{-2c}\left(\stat_{g,1}+\refcorrp\stat_{g,2}\right)}{-2c\stat_{g,3}\Jgy} 
	= \frac{1}{-2c\stat_{g,3}\Jgy} \sqrt{-2c}\int_0^1\We(\s)\rd\By(\s),
	\end{align*}
	where $\By:=\Bgy+\refcorrp\Bgx$. It is not hard to find that, based on the construction in (\ref{eqn:rankbasedscore_g})-(\ref{eqn:rankscorepartialsum}), $\By$ is the limit of the partial-sum process $\frac{1}{\sqrt{T}}\sum_{t=1}^{\lfloor sT \rfloor} \frac{-\dot{g_y}}{g_y}\left(G_y^{-1}\left(\frac{R_{y,t}}{T+1}\right)\right)$. Therefore, under $H_0$, $\By$ is a Brownian bridge. As $\lppers\to-\infty$, by Lemma~\ref{lem:convergetostandard}, we have 
	\begin{align*}
	\sqrt{-2c}\int_0^1\We(\s)\rd\By(\s) = \sqrt{-2c}\int_0^1\We^{\mu}(\s)\rd\Wy(\s) \wto \mathcal{N}(0,\Jgy),
	\end{align*}
where $\Wy$ is the associated Brownian motion of $\By$.	Thus $\frac{\standardb_g}{\sqrt{-2c}} \wto \mathcal{N}(0,\frac{1}{\Jgy})$, which in turn completes the proof.
\end{proof}

~\\
Now we have the standard test  
\begin{align*}
\test_{g,s}(\stat_g,\refcorrp) = \indicator\left\{\standardb_g\big/\sigma_{\standardb_g}>\cv_{g,s}\right\}
\end{align*}
where $\cv_{g,s}$ is the $(1-\siglevel)$-quantile of a standard normal distribution. Similarly, employing the (combined) test as in (\ref{eqn:combinedtest}), we obtain the rank-based test
\begin{align*}
\test_{g,\chi_g}(\stat_g,\refcorrp)=\chi_g\test_{g,s}(\stat_g,\refcorrp) + \left(1-\chi_g\right)\test_{g,n}(\stat_g,\refcorrp),
\end{align*}
where $\chi_g=\indicator\{\standardc_g<K_g\}$. \\

Replacing $\stat_g$ by its finite-sample counterpart $\stat_{g}^{(T)}$, defines the feasible test $\test_{g,\chi_g}(\stat_{g}^{(T)}, \refcorrp)$. In the Monte Carlo study in Section~\ref{sec:MonteCarlo}, following \citet{EMW2015}, we choose $K_g = -130$. 


\clearpage
\section{Additional Simulation Results} \label{app:additional_simulation_results}

\begin{figure}[!htb] 
\centering
\includegraphics[width=15cm]{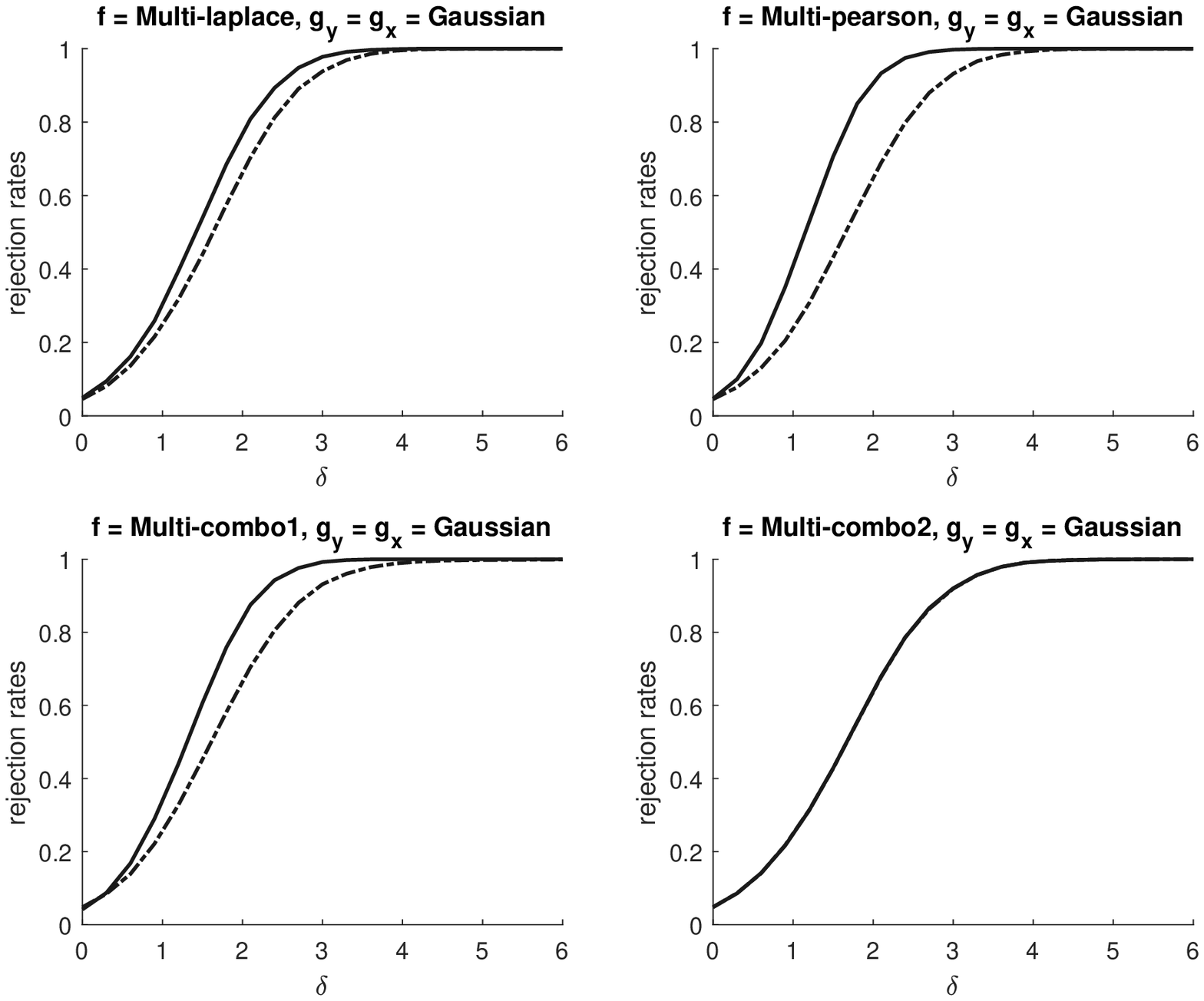}
\caption{Rejection rates of the WZ test (solid lines) and the EMW test (dashed lines) for fixed value of $c = -25$ and different values of $\delta \in [0, 6]$. For four cases, $\rho = -0.9$ and $T = 2,000$.}
\label{hfigure:others_rho05_t2000}
\end{figure}

\begin{figure}[!htb] 
\centering
\includegraphics[width=15cm]{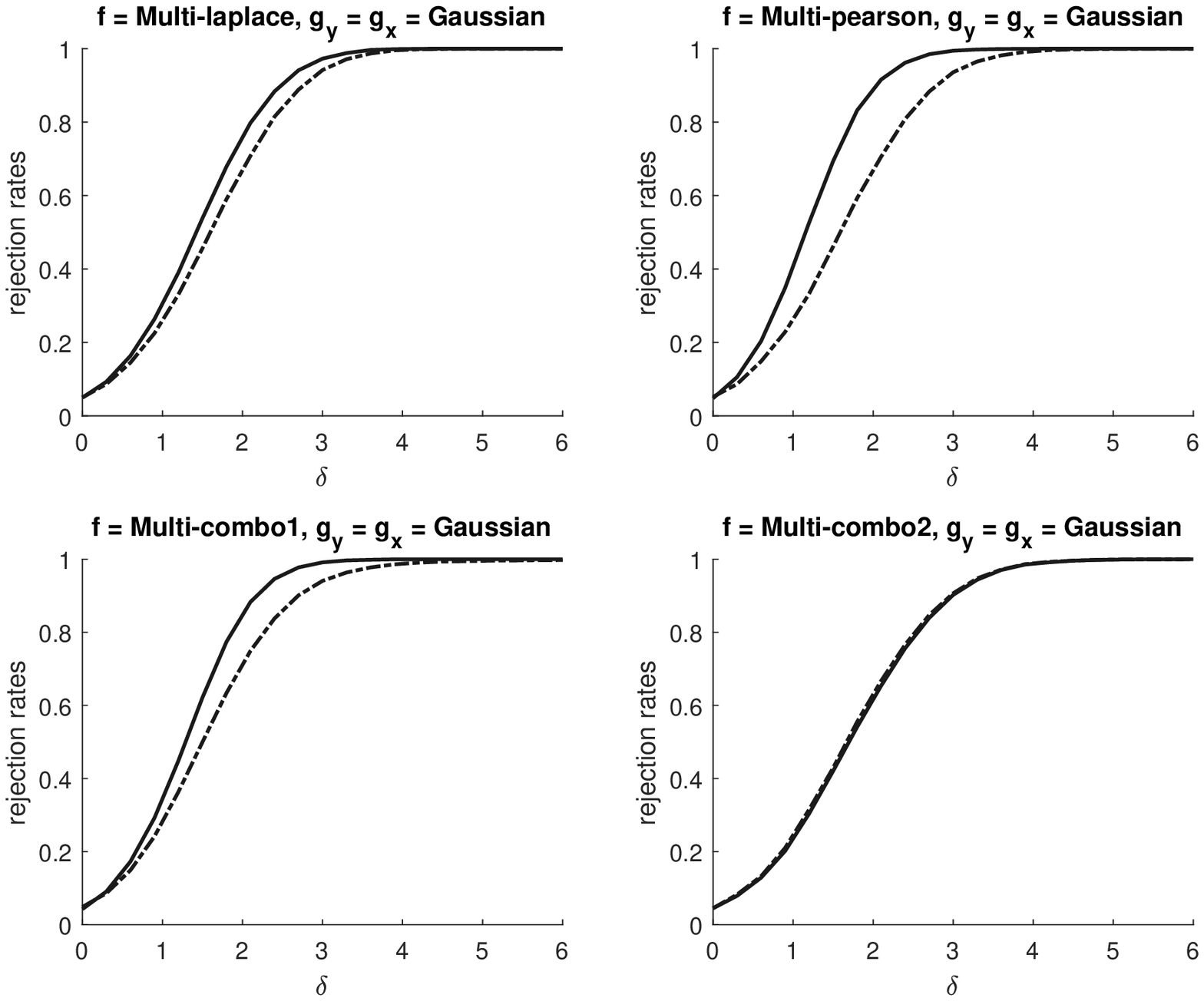}
\caption{Rejection rates of the WZ test (solid lines) and the EMW test (dashed lines) for fixed value of $c = -25$ and different values of $\delta \in [0, 6]$. For four cases, $\rho = -0.5$ and $T = 200$.}
\label{hfigure:others_rho05_t200}
\end{figure}

\begin{figure}[!htb] 
\centering
\includegraphics[width=15cm]{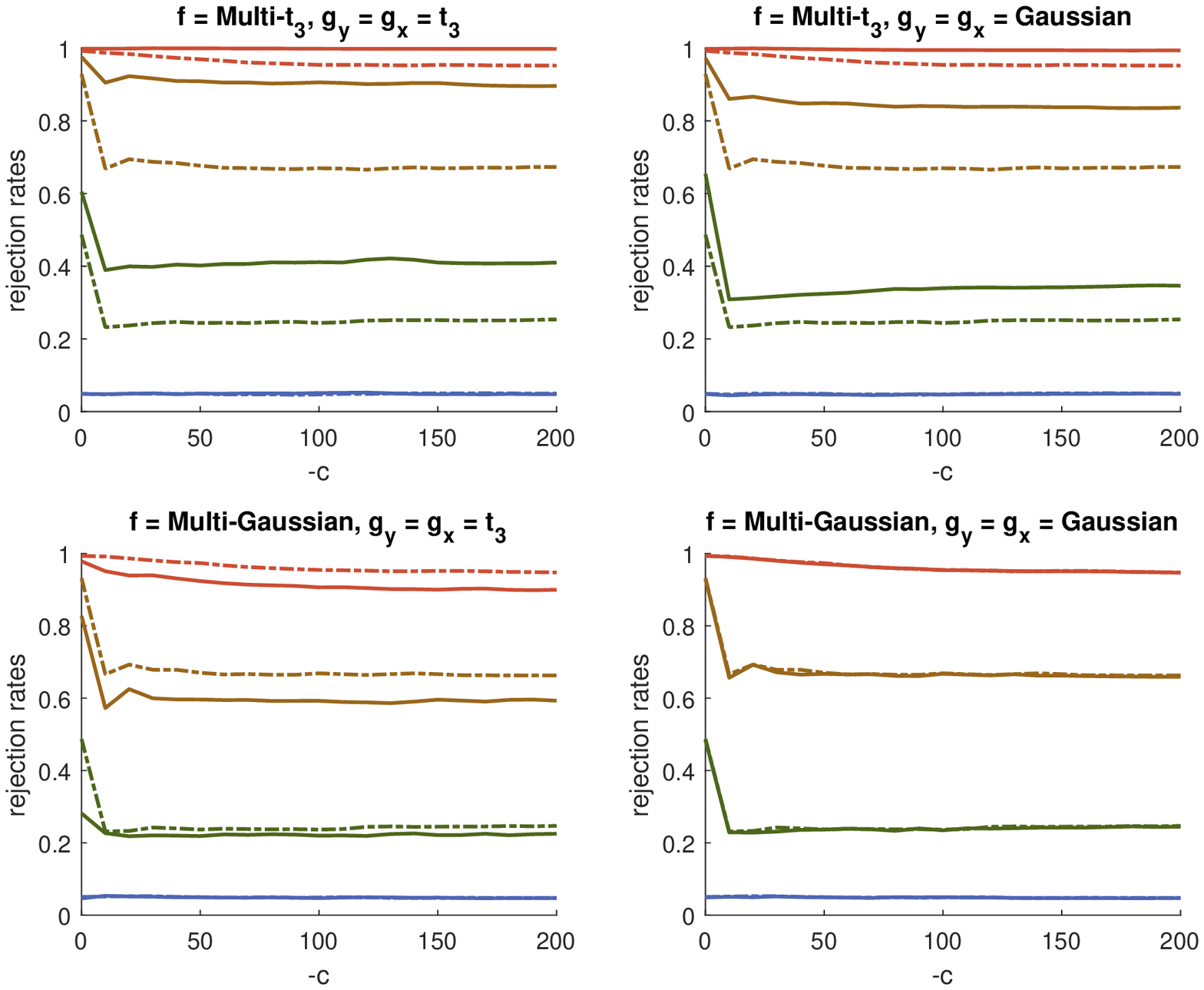}
\caption{Rejection rates of the WZ test (solid lines) and the EMW test (dashed lines) for different values of $\delta$ = 0, 1, 2, and 3, corresponding to lines in blue, green, brown, and red, respectively. For all cases, $\rho = -0.9$ and $T = 2,000$.}
\label{figure:rho09_t2000}
\end{figure}

\begin{figure}[!htb] 
\centering
\includegraphics[width=15cm]{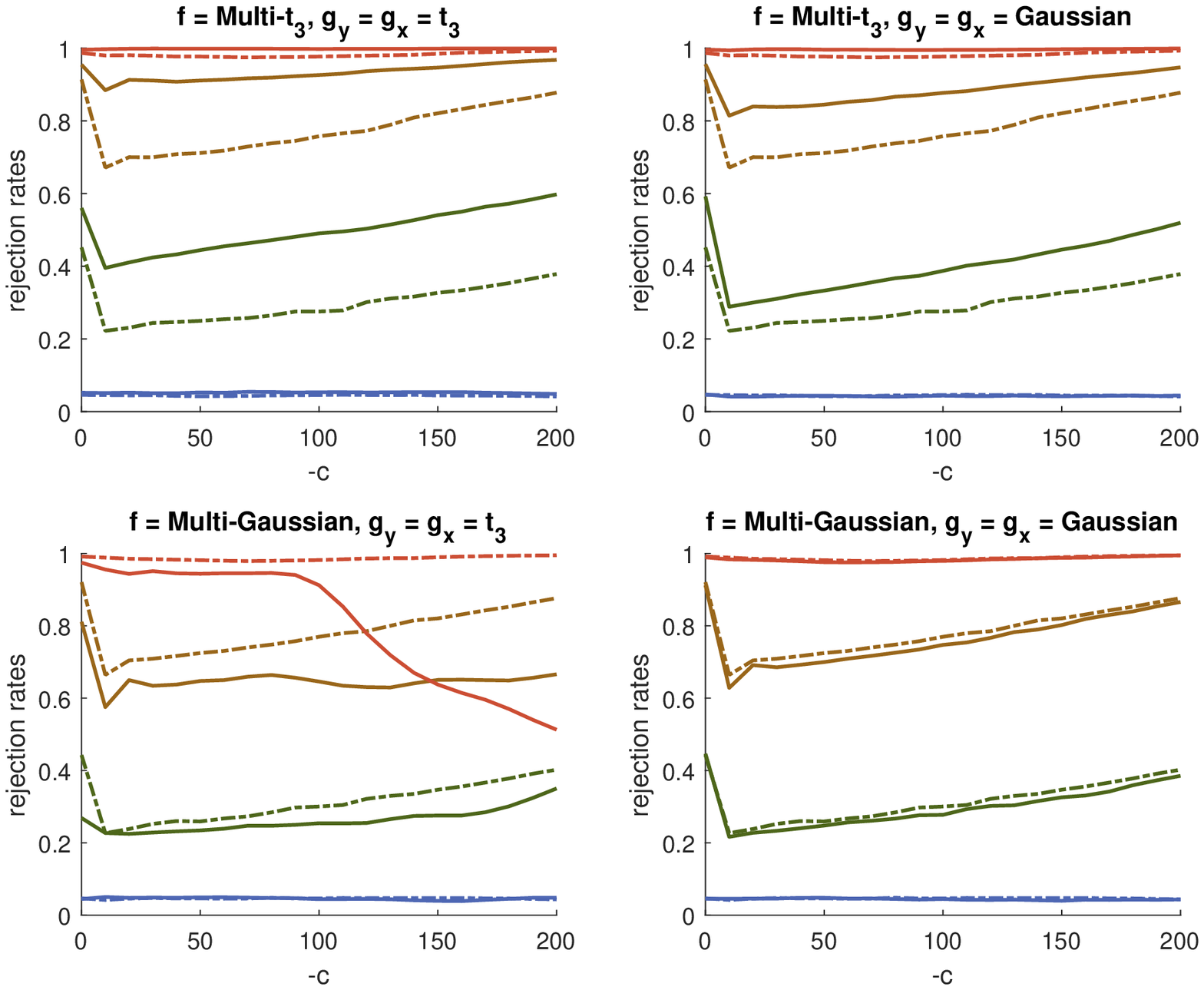}
\caption{Rejection rates of the WZ test (solid lines) and the EMW test (dashed lines) for different values of $\delta$ = 0, 1, 2, and 3, corresponding to lines in blue, green, brown, and red, respectively. For all cases, $\rho = -0.9$ and $T = 200$.}
\label{figure:rho09_t200}
\end{figure}

\begin{figure}[!htb] 
\hspace{-30mm}
\includegraphics[width=20cm]{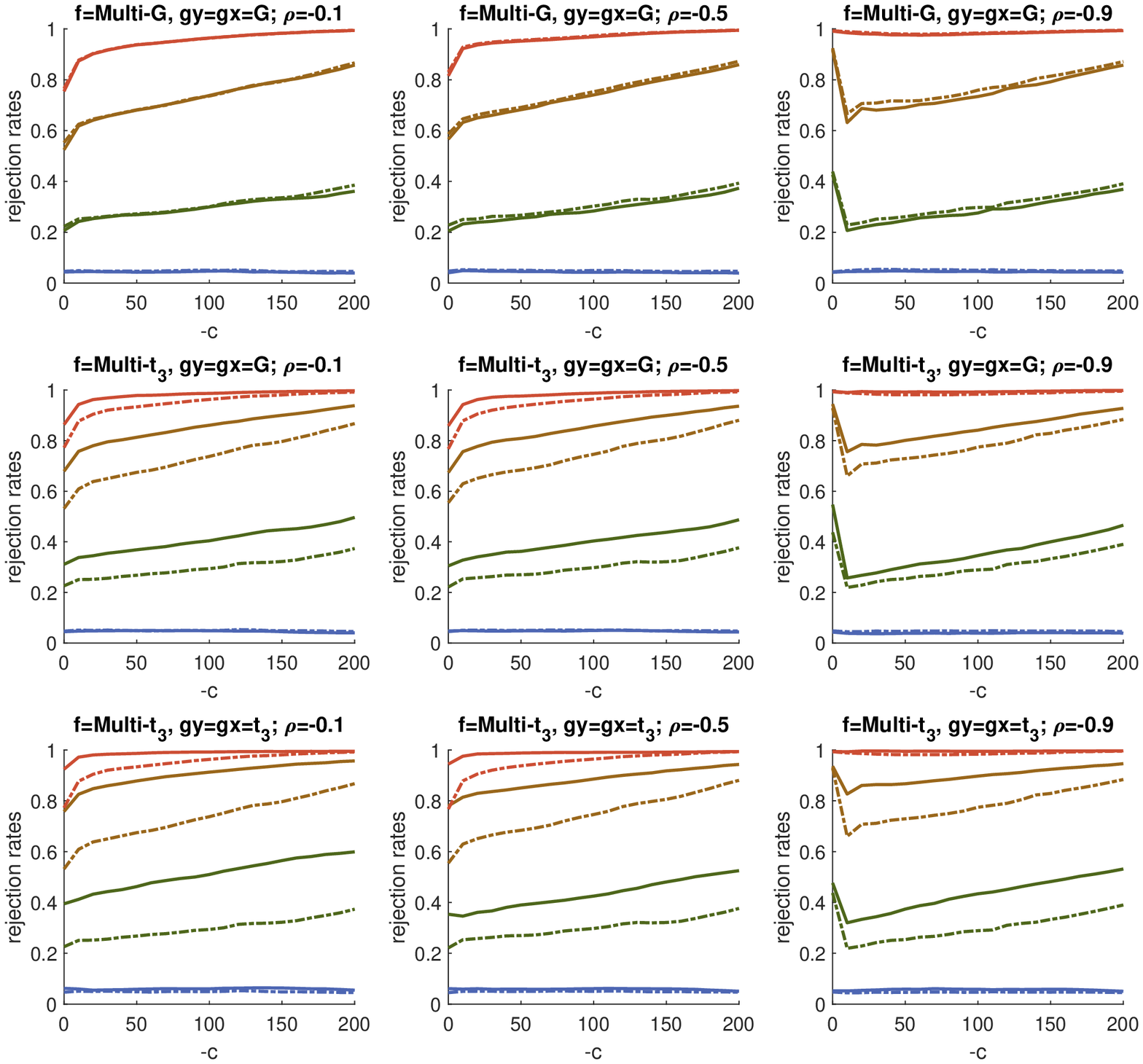}
\caption{Rejection rates of the WZ test (solid lines) and the EMW test (dashed lines) for different values of $\delta$ = 0, 1, 2, and 3, corresponding to lines in blue, green, brown, and red, respectively, under \emph{heteroskedasticity}. For all the four cases, $T = 200$.}
\label{figure:garch_t200}
\end{figure}

\begin{figure}[!htb] 
\hspace{-30mm}
\includegraphics[width=20cm]{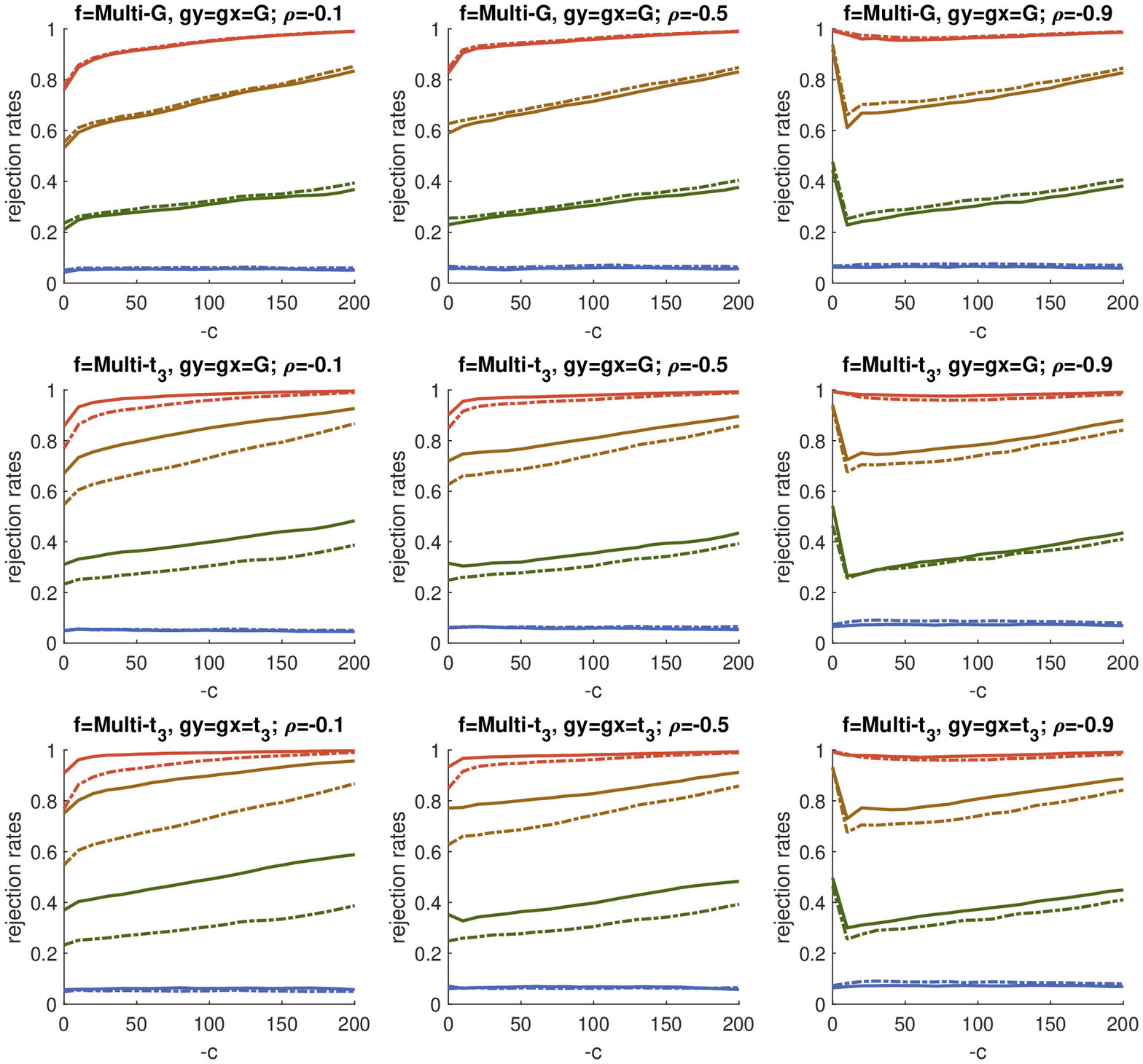}
\caption{Rejection rates of the WZ test (solid lines) and the EMW test (dashed lines) for different values of $\delta$ = 0, 1, 2, and 3, corresponding to lines in blue, green, brown, and red, respectively, under \emph{heteroskedasticity}. For all the four cases, $T = 200$.}
\label{figure:garchyx_t200}
\end{figure}

\end{document}